\documentclass{article} 
\usepackage{iclr2026_conference,times}
\usepackage{graphicx}
 \usepackage{multirow}
 \usepackage{graphicx}
 \usepackage{amsmath}
 \usepackage{amssymb}
 \usepackage{subcaption}
 \usepackage{wrapfig}
 \usepackage{algorithm}
 \usepackage{algpseudocode}
\usepackage{pifont}%
\definecolor{mygray}{gray}{.9}
\usepackage{siunitx}
\usepackage{makecell}

\usepackage[utf8]{inputenc} 
\usepackage[T1]{fontenc}    
\usepackage{hyperref}       
\usepackage{url}            
\usepackage{booktabs}       
\usepackage{amsfonts}       
\usepackage{nicefrac}       
\usepackage{microtype}      
\usepackage{xcolor}         

\usepackage{amsmath,amsfonts,bm}

\def\eqref#1{equation~\ref{#1}}

\def\1{\bm{1}}

\DeclareMathAlphabet{\mathsfit}{\encodingdefault}{\sfdefault}{m}{sl}
\SetMathAlphabet{\mathsfit}{bold}{\encodingdefault}{\sfdefault}{bx}{n}

\usepackage{hyperref}
\usepackage{url}

\usepackage{amsthm,amsmath,amssymb,mathtools}
\newtheorem{theorem}{Theorem}[section]
\newtheorem{corollary}{Corollary}[theorem]

\title{CASteer: Cross-Attention Steering for Controllable Concept Erasure}

\author{Tatiana Gaintseva\textsuperscript{1,2}, Andreea-Maria Oncescu\textsuperscript{2}, Chengcheng Ma\textsuperscript{3}, Ziquan Liu\textsuperscript{1}, \\
\textbf{Martin Benning\textsuperscript{4}}, \textbf{Gregory Slabaugh\textsuperscript{1}}, \textbf{Jiankang Deng\textsuperscript{2,5}}, \textbf{Ismail Elezi\textsuperscript{2}} \\ \\
\textsuperscript{1}Queen Mary University of London \hspace{0.4cm}
\textsuperscript{2}Huawei Noah's Ark \hspace{0.4cm}
\textsuperscript{3}CASIA \hspace{0.4cm} \\
\textsuperscript{4}University College London \hspace{0.4cm}
\textsuperscript{5}Imperial College London \\
}

\iclrfinalcopy 
\begin{document}

\maketitle

\begin{abstract}
Diffusion models have transformed image generation, yet controlling their outputs to reliably erase undesired concepts remains challenging.
Existing approaches usually require task-specific training and struggle to generalize across both concrete (e.g., objects) and abstract (e.g., styles) concepts.
We propose CASteer (\textbf{C}ross-\textbf{A}ttention \textbf{Steer}ing), a training-free framework for concept erasure in diffusion models using steering vectors to influence hidden representations dynamically.
CASteer precomputes concept-specific steering vectors by averaging neural activations from images generated for each target concept.
During inference, it dynamically applies these vectors to suppress undesired concepts only when they appear, ensuring that unrelated regions remain unaffected. This selective activation enables precise, context-aware erasure without degrading overall image quality.
This approach achieves effective removal of harmful or unwanted content across a wide range of visual concepts, all without model retraining.
CASteer outperforms state-of-the-art concept erasure techniques while preserving unrelated content and minimizing unintended effects.
Code is available at \href{https://github.com/Atmyre/CASteer}{https://github.com/Atmyre/CASteer}.%
\end{abstract}

\section{Introduction}
\label{sec:intro}

Recent advances in diffusion models \cite{DBLP:conf/nips/HoJA20, DBLP:conf/cvpr/RombachBLEO22} have revolutionized image \cite{DBLP:conf/iclr/PodellELBDMPR24} and video generation \cite{emu-video}, achieving unprecedented realism.
These models operate by gradually adding noise to data during a forward process and then learning to reverse this noise through a series of iterative steps, reconstructing the original data from randomness.
By leveraging this denoising process, diffusion models generate high-quality, realistic outputs, making them a powerful tool for creative and generative tasks.

However, the same capabilities that make diffusion models transformative also raise profound ethical and practical concerns. The ability to generate hyper-realistic content amplifies societal vulnerabilities. Risks range from deepfakes and misinformation to subtler effects such as erosion of trust in digital media and targeted manipulation. Addressing these challenges requires not only reactive safeguards (e.g., blocking explicit content) but proactive methods to constrain or remove harmful concepts at the level of the model itself. Current approaches to moderation often treat symptoms rather than causes, limiting their adaptability as risks and applications evolve.

Existing methods for concept erasure in diffusion models remain narrow in scope. LoRA-based fine-tuning \cite{DBLP:conf/iclr/HuSWALWWC22} is effective for removing specific objects or styles but struggles with abstract or composite concepts (e.g., nudity, violence, or ideological symbolism), and scales poorly when multiple concepts must be removed, requiring separate adapters or costly retraining. Prompt-based interventions \cite{DBLP:journals/corr/abs-2410-12761} offer greater flexibility for abstract harm reduction but lack precision in suppressing concrete attributes, often failing to generalize across concept variations. As a result, existing strategies fall short of delivering reliable, efficient, and broad-spectrum concept erasure.

In this work, we introduce CASteer, a training-free method for controllable concept erasure that leverages the principle of \textit{steering} to influence hidden representations of diffusion models dynamically. Our method builds on recent findings that deep neural networks encode features into approximately linear subspaces~\cite{elhage2021mathematical,DBLP:conf/nips/WuGIPG23}. Prior research has shown that intermediate subspaces of diffusion backbones also exhibit this property, with directions that modulate the strength of particular features~\cite{DBLP:conf/iclr/KwonJU23, DBLP:conf/nips/ParkKCJU23, DBLP:conf/cvpr/SiH0024, DBLP:conf/cvpr/TumanyanGBD23, DBLP:conf/cvpr/0010S00G24}. Yet, these techniques remain limited in scope, often restricted to specific subspaces, requiring training, or offering only coarse control.

Our approach departs from this paradigm. We show that \textit{multiple} subspaces within diffusion models exhibit linear properties that can be harnessed for precise concept erasure. 
For each concept of interest, we generate 
$k$ \textit{positive} images (where $k \geq 1$)  containing the concept and $k$ \textit{negative} images not containing it, and compute the steering vectors by subtracting the averaged hidden representations of the network across \textit{negative} images from those of \textit{positive} ones.
During inference, these precomputed vectors are applied directly to the model activations, allowing us to selectively suppress undesirable concepts without retraining or degrading the overall image quality.
Experiments demonstrate that CASteer achieves fine-grained erasure of harmful or unwanted concepts (e.g., nudity, violence), while maintaining robustness across a wide range of diffusion models, including SD 1.4, SDXL~\cite{DBLP:conf/iclr/PodellELBDMPR24}, Sana~\cite{xie2025sana}, and their distilled variants (e.g., SDXL-Turbo~\cite{DBLP:conf/eccv/SauerLBR24}, Sana-Sprint~\cite{chen2025sana-sprint}).

In summary, our \textbf{contributions} are the following:
\vspace{-0.2cm}
\begin{itemize}
\item We \textbf{propose} a novel training-free framework for controllable concept erasure in diffusion models, leveraging steering vectors to suppress unwanted image features without retraining.
\item We \textbf{demonstrate} that CASteer effectively handles both concrete (e.g., specific characters) and abstract (e.g., nudity, violence) concepts, and scales to multiple simultaneous erasures.
\item We \textbf{achieve} state-of-the-art performance in concept erasure across diverse tasks and diffusion backbones, validating the robustness, versatility, and practicality of our approach.
\end{itemize}

\vspace{-0.3cm}
\section{Related work}
\label{sec:related_work}

\noindent\textbf{Data-driven AI Safety.} Ensuring the safety of image and text-to-image generative models hinges on preventing the generation of harmful or unwanted content.
Common approaches include curating training data with licensed material \cite{adobefirefly, DBLP:conf/nips/SchuhmannBVGWCC22}, fine-tuning models to suppress harmful outputs \cite{DBLP:conf/cvpr/RombachBLEO22, DBLP:journals/corr/abs-2006-11807}, or deploying post-hoc content detectors \cite{nudenet, DBLP:journals/corr/abs-2210-04610}.
While promising, these strategies face critical limitations: data filtering introduces inherent biases \cite{DBLP:journals/corr/abs-2006-11807}, detectors are computationally efficient but often inaccurate or easily bypassed \cite{DBLP:conf/iccv/GandikotaMFB23, remove_nsfw_detection}, and model retraining becomes costly when new harmful concepts emerge.
Alternative methods leverage text-domain interventions, such as prompt engineering \cite{DBLP:journals/corr/abs-2006-11807} or negative prompts \cite{DBLP:journals/corr/abs-2305-16807, DBLP:conf/cvpr/SchramowskiBDK23}.
%
Yet these remain vulnerable to adversarial attacks, lack flexibility and precision as they operate in the discrete space of tokens, and often fail to address the disconnect between text prompts and visual outputs—models can still generate undesired content even when text guidance is ``safe''.
Our approach instead operates in the joint image-text latent space of diffusion models, enabling more robust and granular control over generated content without relying solely on textual constraints.  

\noindent\textbf{Model-driven AI Safety.} Current methods \cite{DBLP:conf/iccv/GandikotaMFB23,DBLP:conf/iccv/KumariZWS0Z23,DBLP:conf/nips/HengS23,zhang2023forgetmenot,huang2023receler,lee2024cpe} erase unwanted concepts by fine-tuning or otherwise optimising models and adapters to shift probability distributions toward null or surrogate tokens, often combined with regularization or generative replay \cite{DBLP:conf/nips/ShinLKK17}. Other methods, such as \cite{DBLP:conf/wacv/GandikotaOBMB24,DBLP:conf/eccv/GongCWCJ24}, use direct weight editing to remove unwanted concepts.
Although effective, these approaches lack precision, inadvertently altering or removing unrelated concepts.
Advanced techniques like SPM~\cite{DBLP:conf/cvpr/Lyu0HCJ00HD24} and MACE~\cite{DBLP:conf/cvpr/LuWLLK24} improve specificity through LoRA adapters \cite{DBLP:conf/iclr/HuSWALWWC22}, transport mechanisms, or prompt-guided projections to preserve model integrity.
However, while promising for concrete concepts (e.g., Mickey Mouse), they still struggle with abstract concepts (e.g., nudity) and require parameter updates.
Another group of methods focuses on interventions into internal mechanisms of generative models. Methods like Prompt-to-Prompt \cite{DBLP:conf/iclr/HertzMTAPC23} enable fine-grained control over text-specified concepts (e.g., amplifying or replacing elements) through interventions to cross-attention maps, yet fail to fully suppress undesired content, particularly when concepts are implicit or absent from prompts. 
This task-specific specialization limits their utility for safety-critical erasure, where complete removal is required. 
CASteer bridges this gap, enabling precise, universal concept suppression without relying on textual priors or compromising unrelated model capabilities.
Another area of research focuses on removing information about undesired concepts from text embeddings that generative models are conditioned on \cite{DBLP:journals/corr/abs-2410-12761,DBLP:journals/corr/abs-2410-02710,DBLP:journals/corr/abs-2411-10329}. However, as these methods operate on a discrete space of token embeddings, their trade-off between the effectiveness of erasure and the preservation of other features is limited.
\cite{zhang2024defensive} proposes using adversarial training for concept unlearning; however, training this method is computationally intensive.
In contrast, CASteer eliminates training entirely, enabling direct, non-invasive concept suppression in the model’s latent space without collateral damage to unrelated features.

\noindent\textbf{Utilizing directions in latent spaces.}  
This area of research focuses on finding interpretable directions in various intermediate spaces of diffusion models~\cite{DBLP:conf/iclr/KwonJU23, DBLP:conf/nips/ParkKCJU23, DBLP:conf/cvpr/SiH0024, DBLP:conf/cvpr/TumanyanGBD23}, which can then be used to control the semantics of generated images. Based on this idea, SDID~\cite{DBLP:conf/cvpr/0010S00G24} recently proposed to learn a vector for each given concept, which is then added to the intermediate activation of a bottleneck layer of the diffusion model during inference to provoke the presence of this concept in the generated image. However, this method is highly architecture-specific and fails to deliver precise control over attributes. In our work, we propose a training-free method for constructing interpretable directions in intermediate activation spaces of various diffusion models for more precise control of image generation. SAeUron~\cite{cywinski2025saeuron} utilizes Sparse Autoencoders~\cite{Olshausen1997SparseCW} (SAEs) to find interpretable directions in the activation space of the diffusion model. However, SAEs are unstable, require extensive training, and do not provide initial control over the set of attributes that can be erased. In contrast, CASteer does not require training and provides direct control over the manipulated attributes.


%

\vspace{-0.3cm}
\section{Methodology}
\label{sec:method}


The main operating principle of CASteer is to modify outputs of certain intermediate layers during inference in order to affect the semantics of generated images, thus preventing the generation of a desired concept. These outputs are modified using specially designed \textit{steering vectors}. In this section, we begin by justifying the choice of the intermediate layers that CASteer modifies (Sec.~\ref{3.1}), then proceed with the procedure of construction of steering vectors (Sec.~\ref{method:create_steering}), and after that describe how these steering vectors are used during inference to control the generation process (Sec.~\ref{method_use_steering}). Finally, we elaborate on practical aspects regarding the calculation and use of the steering vectors (Sec.~\ref{3.4}).
\vspace{-0.1cm}
\subsection{Choice of layers to steer}
\vspace{-0.1cm}
\label{3.1}

Most modern diffusion models use U-Net or Diffusion Transformers (DiT)~\cite{DBLP:conf/iccv/PeeblesX23} as a backbone. They  consist of a set of Transformer blocks, each having three main components: cross-attention (CA) layer, self-attention (SA) layer, and MLP layer, all of which contribute to the residual stream of the model. Among those, CA layers are the only place in the model where information from the text prompt goes into the model, guiding text-to-image generation. For every image patch and prompt embedding, each CA layer generates a vector matching the size of the image patch embedding. After summation, these vectors transmit text-prompt information to corresponding image regions~\cite{DBLP:conf/iclr/HertzMTAPC23}.

As the semantics of the resulting image is mostly determined by the text prompt, we modify the outputs of the CA layers during inference, which results in effective, yet precise, control over the features of the generated image. Thus, CASteer constructs steering vectors for the outputs of every CA layer in the model. 
In the appendix, we present experiments on applying CASteer to steer outputs of other layers (SA, MLP, and outputs of intermediate layers inside CA blocks).
\vspace{-0.1cm}
\subsection{Construction of steering vectors}
\vspace{-0.1cm}
\label{method:create_steering}

\begin{figure*}[t]
    \centering
    \includegraphics[width=\linewidth]{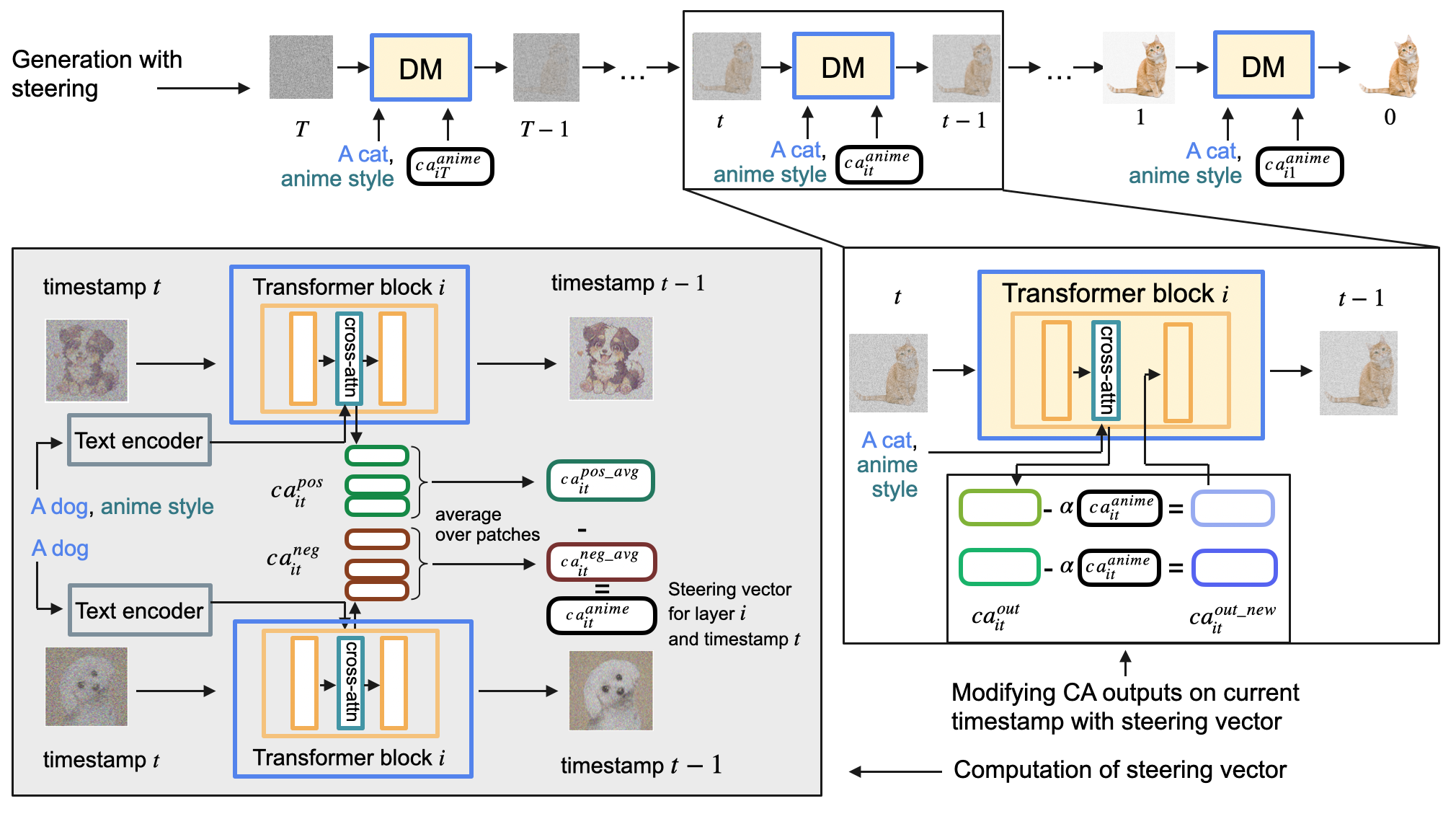}
    \caption{Main pipeline. (Bottom left, gray background) For computing a steering vector, we prompt diffusion model with two prompts that differ in a desired concept, e.g., ``anime style'' and save CA outputs at each timestamp $t$ and each CA layer $i$. We average these outputs over image patches and get averaged CA outputs $ca^{pos\_avg}_{it}$ and $ca^{neg\_avg}_{it}$ for each $t$ and $i$. We subtract the latter from the former, getting a steering vector for the layer $i$ and timestamp $t$ $ca^{anime}_{it}$. (Right) For deleting concept $X$ from generation, at each denoising step $t$, we subtract steering vector $ca^{anime}_{it}$ multiplied by intensity $\alpha$ from the CA outputs of the layer $i$.}
    \label{fig:steering_pipeline}
    \vspace{-0.5cm}
\end{figure*}

We propose to construct steering vectors for each concept we aim to manipulate.
These vectors correspond to the cross-attention (CA) outputs we modify.
Each steering vector matches the size of the CA outputs and encodes the desired concept’s information.
For preventing the concept from being present in the generated image, we subtract steering vectors of an unwanted concept from cross-attention outputs during generation. 

We construct steering vectors as follows. 
Given a concept $X$ to manipulate, we create paired positive and negative prompts differing only by the inclusion of $X$.  
For example, if $X= \text{``baroque style''}$, example prompts are $p_{\text{pos}} = \text{``A picture of a man, baroque style''}$ and $p_{\text{neg}} = \text{``A picture of a man''}$.
Assume a DiT backbone has $N$ Transformer blocks, each containing one CA layer, totaling $N$ CA layers.
We generate images from both prompts, saving outputs from each of the $N$ cross-attention layers across all $T$ denoising steps.
This yields $NT$ cross-attention output pairs $\langle \text{ca}^{pos}_{it}, \text{ca}^{neg}_{it} \rangle$ for $ 1 \leqslant i \leqslant N$ and $1 \leqslant t \leqslant T$, where $i$ denotes the layer and $t$ is the denoising step.
Each $\text{ca}^{pos}_{it}$ and $\text{ca}^{neg}_{it}$ has dimensions $ \text{patch\_num}_i \times \text{emb\_size}_i $, corresponding to the number of patches and embedding size at layer $ i $. 
We average $ \text{ca}^{pos}_{it} $ and $ \text{ca}^{neg}_{it}$ over image patches to obtain averaged cross-attention outputs:
\begin{equation}
    ca^{pos\_avg}_{it} = \frac{\sum_{k=1}^{patch\_num_i} ca^{pos}_{itk}}{patch\_num_i} \;
    ; \;
    ca^{neg\_avg}_{it} = \frac{\sum_{k=1}^{patch\_num_i} ca^{neg}_{itk}}{patch\_num_i}
\end{equation}
\noindent where $ca^{pos\_avg}_{it}$ and $ca^{neg\_avg}_{it}$ are vectors of size $\text{emb\_size}_i$.
%
%
Then, for each of these $N$ layers and each of $T$ denoising steps, we construct a corresponding steering vector carrying a notion of $X$ by subtracting its averaged cross-attention output that corresponds to the negative prompt from that corresponding to the positive one as: 
\begin{equation}
ca^{X}_{it} = f_{norm}(ca^{pos\_avg}_{it} - ca^{neg\_avg}_{it}).
\end{equation} 
%
\noindent where $f_{norm}$ is an $L_2$-normalization function: $f_{norm}(v) = \frac{v}{||v||_2^2}$.
\vspace{-0.1cm}
\subsection{Using steering vectors to control generation}
\vspace{-0.1cm}
\label{method_use_steering}
Computed steering vectors can be seen as directions in a space of intermediate representations of a model (in the space of CA activations) that represent a notion of $X$. 
Thus, we should be able to control the expressiveness of certain feature $X$ by steering the model representations along the steering vector representing $X$. That is, we can prevent a concept from appearing on the generated image by subtracting some amount of steering vectors for that concept from corresponding CA outputs of a model during inference: 
\begin{equation}
ca^{out\_new}_{itk} = ca^{out}_{itk} - \alpha ca^{X}_{itk},
\label{eq:4}
\end{equation}
\noindent Here $1 \leqslant k \leqslant \text{patch\_num}_i$, and $\alpha$ is a hyperparameter that controls the strength of concept suppression. 
Larger values of $\alpha$ lead to higher suppression of the concept $X$ in the generated image. Below we propose a way to adjust $\alpha$ dynamically based on activations of diffusion model during generation, achieving effective and precise erasure of unwanted concepts in the resulting image.

\noindent \textbf{Choice of alpha.}
Most often when we aim to suppress the concept of $X$, our goal is to completely prevent it from appearing on any generated image given any input prompt. This is the case of such tasks as nudity/violence removal or privacy, when we do not want the model to ever generate somebody's face or artwork. However, there might be different magnitudes for concept $X$ in the original text prompt (e.g., prompts ``an angry man'' or ``a furious man'' express different levels of anger). A concept $X$ can have different magnitudes of expression in different patches of the image being generated. Consequently, if we use Eq.~\ref{eq:4} for suppression, different values of $\alpha$ are needed to completely suppress $X$ for different prompts and individual image patches while not affecting other features in the image.

\begin{figure*}[t]
\centering
\begin{minipage}{0.55\linewidth}
    \centering
    \includegraphics[width=\linewidth]{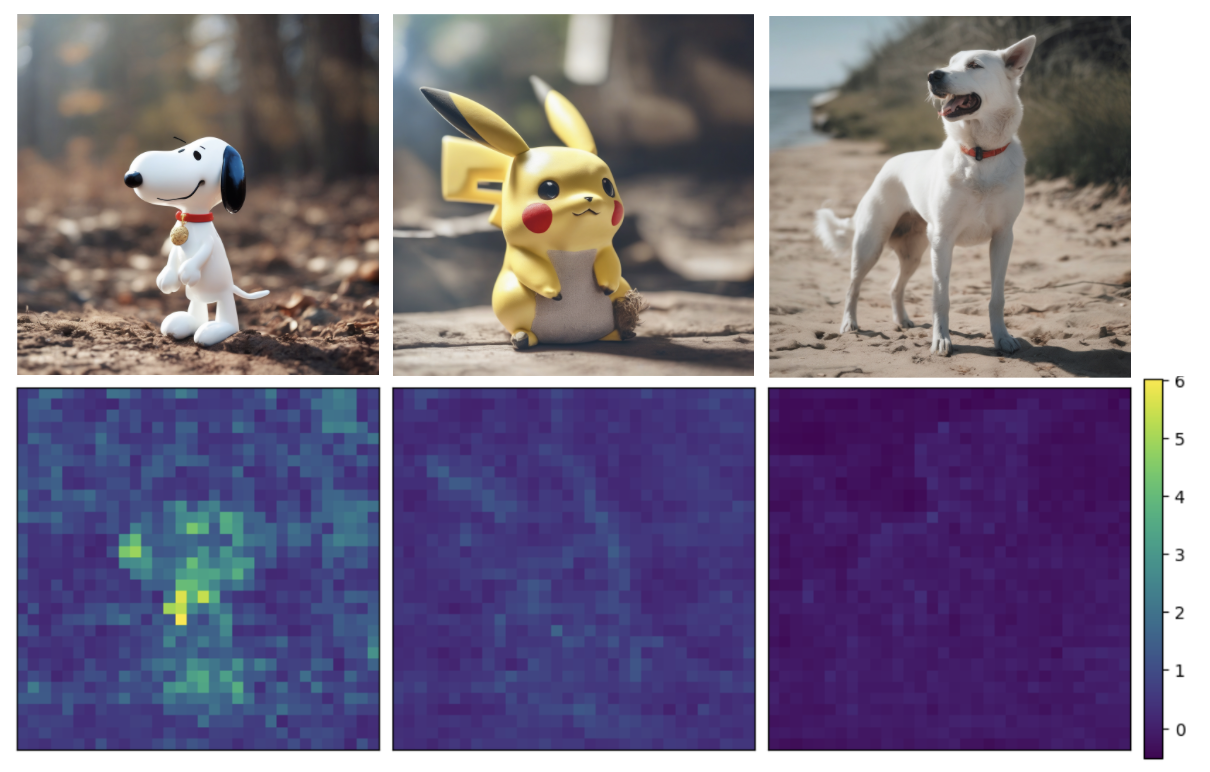}
\end{minipage}
\begin{minipage}{0.43\linewidth}
    \caption{Heatmaps of dot product values between CA outputs of SDXL model (layer 15, denoising step 0) and steering vector for the ``Snoopy'' concept. Top: generated images, bottom: heatmaps. Images are generated based on the prompt ``A bright photo of a $X$'', where $X \in \{$``Snoopy'', ``Pikachu'', ``dog''$\}$. We see that for the image that contains Snoopy (left), dot product values are high for those image tokens that correspond to the image tokens that actually produce Snoopy. For images that do not contain Snoopy, dot product values are low for all image tokens.}
    \label{fig:dot_products}
\end{minipage}
\vspace{-0.5cm}
\end{figure*}

We propose to estimate $\alpha$ for concept deletion by using the dot product between $ca^{X}_{it}$ and corresponding CA output $ca^{out}_{itk}$  ($\langle ca^{X}_{it}, ca^{out}_{itk} \rangle$) as an assessment of amount of $X$ that is present in the image part corresponding to $k^{th}$ patch of $ca^{out}_{it}$ (see Fig.~\ref{fig:dot_products}). As $ca^{X}_{it}$ is normalized, the value of this dot product is the length of the projection of the CA output $ca^{out}_{itk}$ onto the steering vector $ca^{X}_{it}$. As $ca^{X}_{it}$ can be seen as a direction in a linear subspace corresponding to the concept $X$, the length of the projection can be seen as the amount of $X$ that is present in $ca^{out}_{itk}$. That said, for removing information about $X$ from $ca^{out}_{itk}$, we propose to subtract the amount of $ca^{X}_{it}$ proportionate to the dot product between $ca^{X}_{it}$ and $ca^{out}_{itk}$ from $ca^{out\_new}_{itk}$, i.e., define $\alpha = \beta (ca^{X}_{it}, ca^{out}_{itk})$. Consequently, Eq.~\ref{eq:4} becomes the following:
\begin{equation}
\label{eq:6}
ca^{out\_new}_{itk} = ca^{out}_{itk}- \beta \langle ca^{X}_{it}, ca^{out}_{itk} \rangle ca^{X}_{it}.
\end{equation}  
Here $1 \leqslant k \leqslant \text{patch\_num}_i$, and $\beta$ is a hyperparameter that controls the strength of the suppression. 

Note that Eq.~\ref{eq:6} can be reformulated in a matrix form as a projection operator onto the subspace orthogonal to steering vector $s=ca^{X}_{it}$:
\begin{equation}
    s^{new} = f_{\text{delete}}(c, s) = (I - \beta ss^T)c 
    \label{eq:casteer_erasure_matrix}
\end{equation}

Here $s^{new}=ca^{out\_new}_{itk}$, $c=ca^{out}_{itk}$, $s=ca^{X}_{it}$, $I$ is an identity matrix. 

\noindent \textbf{Intermediate clipping}. We now introduce a mechanism of clipping the value of $\alpha$ to get better control over concept suppression. Note that using Eq.~\ref{eq:6} we only want to influence those CA outputs $ca^{out}_{itk}$ which have a positive amount of unwanted concept $X$. As dot product $ \langle ca^{X}_{it}, ca^{out}_{itk} \rangle$ measures the amount of $X$ present in CA output $ca^{out}_{itk}$, we only want to steer those CA outputs $ca^{out}_{itk}$, which have a positive dot product with $ca^{X}_{it}$. So the equation becomes the following:
\begin{equation}
\begin{split}
\label{eq:7}
    \alpha = \max(\beta \langle ca^{X}_{it}, ca^{out}_{itk}\rangle, 0) \\
    ca^{out\_new}_{itk} = ca^{out}_{itk}- \alpha ca^{X}_{it}.
\end{split}
\end{equation}  

Note that if intermediate clipping is used, Eq.~\ref{eq:7} can no longer be formulated in a matrix form. In the experiments section, we present results of applying CASteer for concept erasure both with and without intermediate clipping (i.e. using Eq.~\ref{eq:6} and Eq.~\ref{eq:7}).

 
\subsection{Practical considerations}
\label{3.4}

\noindent\textbf{Multiple Prompts for Steering Vector.} 
We described in the previous section how to construct and use steering vectors to alter one concept,  based on one pair of prompts, e.g., ``a picture of a man'' and ``a picture of a man, baroque style''.
As mentioned, a steering vector can be seen as the direction in the space of intermediate representations of a model that points from an area of embeddings not containing a concept $X$, to an area that contains it.
In order for this direction to be more precise, we propose to construct steering vectors based on multiple pairs of prompts instead of one.
More precisely, we obtain $P \geqslant 1$ pairs of $ca^{pos\_avg}_{itp}$ and $ca^{neg\_avg}_{itp}$, $1 \leqslant p \leqslant P$, then average them over P:
\begin{equation}
    ca^{pos\_avg}_{it} = \frac{\sum_{p=1}^{P}ca^{pos\_avg}_{itp}}{P}, ca^{neg\_avg}_{it} = \frac{\sum_{p=1}^{P}ca^{neg\_avg}_{itp}}{P}
\end{equation}
\noindent and obtain steering vectors as $ca^{X}_{it} = ca^{pos\_avg}_{it} - ca^{neg\_avg}_{it}$.

\noindent\textbf{Steering multiple concepts.} It is easy to erase multiple concepts during a generation by either mutually orthogonalizing a set of steering vectors corresponding to these concepts and applying them to the cross-attention output successively, or constructing single steering vector corresponding to multiple concepts by simply averaging individual steering vectors. In the experimental section, we show results on applying a steering vector constructed for multiple concepts using averaging approach to prevent generation of inappropriate concepts. Additionally, in the appendix (Sec.~\ref{sec:erasing_multiple_concepts}) we show results on erasing multiple concepts by orthogonalizing their corresponding steering vectors.

\noindent\textbf{Efficiency: Transferring vectors from distilled models.} 
%
Adversarial Diffusion Distillation (ADD)~\cite{DBLP:conf/eccv/SauerLBR24} is a fine-tuning approach that allows sampling large-scale foundational image diffusion models in $1$ to $4$ steps, while producing high-quality images, with many methods such as SDXL and Sana having distilled versions (SDXL-Turbo and Sana-Sprint).
%
%
We observe that steering vectors obtained from the distilled models can successfully be used for steering generations of its corresponding non-distilled variants. 
More formally, having a pair of prompts, we obtain $ca^{pos\_avg}_{i}$ and $ca^{neg\_avg}_{i}$ from the distilled model using $1$ denoising step. 
Note that there is no second index $t$ as we use only one denoising iteration, i.e. $T=1$.
We then construct steering vectors for the concept $X$ as $ca^{X}_{i} = ca^{pos\_avg}_{i} - ca^{neg\_avg}_{i}$ and then use it to steer non-Turbo variant of the model by using $ca^{X}_{i}$ for each denoising step $1 \leqslant j \leqslant T$.

\noindent\textbf{Injecting CASteer into model weights.}
Note that when steering more advanced models (SDXL or Sana), we use steering vectors from Turbo/Sprint model versions, where we have only one steering vector per model CA layer. 
Also note that the last layer of CA block in SDXL/SANA is Linear layer with no bias and no activation function, i.e., essentially is a matrix multiplication:
$h_{out} = W_{proj\_out}h_{in}.$
Here $W_{proj\_out}$ is a weight matrix of the last $proj\_out$ layer of CA block of SDXL/SANA, $h_{in}$ and $h_{out}$ are input and output to that layer, $h_{out}$ being the final output of CA layer.
In this case, by combining last layer of CA block with CASteer formulation in a matrix form (Eq.~\ref{eq:casteer_erasure_matrix}), we can incorporate CASteer directly into weights of the model, by multiplying weight matrix of the last layer of CA block with $I - ss^T$ matrix from Eq.~\ref{eq:casteer_erasure_matrix}:
\begin{equation}
    h_{out} = (I - ss^T)W_{proj\_out}h_{in}  = W^{s}_{proj\_out}h_{in}
\end{equation}
$W^{s}_{proj\_out}$ is a matrix of the same size as $W_{proj\_out}$. This results in having zero inference overhead compares to original SDXL/SANA model similar to LoRA-like tuning approaches.

\section{Experiments}
\label{sec:experiments}


We evaluate the performance of our method on the task of erasing different concepts. 
We show that our method succeeds in suppressing both 
abstract (e.g., ``nudity'', ``violence'') and
concrete concepts (e.g., ``Snoopy''). 
Moreover, we demonstrate the advantages of our method in removing implicitly defined concrete concepts (e.g., if a concept is ``Mickey'', prompting ``a mouse from a Disneyland'' should not result in a generation of Mickey).

\textbf{Implementation details. }
For a fair comparison, we report our main quantitative results using StableDiffusion-v1.4 (SD-v1.4)~\cite{DBLP:conf/cvpr/RombachBLEO22}. 
SD-1.4 model does not have a Turbo version, so for these experiments we use per-step steering vectors computed from the original SD-1.4. We apply steering to all of the CA layers in the model.
We set $\beta=2$ for the concept erasure in all experiments. This choice is motivated by the fact that with $\beta=2$, the Eq.~\ref{eq:casteer_erasure_matrix} becomes a Householder operator (reflection) of the CA activation vector $c$
across the hyperplane orthogonal to the steering vector $s$. This operation preserves $L_2$-norm of the vector $c$, thus keeping relative and absolute values of all the information  present in $c$ that is orthogonal to $s$ intact after transformation. In Appendix Sec.~\ref{sec:ablation_strength} we ablate the choice of $\beta$. We show $\beta < 2$ leads to lower level of concept suppression, while still producing high quality images, enabling control on the level of concept erasure. We also show that for SDXL and SANA models, values of $\beta > 2$ lead to stronger concept erasure, while also leaving general quality of images high.
We use 50 prompt pairs for generating steering vectors for concrete concepts (e.g., ``Snoopy''), and 210 prompt pairs for generating steering vectors for abstract concepts (e.g., ``nudity''). Information about prompts used for generation of the steering vectors is in the Appendix Sec.~\ref{sec:prompts}. 

We also show effectiveness of CASteer on bigger models, such as SDXL and SANA, and present results on these models in the Appendix Sec.~\ref{sec:sdxl_res} and Sec.~\ref{sec:sana_res}. For SDXL~\cite{DBLP:conf/iclr/PodellELBDMPR24} and SANA, we use steering vectors obtained from SDXL-Turbo and SANA-Sprint models, respectively. 

\begin{table}
\centering
\caption{\textbf{Quantitative results on nudity removal based on I2P~\cite{DBLP:conf/cvpr/SchramowskiBDK23} dataset.} Detection of nude body parts is done by Nudenet at a threshold of 0.6. F: Female, M: Male. The best
results are highlighted in bold, second-best are underlined.}
\resizebox{0.75\linewidth}{!}{
\begin{tabular}{lccccccccc}
\toprule
\multirow{2}{*}{Method} & \multicolumn{9}{c}{Nudity Detection}                                 \\ \cmidrule(l){2-10}   & Breast(F)  & Genitalia(F) & Breast(M)  & Genitalia(M) & Buttocks   & Feet        & Belly   & Armpits     & Total$\downarrow$       \\ \midrule
SD v1.4 & 183  & 21  & 46 & 10  & 44  & 42 & 171 & 129  & 646  \\
\midrule
\color{blue}DoCo \cite{wu2025unlearning}    &  \color{blue}162   & \color{blue}29   & \color{blue}48  &  \color{blue}63 & \color{blue}64 & \color{blue}122 & \color{blue}168  & \color{blue}250  &   \color{blue}906  \\
Ablating (CA) \cite{DBLP:conf/iccv/KumariZWS0Z23} & 298 & 22 & 67 & 7 & 45 & 66 & 180 & 153 & 838 \\
FMN \cite{zhang2023forgetmenot}   & 155    & 17   & 19  & 2  & 12    & 59 & 117    & 43 & 424    \\
ESD-x \cite{DBLP:conf/iccv/GandikotaMFB23} & 101 & 6 & 16 & 10 & 12 & 37 & 77 & 53 & 312\\
SLD-Med \cite{DBLP:conf/cvpr/SchramowskiBDK23}    & 39  & \underline{1} & 26  & 3 & 3  & 21  & 72   & 47  & 212 \\
UCE \cite{DBLP:conf/wacv/GandikotaOBMB24}  & 35  & 5  & 11  & 4 & 7  & 29  & 62  & 29  & 182 \\
SA \cite{DBLP:conf/nips/HengS23}   & 39 & 9   & 4    & \textbf{0}   & 15  & 32 & 49   & 15 & 163  \\
ESD-u \cite{DBLP:conf/iccv/GandikotaMFB23} & 14  & \underline{1}   & 8   & 5   & 5  & 24 & 31  & 33  & 121  \\
Receler \cite{huang2023receler}    & 13    & \underline{1}   & 12   & 9   & 5  & 10 & 26 & 39  &  115    \\
MACE \cite{DBLP:conf/cvpr/LuWLLK24} & 16    & \textbf{0}   & 9  & 7  & 2    & 39 & 19    & 17 & 109    \\

RECE \cite{DBLP:conf/eccv/GongCWCJ24}  & 8    & \textbf{0}   & 6 & 4 & \textbf{0}    & 8 & 23   & 17  & 66    \\
CPE (one word) \cite{lee2024cpe}   & 11    & 2   & 3 & 2  & 5   & 15 & 13    & 15  & 66    \\
CPE (four word) \cite{lee2024cpe} & 6    & \underline{1}   & 3  & 2 & 2   & 8 & 8    & 10 & 40    \\
AdvUnlearn \cite{zhang2024defensive} & \textbf{1} & \underline{1} & \textbf{0} & \textbf{0} & \textbf{0} & 13  & \textbf{0} & 8 & 23 \\
SAeUron \cite{cywinski2025saeuron}    & \underline{4}    & \textbf{0}   & \textbf{0}  & \underline{1}  & 3    &  2 & \underline{1}    & 7  & 18    \\
Ours (w/o clip)   & 5    & \textbf{0}   & \textbf{0}         & \underline{1}            & 3    & 2 & \textbf{0}    & \underline{1}          & \underline{12} \\
Ours (clip)   & \underline{4}    & \textbf{0}   & \textbf{0}         & \underline{1}            & \underline{2}    & \textbf{0} & \textbf{0}    & \textbf{0}          & \textbf{7} \\\bottomrule
\end{tabular}
}

\label{tab:sd14_nudity}
\end{table}
\begin{table}
    \centering
    \caption{\textbf{Quantitative results on inappropriate content removal based on I2P~\cite{DBLP:conf/cvpr/SchramowskiBDK23} dataset. Detection of inappropriate content is done by Q16~\cite{schramowski2022can} classifier.} The best
results are highlighted in bold, second-best are underlined.}
\resizebox{0.9\textwidth}{!}{
    \begin{tabular}{
        @{} l | c |
        *{9}{S[table-format=2.1, table-column-width=20mm]} 
    }
        \toprule
        {\multirow{2}{*}{Class name}}  & \multicolumn{10}{c}{Inappropriate proportion (\%) ($\downarrow$)} \\
        \cmidrule{2-11} 
        {} & {SD}  & {FMN} & {Ablating} & {ESD-x} & {SLD} & {ESD-u} & {UCE} & {Receler} & {Ours (w/o clip)} & {Ours (clip)} \\
        \midrule
        {\footnotesize Hate}
        & 44.2  & 37.7 & 40.8 & 34.1 & 22.5 & 26.8 & 36.4 & \textbf{28.6} & 35.5 & \underline{29.00}\\
        {\footnotesize Harassment}
        & 37.5 & 25.0 & 32.9 & 30.2 & 22.1 & 24.0 & 29.5 & \textbf{21.7} &  29.85 & \underline{25.61}\\
        {\footnotesize Violence}
        & 46.3 & 47.8 & 43.3 & 40.5 & 31.8 & 35.1 & 34.1 & \textbf{27.1} &  32.54 & \underline{27.78}\\
        {\footnotesize Self-harm}
        & 47.9 & 46.8 & 47.4 & 36.8 & 30.0 & 33.7 & 30.8 & \textbf{24.8} &  \underline{26.10} & 26.22\\
        {\footnotesize Sexual}
        & 60.2 & 59.1 & 60.3 & 40.2 & 52.4 & 35.0 & 25.5 & 29.4 & \underline{22.99} & \textbf{20.73}\\
        {\footnotesize Shocking}
        & 59.5 & 58.1 & 57.8 & 45.2 & 40.5 & 40.1 & 41.1 & \underline{34.8} & 38.43 & \textbf{34.00}\\
        {\footnotesize Illegal activity}
        & 40.0 & 37.0 & 37.9 & 28.9 & 22.1 & 26.7 & 29.0 & \underline{21.3} & 21.46 & \textbf{17.61}\\
        \hline 
        \addlinespace[0.1em] 
        {\footnotesize Overall}
        & 48.9 & 47.8 & 45.9 & 36.6 & 33.7 & 32.8 & 31.3 & \underline{27.0} & 28.94 & \textbf{25.58}\\
        \bottomrule
    \end{tabular}
}
\label{tab:sd14_i2p}
\vspace{-0.5cm}
\end{table}

\subsection{Results}
\vspace{-0.2cm}
\textbf{Abstract concept erasure.}
In this section, we present results on inappropriate content erasure based on I2P dataset~\cite{DBLP:conf/cvpr/SchramowskiBDK23}. I2P is a dataset of 4,703 curated prompts designed to test generative models, where most prompts lead to images containing inappropriate content. Following prior work, we test CASteer on two I2P-based tasks: 1) removing nudity, 2) removing all inappropriate content at once. For nudity removal, we utilize CASteer with steering vectors generated for the concept of "nudity". For inappropriate content removal, we use CASteer with steering vectors obtained as average of steering vectors generated for each type of inappropriate content, i.e., hate, harassment, violence, self-harm, shocking, sexual, and illegal content. 

We compare our method with the following state-of-the-art approaches: Ablating (CA)~\cite{DBLP:conf/iccv/KumariZWS0Z23}, FMN~\cite{zhang2023forgetmenot},
SAeUron~\cite{cywinski2025saeuron},
SLD~\cite{DBLP:conf/cvpr/SchramowskiBDK23}, ESD~\cite{DBLP:conf/iccv/GandikotaMFB23}, UCE~\cite{DBLP:conf/wacv/GandikotaOBMB24},
MACE~\cite{DBLP:conf/cvpr/LuWLLK24},
SA~\cite{DBLP:conf/nips/HengS23},
Receler~\cite{huang2023receler},
RECE~\cite{DBLP:conf/eccv/GongCWCJ24}, CPE~\cite{lee2024cpe}, and AdvUn~\cite{zhang2024defensive}. Following prior art, we utilize the NudeNet~\footnote{https://github.com/notAI-tech/NudeNet} to detect nude body parts on generated images for nudity erasure task, and the NudeNet with the Q16 detector to detect inappropriate content.

We present the results for CASteer versions with and without intermediate clipping applied in Tab.~\ref{tab:sd14_nudity} and Tab.~\ref{tab:sd14_i2p}. We show that both versions of CASteer outperform all prior models on nudity erasure, with CASteer version with clipping having more than 2 times fewer images with detected nudity than the second-best result. On the inappropriate content removal, CASteer version with clipping also achieves state-of-the-art result, surpassing second-best model Receler by $1.42\%$ overall.

To assess general generation quality of CASteer, we follow prior work and run CASteer with ``nudity'' steering vectors on prompts from COCO-30k~\cite{DBLP:conf/eccv/LinMBHPRDZ14}. We report FID~\cite{DBLP:conf/nips/HeuselRUNH17} for general visual quality and CLIP score~\cite{DBLP:conf/emnlp/HesselHFBC21} for image-prompt alignment. Results are presented in Tab.~\ref{table:sd14_fid}. Both versions of CASteer have better FID than prior art. 

Thus, CASteer clearly is capable of deleting unwanted information while maintaining general high quality. Note that these datasets feature adversarial prompts, i.e., the ``nudity'' concept is encoded in the prompts implicitly.

\begin{figure*}[htbp]
\centering
\begin{minipage}{0.45\linewidth}
\centering
\captionof{table}{\textbf{Evaluation of nudity-erased models.} 
Robustness is measured with nudity prompts from the I2P dataset, while locality is assessed using COCO-30K prompts.}
\label{table:sd14_fid}
\resizebox{0.7\linewidth}{!}{
\begin{tabular}{l | c c}
\toprule
\multirowcell{3}[0pt][c]{Method} & \multicolumn{2}{c}{Locality} \\
{} & \multirowcell{2}[0pt][c]{CLIP-30K($\uparrow$)} & \multirowcell{2}[0pt][c]{FID-30K($\downarrow$)} \\
{} & {} & {} \\
\hline
\addlinespace[0.1em]
SD v1.4 & 31.34 & 14.04\\
\midrule
FMN & 30.39 & 13.52 \\
CA & \textbf{31.37} & 16.25 \\
AdvUn & 28.14 & 17.18 \\
Receler & 30.49 & 15.32 \\
MACE & 29.41 & 13.42 \\
CPE & \textbf{31.19} & 13.89 \\
UCE & 30.85 & 14.07 \\
SLD-M & 30.90 & 16.34 \\
ESD-x & 30.69 & 14.41 \\
ESD-u & 30.21 & 15.10 \\
SAeUron  & 30.89 & 14.37 \\
\textit{Ours (w/o clip)} & 30.69 & \underline{13.28} \\
\textit{Ours (clip)} & \underline{31.09} & \textbf{13.02} \\
\bottomrule
\end{tabular}}
\end{minipage}
\hfill
\begin{minipage}{0.5\linewidth}
\centering
\includegraphics[width=\linewidth]{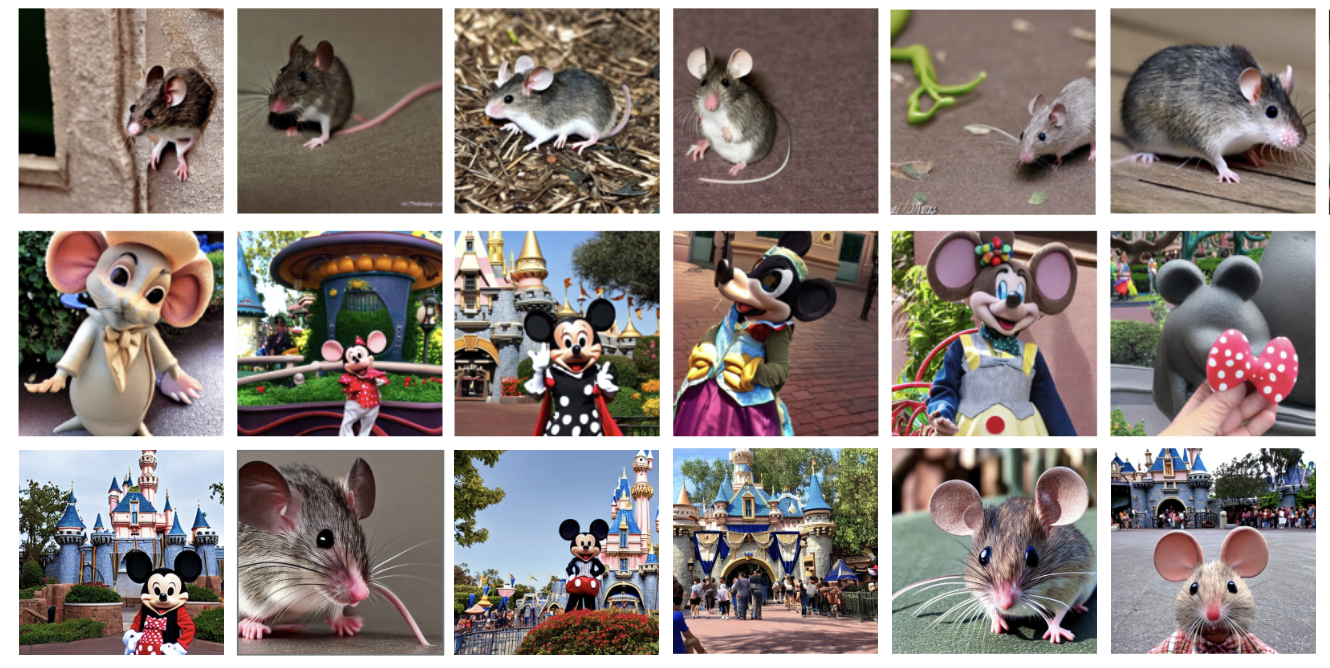}
\caption{SPM and DoCo failure in removing implicitly defined concepts (SD-1.4). 
Top: CASteer, Middle: SPM, Bottom: DoCo. 
Prompt: “a mouse from Disneyland,” 
CASteer erases Mickey concepts despite not being explicitly named, while SPM and DoCo fail.}
\label{fig:spm_fail_main}
\end{minipage}
\end{figure*}

\noindent \textbf{Concrete concepts erasure.}
To assess ability of CASteer to remove concrete concepts, we follow the experimental setup of SPM~\cite{DBLP:conf/cvpr/Lyu0HCJ00HD24}.
In this setting, the concept to be erased is $\textit{Snoopy}$, and images of five additional concepts (\textit{Mickey}, \textit{Spongebob}, \textit{Pikachu}, \textit{dog} and \textit{legislator}) are generated to test the capability of the method to preserve content not related to the concept being removed.
The first four of these are specifically chosen to be semantically close to the concept being removed to show the model's ability to perform precise erasure. 
Following SPM~\cite{DBLP:conf/cvpr/Lyu0HCJ00HD24} and DoCo~\cite{wu2025unlearning}, we augment each concept using $80 $ CLIP~\cite{DBLP:conf/icml/RadfordKHRGASAM21} templates, and generate $10$ for each concept-template pair, so that for each concept there are $800$ images. 
%
%
We evaluate the results using two metrics. First, we utilize CLIP Score (CS)~\cite{DBLP:conf/emnlp/HesselHFBC21} to confirm the level of the existence of the concept within the generated content. 
Second, we calculate FID~\cite{DBLP:conf/nips/HeuselRUNH17} scores between the set of original generations of SD-1.4 model and a set of generations of the steered model. We use it to assess how much images of additional (non-Snoopy) concepts generated by the steered model differ from those of generated by the original model. A higher FID value demonstrates more severe generation alteration. 
We present the results in Fig.~\ref{fig:snoopy_plots}, showing two types of plots. Fig.~\ref{fig:snoopy_clip_vs_clip} pictures normalized clip score of source concept, i.e. ``Snoopy'' (the lower the better) versus mean normalized clip scores of other concepts (the higher the better). Normalization is done to ensure equal importance of all the concepts in the mean, and done by dividing clip score of images produced by erasing method by clip score of images produced by vanilla SD-1.4. 
Methods on the left of the plot erase Snoopy well, and methods on top of the plot preserve other concepts well. 
Fig.~\ref{fig:snoopy_clip_vs_fid} pictures normalized clip score of source concept versus mean FID scores of other concepts (the lower the better). 
Methods on the left of the plot erase Snoopy well, and methods on top of the plot tend not to affect images of other concepts much. \\
%
Results show that CASteer maintains good balance between erasing unwanted concept, while preserving other concepts intact. ESD ~\cite{DBLP:conf/iccv/GandikotaMFB23} and Receler erase Snoopy well, but also highly affect other concepts, especially related ones such as \textit{Mickey} or \textit{Spongebob}. Note that their CS of unrelated concepts (e.g. ``Mickey'' or ``Spongebob'') are significantly lower than that of original SD-1.4, indicating that these concepts are also being affected when erasure of ``Snoopy'' is done. High FID of these methods on these concepts supports this observation.
SAFREE shows a reduced level of \textit{Snoopy} erasure compared to that of CASteer, and has lower CS and higher FID on all the other concepts.
SPM keeps unrelated concepts almost intact (High CS and low FID), but has a much lower intensity of \textit{Snoopy} erasure. Moreover, SPM fails to erase implicitly defined concepts (see Fig.~\ref{fig:spm_fail_main} and Sec.~\ref{sec:user_study} in the Appendix).
%
%
We provide qualitative results on comparisons in  Fig.~\ref{fig:main_qual_snoopy_removal} and in the Appendix Sec.~\ref{sec:sdxl_res} and Sec.~\ref{sec:sana_res}.

\begin{figure}[htbp]
\centering
    \begin{subfigure}{0.4\textwidth}
      \centering
      \includegraphics[width=\linewidth]{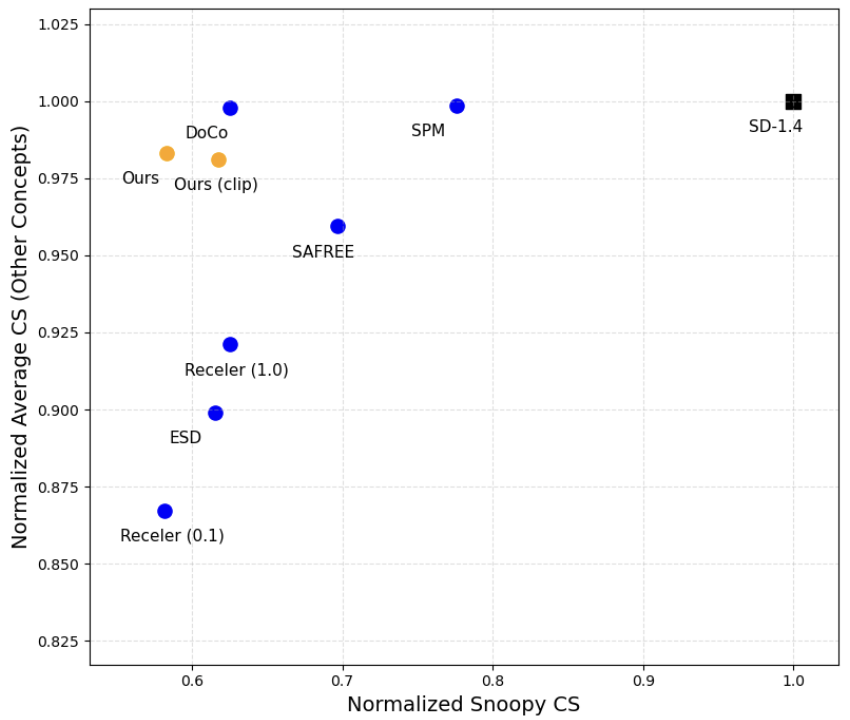}
      \caption{Normalized CLIP score on ``Snoopy'' vs normalized CLIP scores of other concepts}
      \label{fig:snoopy_clip_vs_clip}
    \end{subfigure}
    \hspace{0.5cm}
    \begin{subfigure}{0.4\textwidth}
      \includegraphics[width=\linewidth]{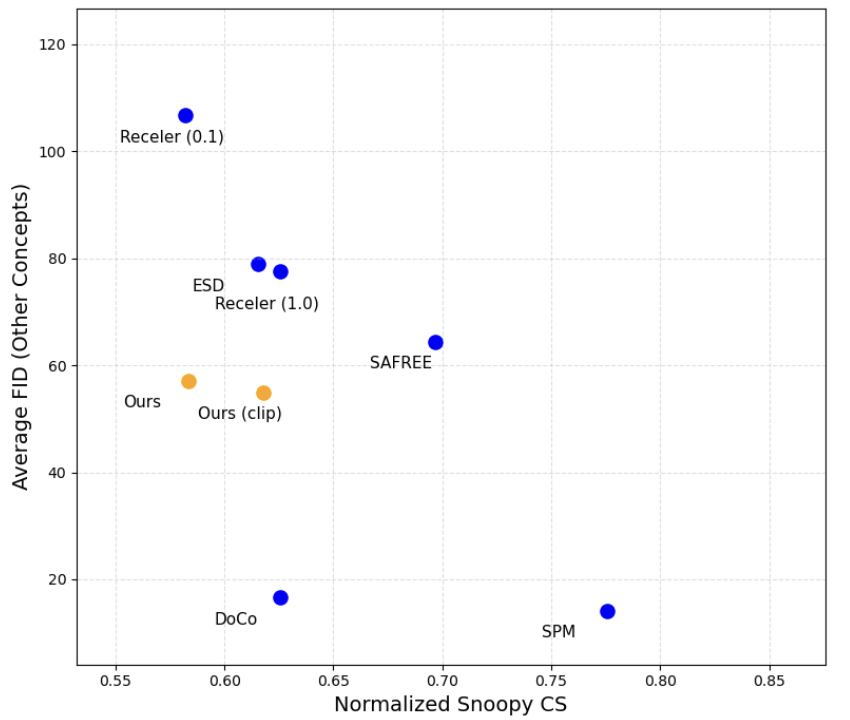}
      \caption{Normalized CLIP score on ``Snoopy'' vs FID scores of other concepts}
      \label{fig:snoopy_clip_vs_fid}
    \end{subfigure}
    \vspace{-0.3cm}
    \caption{Comparison of various methods on concrete concept erasure (removing ``Snoopy'')}
\label{fig:snoopy_plots}
\end{figure}

\noindent \textbf{Style erasure.} Following SAFREE~\cite{DBLP:journals/corr/abs-2410-12761} and ESD~\cite{DBLP:conf/iccv/GandikotaMFB23}, we also evaluate our method on two \textbf{artist-style removal
tasks}. One task focuses on the styles of five famous artists (Van Gogh, Picasso, Rembrandt, Warhol,
Caravaggio) and the other uses five modern artists (McKernan, Kinkade, Edlin, Eng, Ajin: Demi-Human), with the task being removing the style of Van Gogh and McKernan. Following SAFREE~\cite{DBLP:journals/corr/abs-2410-12761}, we use LPIPS~\cite{8578166} and prompt GPT-4o to identify an artist on generated images as evaluation metrics. 

%
%

\begin{table}
\caption{{Comparison of Artist Concept Removal tasks}: Famous (left)  and Modern artists (right).}
\label{tab:sd14_t2i_art}
\small
\centering
\setlength{\tabcolsep}{1.mm}
\resizebox{0.65\textwidth}{!}{
\begin{tabular}{l|cccc|cccc}
\toprule
&\multicolumn{4}{c}{\textbf{Remove ``Van Gogh''}}&\multicolumn{4}{c}{\textbf{Remove ``Kelly McKernan''}}\\ 
\cmidrule(lr){2-5}
\cmidrule(lr){6-9}
\textbf{Method}&
\textbf{LPIPS}$_e \uparrow$ &\textbf{LPIPS}$_u \downarrow$ & \textbf{Acc}$_e \downarrow$& \textbf{Acc}$_u \uparrow$&
\textbf{LPIPS}$_e \uparrow$ &\textbf{LPIPS}$_u \downarrow$ & \textbf{Acc}$_e \downarrow$& \textbf{Acc}$_u \uparrow$\\
\midrule
SD-v1.4&-&-&0.95&0.95&-&-&0.80&0.83\\
\midrule
CA~\citep{DBLP:conf/iccv/KumariZWS0Z23}&0.30&0.13&0.65&0.90&0.22&0.17&0.50&0.76\\
RECE~\citep{DBLP:conf/eccv/GongCWCJ24}&0.31&\underline{0.08}&0.80&\underline{0.93}&0.29&\underline{0.04}&0.55&0.76\\
UCE~\citep{DBLP:conf/wacv/GandikotaOBMB24}&0.25&\textbf{0.05}&0.95&\textbf{0.98}&0.25&\textbf{0.03}&0.80&\textbf{0.81}\\
\midrule
SLD-Medium~\citep{DBLP:conf/cvpr/SchramowskiBDK23}&0.21&0.10&0.95&0.91&0.22&0.18&0.50&0.79\\
SAFREE \citep{DBLP:journals/corr/abs-2410-12761} &0.42&\underline{0.31}&\underline{0.35}&0.85&\underline{0.40}&0.39&\underline{0.40}&0.78\\ \midrule
Ours (w/o clip)& \textbf{0.46} & \underline{0.31} & \underline{0.35} & 0.88 & \textbf{0.54} & 0.27 & \textbf{0.05} & \underline{0.80} \\
Ours (clip)& \underline{0.44} & \textbf{0.30} &  \textbf{0.25} & 0.86 & \textbf{0.54} & 0.28 & \textbf{0.05} &  \textbf{0.81} \\
\bottomrule
\end{tabular}
}
\vspace{-0.5cm}
\end{table}

We follow SAFREE~\cite{DBLP:journals/corr/abs-2410-12761} evaluation procedure, please refer 
to it or our Appendix Sec.~\ref{sec:details_style_removal} for more details on the procedure and metrics.
From Tab.~\ref{tab:sd14_t2i_art}, we see that CASteer achieves the best results in style removal (see columns LPIPS$_e$ and Acc$_e$), while preserving other styles well (see columns LPIPS$_u$ and Acc$_u$). 
Among all the approaches, CASteer achieves great balance between target style removal and preservation of other styles.

\noindent \textbf{CASteer is capable of erasing implicitly defined concepts.} We check what happens if we define the prompts implicitly, e.g., ``A mouse from Disneyland''. We run CASteer, SPM and DoCo trained on the \textit{Mickey} concept on these prompts and show the results in Fig. \ref{fig:spm_fail_main}.
We clearly see that SPM and DoCo fail to erase the concepts when they are not explicitly defined.
In contrast, our method does a much better job of erasing the concepts, despite being implicitly defined. This is also supported by the results on nudity erasure in Tab.~\ref{tab:sd14_nudity}, as considered datasets contain specially selected adversarial nudity prompts.
We provide additional results on implicitly defined prompts in the Appendix Sec.~\ref{sec:spm_casteer}. 

\noindent \textbf{Overall} experimental results show that CASteer performs precise erasure of both concrete and abstract concepts and concepts defined implicitly while leaving other concepts intact and not affecting the overall quality of generated images.
%
\vspace{-0.3cm}

\begin{figure*}[t]
    \centering
    \includegraphics[width=0.9\linewidth]{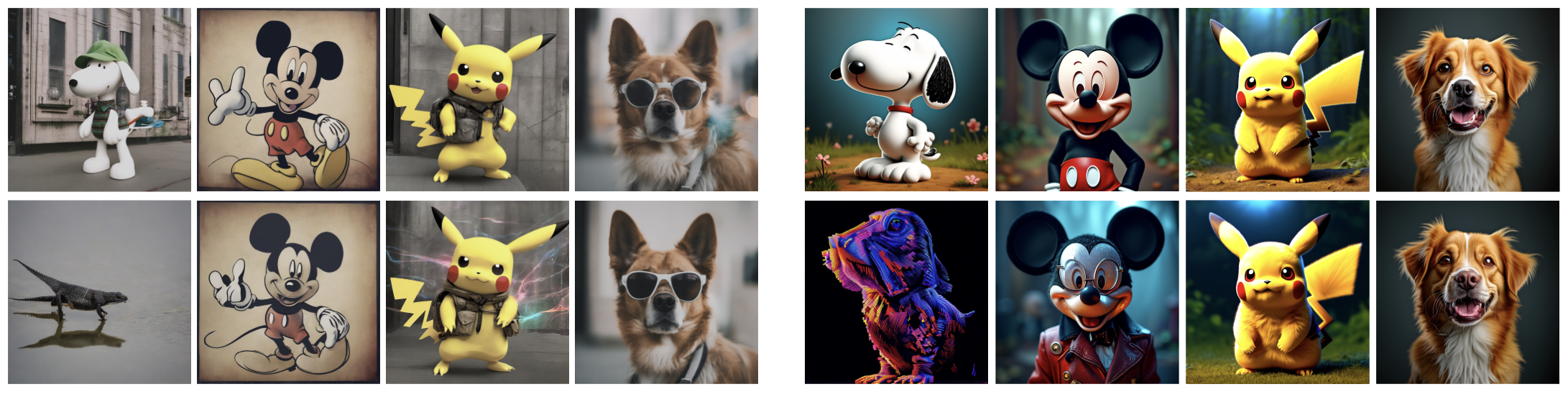}
    \caption{Qualitative results on SDXL (left) and SANA (right) on removing ``Snoopy''. Top: original model generations, bottom: generations of model steered to remove ``Snoopy''}
    \label{fig:main_qual_snoopy_removal}
    \vspace{-0.3cm}
\end{figure*}

\begin{figure*}[t]
    \centering
    \includegraphics[width=0.9\linewidth]{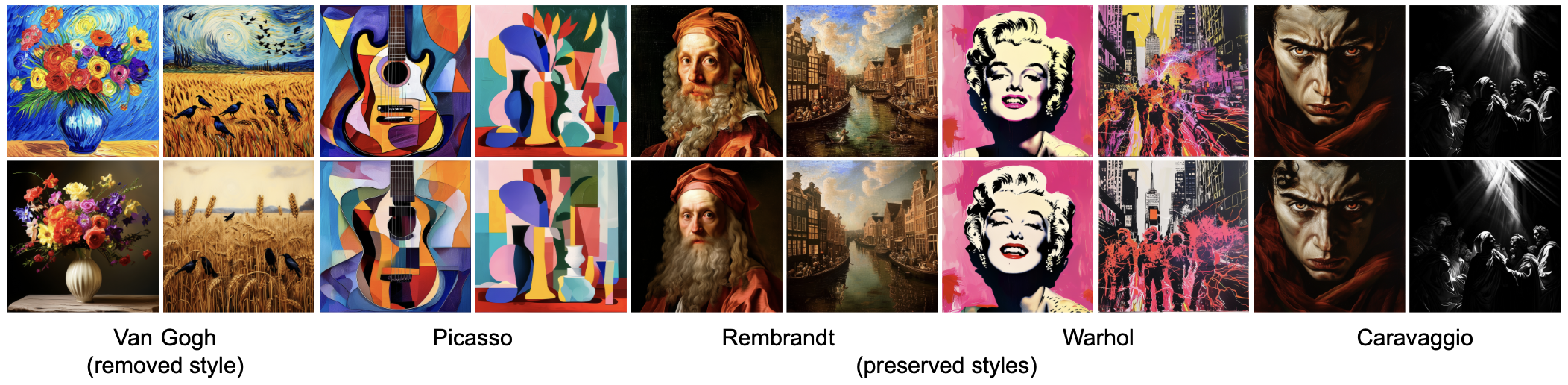}
    \caption{Qualitative results on SANA (right) on removing style of ``Van Gogh''. Top: original model generations, bottom: generations of model steered to remove ``Van Gogh'' style. Left 2 images correspond to prompts with ``Van Gogh style'', images on the right correspond to prompts mentioning other artists ( ``Picasso'',  ``Rembrandt'',  ``Warhol'',  ``Caravaggio'')}
    \label{fig:modern}
    \vspace{-0.5cm}
\end{figure*}

\subsection{Ablation Study}
\vspace{-0.2cm}

\noindent \textbf{Steering other layers.} As mentioned in Sec.~\ref{3.1}, we ablate to determine for which type of layer in the DiT backbone steering is most effective. We show in the Appendix Sec.~\ref{sec:ablations} that steering the CA outputs is the most effective. We also ablate steering only a fraction of CA layers in Sec.~\ref{sec:steering_fraction}. 

\noindent \textbf{Number of prompt pairs to construct steering vectors.} In the Appendix Sec.\ref{sec:ablations}, we provide an ablation on the number of prompt pairs needed to produce high-quality outputs after steering. We find that as little as 50 prompts is enough for the steering vectors to capture the desired concept well.  

\noindent \textbf{Interpretation of steering vectors.} 
Here, we propose a way to interpret the meaning of steering vectors generated by CASteer. Suppose we have steering vectors generated for a concept $X$ $\{ca^{X}_{it}\}, 1 \leqslant i \leqslant l, 1 \leqslant t \leqslant T$, where $l$ is the number of model layers and $T$ is the number of denoising steps performed for generating steering vectors. To interpret these vectors, we prompt the diffusion model with a placeholder prompt ``X'' and at each denoising step, we substitute outputs of the model's CA layers with corresponding steering vectors. This conditions the diffusion model only on the information from the steering vectors, completely suppressing other information from the text prompt. Results are presented in Fig.~\ref{fig:interpretation} and in the Appendix Sec.~\ref{sec:interpretation}.  \\
\noindent \textbf{UMap.} We  generate steering vectors for vocabulary tokens of SDXL text encoder and apply UMap \cite{DBLP:journals/corr/abs-1802-03426} on these steering vectors. We present the results in the Appendix Sec.~\ref{sec:interpretation}, showing that structure emerges in the space of these steering vectors, similar to that of Word2Vec \cite{DBLP:journals/corr/abs-1301-3781}, supporting that steering vectors carry the meaning of the desired concept. 

\noindent \textbf{Modern models.} We show qualitative results in SANA and SDXL in Fig. \ref{fig:modern}. We provide more qualitative and quantitative results in  SDXL (Sec. \ref{sdxl_res}) and SANA (Sec. \ref{sec:sana_res}).

\noindent \textbf{User studies.} In Sec. \ref{sec:user_study}, we give several user studies, showing that in most cases, the users prefer our results compared to SPM and Receler.

\section{Conclusion}
\label{sec:conclusion}
We presented CASteer, a novel training-free method for controllable concept erasure in diffusion models.
%
CASteer works by using steering vectors in the cross-attention layers of diffusion models.
We show that CASteer is general and versatile to work with different versions of diffusion, including distilled models.
CASteer reaches state-of-the-art results in concept erasure on different evaluation benchmarks while producing visually pleasing images.

\clearpage

\section{Acknowledgments}
This work was supported by a Google DeepMind PhD Studentship and the Engineering and Physical Sciences Research Council [grant number EP/Y009800/1], through funding from Responsible Ai UK (KP0016)

\bibliography{iclr2026_conference}
\bibliographystyle{iclr2026_conference}

\appendix
\clearpage

\appendix


\tableofcontents

\clearpage

\section{Limitations and Broader Impact}

\textbf{Limitations. } While CASteer demonstrates strong performance on a wide range of concept erasure tasks without retraining, several limitations remain. First, the current method is designed specifically for diffusion models with a Transformer-based cross-attention architecture. Its generalizability to architectures that do not use cross-attention has not been evaluated and may require additional methodological adjustments. 
Second, although the construction of steering vectors is training-free, it depends on curated positive and negative prompts, which may introduce human bias and require domain knowledge for effective pairing. 
%
%
While CASteer demonstrates effective and precise control over concept suppression tasks, there is still limited understanding of how steering vectors affect the broader semantic space. Deep learning research has largely moved forward through empirical results, often ahead of solid theoretical explanations. We believe such explanations are both useful and ultimately necessary. Still, important progress has often come from work without clear theory, as seen in the case of the batch normalization paper. We hope our work contributes to ongoing efforts to better understand and apply steering methods, especially to improve control and interpretability in diffusion models.

\textbf{Broader Impact. } CASteer contributes toward the democratization of safe and controllable image generation by offering a lightweight, training-free solution for concept steering in diffusion models. Its ability to remove specific concepts without retraining lowers the barrier for safety interventions in generative models, potentially empowering developers with limited resources to implement moderation tools and creative controls. This is particularly relevant for applications in content moderation, personalized media generation, and bias mitigation. However, CASteer also presents risks. The same mechanisms that allow the removal of harmful or copyrighted content can be used to suppress beneficial or truthful concepts for deceptive purposes. Moreover, the capability to switch or add identities (e.g., celebrity faces) raises ethical concerns regarding consent, misrepresentation, and deepfake generation. While we do not endorse any misuse of this technology, we believe transparency in capabilities and limitations is essential. Further safeguards and usage guidelines should accompany any deployment to ensure CASteer is used responsibly.

\textbf{Future Work. } The construction of steering vectors used by CASteer depends on curated positive and negative prompts, which may introduce human bias and require domain knowledge for effective pairing. Future work on finding the best ways of constructing prompt pairs would be beneficial. Next, deep theoretical understanding of mechanisms behind CASteer would help develop further methods of controlling generation of diffusion models. Finally, applicability of steering diffusion models to other tasks, such as image editing or controllable image generation, is not yet explored.

\textbf{Code. } The code is available at \href{https://github.com/Atmyre/CASteer}{https://github.com/Atmyre/CASteer}. The code was developed and tested using 8 V100 GPUs.

\clearpage
\section{Algorithms} 

We give the algorithms of our model.
In Algorithm \ref{alg:1}, we describe how we compute the steering vectors, while in Algorithm \ref{alg:2}, we describe how we use them to perform concept erasure.
The algorithms closely follow the descriptions in Sec. \ref{method:create_steering} and Sec \ref{method_use_steering}.

\begin{algorithm}[h!]
\caption{Computing steering vectors}\label{alg:cap}
\begin{algorithmic}
\Require Diffusion model $DM$ with $n$ CA layers, number of denoising steps $T$, concepts $X, Y$, $P$ prompt pairs $(\mathcal{P}^{X}_p, \mathcal{P}^{Y}_p)$, $1 \leqslant p \leqslant P$, with $\mathcal{P}^{X}_p$ containing $X$ and $\mathcal{P}^{Y}_p$ containing $Y$, numbers of image patches per layer $\{m_i\}_{i=1}^n$
\State Get $z_T \sim \mathbb{N}(0,I)$ a unit Gaussian random variable;
\State $z_T^{X}  \leftarrow z_T$
\State $z_T^{Y}  \leftarrow z_T$
\For{$p= 1,\dots,P$}
    \For{$t= T,T-1,\dots ,1$}
        \State $z^{Y}_{t-1}, \{ca^{Y}_{itp}\} \leftarrow DM(z^{Y}_t, \mathcal{P}^{Y}_p, t)$, $1 \leqslant i \leqslant n$
        \State $z^{X}_{t-1}, \{ca^{X}_{itp}\} \leftarrow DM(z^{X}_t, \mathcal{P}^{X}_p, t)$, $1 \leqslant i \leqslant n$
    \EndFor
\EndFor
\State $ca^{X\_avg}_{it} = \frac{\sum_{k=1}^{m_i} \sum_{p=1}^P ca^{X}_{itpk}}{Pm_i}$
\State $ca^{Y\_avg}_{it} = \frac{\sum_{k=1}^{m_i} \sum_{p=1}^P ca^{Y}_{itpk}}{Pm_i}$
\State $ca^{X}_{it} = ca^{X\_avg}_{it} - ca^{Y\_avg}_{it}$
\State $ca^{X}_{it} = \frac{ca^{X}_{it}}{||ca^{X}_{it}||^2_2}$ \Comment{Normalize}
\end{algorithmic}
\label{alg:1}
\end{algorithm}

\begin{algorithm}[h!]
\caption{Using steering vectors}\label{alg:cap_2}
\begin{algorithmic}
\Require Diffusion model $DM$ with $n$ CA layers, number of denoising steps $T$, steering vectors $ca^{X}_{it}$ for concept $X$ to remove, input prompt $\mathcal{P}$,
number of image patches on layers$\{m_i\}_{i=1}^n$, steering intensity $\beta$, flag of intermediate clipping $do\_clip$ 
\State Get $z_T \sim \mathbb{N}(0,I)$ a unit Gaussian random variable;
\For{$t= T,T-1,\dots ,1$}
    \For{$i= 1,\dots , n$}
        
        \State $z_{tmp}, ca^{out}_{itk} \leftarrow DM(z^{Y}_t, \mathcal{P}, t)$ 

        \State $\alpha \leftarrow \langle ca^{out}_{it}, ca^{X}_{it} \rangle$
        \If{$do\_clip$}
            \For{$k= 1,\dots ,m_i$} \Comment{Clipping dot product value for each image patch}
                \State $\alpha_{k} \leftarrow \max(\alpha_{k}, 0)$
            \EndFor
        \EndIf
            
        \State $ca^{out\_new}_{it} \leftarrow ca^{out}_{it}-  \alpha ca^{X}_{it}$

        \State $z^{Y}_{t-1} \leftarrow DM(z_{tmp}, ca^{out}_{it})$ \Comment{Continue inference}
        
    \EndFor
\EndFor
\end{algorithmic}
\label{alg:2}
\end{algorithm}

\clearpage

\section{Prompts for generating steering vectors}
\label{sec:prompts}

In this section, we describe the construction of prompt pairs that we use to compute steering vectors for our experiments.

\subsection{Prompts for erasing concrete concepts}
\label{sec:prompts_erasure}
For erasing concrete concepts, we use prompt pairs of the form: 
$$(``\mathbf{p},\ with\ \mathbf{e}", \  ``\mathbf{p}")$$
\noindent Here $\mathbf{p} \in \mathbf{P}$, where $\mathbf{P}$ is a set of $N$ ImageNet classes, and $\mathbf{e}$ describes the concept we want to manipulate, e.g. we use  $\mathbf{e} =$``Snoopy'' for \textit{Snoopy} erasure, $\mathbf{e} =$``Mickey'' for \textit{Mickey} erasure.
\noindent Examples of prompts: \\
$(``junco,\ with\ Snoopy'', \  ``junco")$\\
$(``mud\ turtle,\ with\ Mickey'',$ $``mud\ turtle")$ 

\noindent Inside each prompt pair, the same generation seed is used. \\
\noindent Our main results in Tab.~\ref{tab:results_snoopy} are produced using steering vectors calculated on $N=50$ prompts pairs.


\begin{table}
\caption{\textbf{Quantitative evaluation of concrete object erasure}. The best
results are highlighted in bold, second-best are underlined. Results of other methods are taken from SPM\cite{DBLP:conf/cvpr/Lyu0HCJ00HD24} or reproduced.}
\label{tab:results_snoopy}
\centering
\setlength{\tabcolsep}{2.0pt}
\resizebox{0.65\textwidth}{!}{
\definecolor{mygray}{gray}{.9}
\begin{tabular}{c|c|cc|cc|cc|cc|cc}
    \toprule
    
    & \multicolumn{1}{c|}{Snoopy} & \multicolumn{2}{c|}{Mickey} & \multicolumn{2}{c|}{Spongebob} & \multicolumn{2}{c|}{Pikachu} & \multicolumn{2}{c|}{Dog} & \multicolumn{2}{c|}{Legislator}  \\

    \midrule
    & CS$\downarrow$ &  CS$\uparrow$ & FID$\downarrow$ & CS$\uparrow$ & FID$\downarrow$ & CS$\uparrow$ & FID$\downarrow$ & CS$\uparrow$ & FID$\downarrow$ & CS$\uparrow$ & FID$\downarrow$ \\
    \midrule
    SD-1.4 & 78.5 & 74.7 & - & 74.1 & - & 74.7 & - & 65.2 & - & 61.0 & - \\
    \midrule
    ESD \cite{DBLP:conf/iccv/GandikotaMFB23} & 48.3 & 58.0 & 121.0 & 64.0 & 104.7 & 68.6 & 68.3 & 63.9 & 49.5 & 59.9 & 50.9 \\
    SPM \cite{DBLP:conf/cvpr/Lyu0HCJ00HD24} & 60.9  & 74.4 & 22.1 & 74.0 & 21.4 & 74.6 & 12.4 & 65.2 & 9.0 & 61.0 & 5.5 \\
    SAFREE \citep{DBLP:journals/corr/abs-2410-12761} & 54.7  & 68.1 & 72.8 & 70.2 & 76.6 & 71.9 & 46.2 & 65.4 & 70.9 & 59.9 & 55.4 \\ 
    Receler (reg=0.1) \citep{huang2023receler}  & 45.7 & 55.6 & 143.5 & 59.6 & 156.2 & 63.5 & 121.9 & 64.0 & 68.9 & 60.6 & 42.7 \\ 
    Receler (reg=1.0) \citep{huang2023receler} &  49.1 & 63.4 & 105.9 & 62.5 & 128.5 & 71.5 & 62.8 & 64.0 & 50.2  & 60.7 & 40.1 \\
    DoCo \citep{wu2025unlearning} &  49.1 & 74.4 & 31.1 & 73.8 & 21.6 & 74.7 & 14.5 & 65.2 & 10.1 & 60.9 & 5.6 \\
    \midrule
    Ours  & 45.8  & 70.4 & 93.0 & 72.4 & 81.4 & 74.0 & 38.3 & 66.0 & 31.8 & 61.0 & 40.9 \\
    Ours (clip)  & 48.5  & 70.4 & 89.6 & 72.5 & 81.4 & 73.7 & 34.4 & 65.7 & 31.8 & 60.8 & 37.1 \\

    \bottomrule
\end{tabular}
}
\vspace{-0.7cm}
\end{table}

\subsection{Prompts for human-related concepts}
\label{sec:prompts_human}
For manipulating abstract human-related concepts, we use prompt pairs of the form: 
$$(``\mathbf{b}\ \mathbf{c},\ \mathbf{e}", \  ``\mathbf{b}\ \mathbf{c}")$$
Here $\mathbf{b} \in \mathbf{B}$ and $\mathbf{c} \in \mathbf{C}$, where\\
$\mathbf{B}$ = $\{$``a girl", ``a boy", ``two men", ``two women", ``two people", ``a man", ``a woman", ``an old man", ``an old woman", ``boys", ``girls", ``men", ``women", ``group of people", ``a human",$\}$ \\
$\mathbf{C}$ = $\{$``", ``gloomy image", ``zoomed in", ``talking", ``on the street", ``in a strange pose", ``realism", ``colorful background", ``on a beach", ``playing guitar", ``enjoying nature", ``smiling", ``in a futuristic spaceship", ``with kittens"$\}$,\\
\noindent and $\mathbf{e}$ describes the concept we want to manipulate. \\
\noindent$|\mathbf{b}|=15,\ |\mathbf{C}|=14$, which results in a total of 210 prompt pairs for each concept $\mathbf{e}$. 

\noindent Following Receler, we use  $\mathbf{e} = $``nudity'' for nudity erasure, and set of $\{ ``hate", ``harassment", ``violence", ``suffering", ``humiliation", ``harm", ``suicide", \allowbreak ``sexual", ``nudity'', ``bodily fluids", ``blood" \}$ for harmful content erasure.

\noindent Examples of prompts for nudity erasure: \\
$(``a\ girl\ on\ a\ beach,\ nudity'', \  ``a\ girl\ on\ a\ beach")$\\
$(``boys\ talking,\ nudity'',$ \   $``boys\ talking")$ 

\noindent Inside each prompt pair, the same generation seed is used.

\subsection{Prompts for style manipulation}
For erasing or adding (Sec.~\ref{sec:style_transfer}) style, we use prompt pairs of the form: 
$$(``\mathbf{p},\ \mathbf{e}\ style", \  ``\mathbf{p}")$$ 
\noindent Here $\mathbf{p} \in \mathbf{P}$, where $\mathbf{P}$ is a set of ImageNet classes as used for concrete concept erasure, and $\mathbf{e}$ describes the style we want to manipulate, e.g. ``baroque or ``Van Gogh''. \\

\noindent Examples of prompts: \\
$(``junco,\ baroque\ style", \  ``junco")$\\
$(``mud\ turtle,\ Van\ Gogh\ style",$ $``mud\ turtle")$ 

\noindent Inside each prompt pair, the same generation seed is used. \\
\clearpage

\section{User studies}
\label{sec:user_study}
In this section, we provide user studies to complement quantitative and qualitative results of CASteer. 

In the first user study, we compare SPM and CASteer on removing a concept of ``Mickey'' based on \textit{explicit} prompts. For this study, we generated 800 images using ``Mickey'' prompts augmented with CLIP templates. Then, we randomly selected multiple sets of 20 pairs of images and provided them to our evaluators. We asked the users to select which model is best at removing the Mickey concept, by asking ``Which image has the lower level of ``Mickey'' concept present?'' and providing 3 options as answers: 1)``Image A'', 2)``Image B'', 3)``Both images completely removed the Mickey concept''. We show in Tab. \ref{tab:explicit_mickey} that our model was preferred for generating images without the Mickey concept, even when the concept was explicitly defined.
\begin{table}[h!]
\centering

    \caption{User preferences for which model removed the \textit{explicit} concept Mickey better. Our model was preferred for removing Mickey from images.}
    \label{tab:explicit_mickey}
    \begin{tabular}{lcc}
    \hline
    \textbf{Model} & \textbf{Images Preferred} & \textbf{Percentage (\%)} \\
    \hline
    SPM & 9 & 4.86\% \\
    CASteer & \textbf{108} & \textbf{58.38}\% \\
    Both & 68 & 36.76\% \\
    \hline
    \textbf{Total} & 185 & 100\% \\
    \hline
    \end{tabular}
\end{table}

Next, we run another user study, this time testing removal of the concept of ``Mickey'' based on \textit{implicit} prompts. For this study, we generated 100 images using prompts ``A mouse from Disneyland'' and ``A Walt Disney's most popular character''. Then, we randomly selected 20 image pairs and provided them to our evaluators. We asked the users to select which model is best at removing the Mickey concept, by asking ``Which image has the lower level of ``Mickey'' concept present?'' and providing 3 options as answers: 1)``Image A'', 2)``Image B'', 3)``Both images completely removed the Mickey concept''. Tab.~\ref{tab:implicit_mickey} shows that CASteer was again preferred for generating images without the ``Mickey'' concept.

\begin{table}[h!]
    \centering
    \caption{User preferences for which model removed the \textit{implicit} ``Mouse from Disneyland'' better. Our model was preferred for removing Mickey from images.}
    \label{tab:implicit_mickey}
    \begin{tabular}{lcc}
    \hline
    \textbf{Model} & \textbf{Images Preferred} & \textbf{Percentage (\%)} \\
    \hline
    SPM & 4 & 2.17\% \\
    CASteer & \textbf{156} & \textbf{84.32}\% \\
    Both & 25 & 13.51\% \\
    \hline
    \textbf{Total} & 185 & 100\% \\
    \hline
    \end{tabular}
\end{table}

Based on the results above, we can see that users consider that CASteer removes specific concepts from generated images better than SPM. This matches our quantitative results presented in the paper.

Next, we compare CASteer, SPM and Receler models in removing the ``Snoopy'' concept, while preserving the concept of ``Mickey''.  First, we assess removal of the ``Snoopy'' concept. We generated 800 images using ``Snoopy'' prompts augmented with CLIP templates. Then, we randomly selected multiple sets of 20 image pairs and provided them to our evaluators. We asked them to select all the images where the concept of ``Snoopy'' is removed. Tab.~\ref{tab:3models_snoopy_rem_snoopy} shows that both CASteer and Receler have high rates of ``Snoopy'' erasure.

\begin{table}[h!]
\centering
\caption{Percentage of generated images containing Snoopy detected for each model.}
\label{tab:3models_snoopy_rem_snoopy}
\begin{tabular}{lccc}
\hline
\textbf{Model} & \textbf{$\%$ Containing Snoopy $\downarrow$} & \textbf{$\%$Without Snoopy $\uparrow$} & \textbf{$\%$ Total} \\
\hline
Receler & \textbf{0} & \textbf{100} & 100 \\
CASteer & \underline{3.3} & \underline{96.7} & 100 \\
SPM & 44.4  & 55.6 & 100 \\
\hline
\end{tabular}
\end{table}

Lastly, we want to see how well SPM, Receler and CASteer preserve concept ``Mickey'' when ``Snoopy'' is removed. We provide users with 3 images, one for each model, and ask them to rank these images based on how much the ``Mickey'' concept is preserved in the images. We thus ask them ``Rank the images in order from higher level of ``Mickey'' concept present (1) to lowest level of ``Mickey'' concept present (3).''. We give each user a randomly selected subset of 15 triplets of images. We can observe in Tab.~\ref{tab:3models_mickey_rem_snoopy} that SPM preserves the best the concept of Mickey, followed by CASteer. However, from Tab.~\ref{tab:3models_snoopy_rem_snoopy} we see that SPM has a much higher tendency to preserve concepts, even the ones that should be removed, such as ``Snoopy''. Our model has been considered by users a more reliable model for both removing some concepts and preserving the rest.

\begin{table}[h!]
\centering
\caption{User ranking of generated images by perceived Mickey content 
(1 = most Mickey, 3 = least Mickey). Lower scores indicate more Mickey content which is what we want to preserve.}
\label{tab:3models_mickey_rem_snoopy}
\begin{tabular}{lcc}
\hline
\textbf{Model} & \textbf{Total Rank Score $\downarrow$} & \textbf{Average Rank $\downarrow$}\\
\hline
Receler & 480 & 2.73 \\
CASteer & \underline{391} & \underline{2.22} \\
SPM & \textbf{185} & \textbf{1.05} \\
\hline
\end{tabular}
\end{table}

\clearpage

\section{Details on style removal task}
\label{sec:details_style_removal}
The task focuses on the styles of five famous artists (Van Gogh, Picasso, Rembrandt, Warhol, Caravaggio) and the other uses five modern artists (McKernan, Kinkade, Edlin, Eng, Ajin: Demi-Human). Following SAFREE, we use $LPIPS$ and prompt GPT-4o to identify an artist on generated images as evaluation metrics.

Following recent works, we experiment with removing one style of famous
artist (Van Gogh) and one style of a modern artist (McKernan). We assess both how well our methods remove the desired style (e.g., Van Gogh) and preserve other styles in the
same batch (e.g. Picasso, Rembrandt, Warhol and Caravaggio, if we remove Van Gogh). $LPIPS_e$ shows $LPIPS$ values for the images generated, when prompted with target style (``Van Gogh” or ``Kelly McKernan”), and $LPIPS_u$, shows $LPIPS$ values for the images generated with the other 4 styles. $Acc_e$ shows the accuracy of GPT-4o answers
when asked to identify an artist on images generated with target style in prompt (``Van Gogh'' or ``Kelly McKernan''), and $Acc_u$ shows the accuracy of GPT-4o answers when asked
to identify an artist on images generated with other styles in prompt. The goal of any style removal method is to
lower $LPIPS_e$ and $Acce$ (i.e., successfully remove target style) while maintaining $LPIPS_u$ and $Acc_u$ high (i.e., not affecting the generation of other styles). 

Questions to GPT-4o are formulated as ``Is this picture in \textbf{s} style? Just tell me Yes or No.'', where $s$ represents style, e.g., \textbf{s} =``Andy Warhol'' or \textbf{s} =``Van Gogh''. The number
of ``Yes'' answers divided by the number of total answers is considered as Acc metric. We report results averaged over 3 runs for this metric.

\clearpage

\section{Results on SDXL model}
\label{sdxl_res}

Here we provide qualitative and qualitative results on removing concrete and abstract concepts using CASteer with the SDXL model.

\subsection{Experimental setup}

In SDXL experiments, we generate images using CASteer with $\beta=2$ on SDXL-base-1.0 model (stabilityai/stable-diffusion-xl-base-1.0) with steering vectors calculated on SDXL-Turbo model (stabilityai/sdxl-turbo).

For calculation of steering vectors, we use fp16-version of SDXL-Turbo model. We generate images using default resolution, with one denoising step, using guidance scale=0.0 and seed=0. All other parameters are left default. To generate images, we use CASteer on SDXL-base-1.0 model with 30 denoising steps. All other parameters are left default.

Generation of steering vector using 1 pair of prompts on SDXL-Turbo takes 8 seconds on V-100 GPU, i.e. generation of steering vectors for concrete concepts (see Sec.~\ref{sec:prompts}) using 50 prompts takes ~7 minutes, generation of steering vectors for human-related concepts (see Sec.~\ref{sec:prompts}) using 210 prompts takes ~28 minutes.

\subsection{Quantitative results}

We use the same experimental setups as for SD-1.4, described in Sec.~\ref{sec:experiments}. Results are shown in Tab.~\ref{tab:sdxl_nudity},~\ref{tab:sdxl_i2p},~\ref{tab:sdxl_fid},~\ref{tab:sdxl_snoopy},~\ref{tab:sdxl_t2i_art}.

\begin{table}[ht]
\centering
\caption{\textbf{Quantitative results on nudity removal based on I2P~\cite{DBLP:conf/cvpr/SchramowskiBDK23} dataset.} Detection of nude body parts is done by Nudenet at a threshold of 0.6. F: Female, M: Male. The best
results are highlighted in bold, second-best are underlined.}
\label{tab:sdxl_nudity}
\resizebox{0.65\linewidth}{!}{
\begin{tabular}{lccccccccc}
\toprule
\multirow{2}{*}{Method} & \multicolumn{9}{c}{Nudity Detection}                                 \\ \cmidrule(l){2-10}   & Breast(F)  & Genitalia(F) & Breast(M)  & Genitalia(M) & Buttocks   & Feet        & Belly   & Armpits     & Total$\downarrow$       \\ \midrule
SDXL &     85            &  7  &  3  &    2     &    7       &  28   & 84 &  66  &     282  \\
\midrule
Ours (w/o clip)   &   4  &  0 &     0    &  0 &   0  & 6 &   10  &  7   &  27 \\
Ours (clip)   &   5   &  1  &  0  &  1 &  0   &   3  & 9 &  7  & 26          \\ \bottomrule
\end{tabular}
}

\end{table}
\begin{table}[ht]
    \centering
    \caption{\textbf{Quantitative results on inappropriate content removal based on I2P\cite{DBLP:conf/cvpr/SchramowskiBDK23} dataset.} Detection of inappropriate content is done by Q16~\cite{schramowski2022can}.}
    \label{tab:sdxl_i2p}
\resizebox{0.5\textwidth}{!}{
    \begin{tabular}{
        @{} l | c |
        *{3}{S[table-format=2.1, table-column-width=20mm]} 
    }
        \toprule
        {\multirow{2}{*}{Class name}}  & \multicolumn{3}{c}{Inappropriate proportion (\%) ($\downarrow$)} \\
        \cmidrule{2-4} 
        {} & {SDXL}  & {Ours (w/o clip)} & {Ours (clip)} \\
        \midrule
        {\footnotesize Hate}
        & 39.4  & 23.8 & 20.8\\
        {\footnotesize Harassment}
        & 33.0   &  21.7 & 21.7\\
        {\footnotesize Violence}
        & 43.7   &  25.9 & 24.7\\
        {\footnotesize Self-harm}
        & 42.4  & 24.3  & 24.1\\
        {\footnotesize Sexual}
        & 45.3 & 33.0 & 32.2\\
        {\footnotesize Shocking}
        & 49.6 & 32.0  & 30.8\\
        {\footnotesize Illegal activity}
        & 36.0 & 23.5 & 23.5\\
        \hline 
        \addlinespace[0.1em] 
        {\footnotesize Overall}
        & 41.97 & 27.2 & 26.6\\
        \bottomrule
    \end{tabular}
}

\end{table}
\begin{table}[ht]
    \centering
    \caption{\textbf{General quality estimation of images generated by CASteer on SDXL model with nudity erasure.} CLIP score and FID are calculated on COCO-30k dataset}
    \label{tab:sdxl_fid}
    \resizebox{0.25\columnwidth}{!}{
        \begin{tabular}{l | c c}
            \toprule
            \multirowcell{3}[0pt][l]{Method} & \multicolumn{2}{c}{Locality} \\
            {} & \multirowcell{2}[0pt][c]{CLIP-30K($\uparrow$)} & \multirowcell{2}[0pt][c]{FID-30K($\downarrow$)} \\
            {} & {} & {} \\
            \hline 
            \addlinespace[0.1em] 
            
            {SDXL} &  31.53  & 13.29\\
            \addlinespace[-0.1em]
            \midrule
            {\textit{Ours }} & 31.51 & 13.56 \\
            {\textit{Ours (clip)}} & 31.45 & 13.37 \\
            \addlinespace[-0.1em]
            \bottomrule
        \end{tabular}
    }
\end{table}
\begin{table}[ht]
\caption{\textbf{Quantitative evaluation of concrete object erasure.}}
\label{tab:sdxl_snoopy}
\centering
\setlength{\tabcolsep}{2.0pt}
\resizebox{0.65\textwidth}{!}{
\definecolor{mygray}{gray}{.9}
\begin{tabular}{l|c|cc|cc|cc|cc|cc}
    \toprule
    
    & \multicolumn{1}{c|}{Snoopy} & \multicolumn{2}{c|}{Mickey} & \multicolumn{2}{c|}{Spongebob} & \multicolumn{2}{c|}{Pikachu} & \multicolumn{2}{c|}{Dog} & \multicolumn{2}{c}{Legislator}  \\

    \cmidrule{2-12}
    Method & CS$\downarrow$ &  CS$\uparrow$ & FID$\downarrow$ & CS$\uparrow$ & FID$\downarrow$ & CS$\uparrow$ & FID$\downarrow$ & CS$\uparrow$ & FID$\downarrow$ & CS$\uparrow$ & FID$\downarrow$ \\
    \midrule
    SDXL & 74.3 & 73.1 & - & 75.1 & - & 72.7 & - & 66.3 & - & 60.8 \\
    \midrule
    Ours  & 48.7  & 68.5 & 69.0 & 73.7 & 58.4 & 72.7 & 27.8 & 66.4 & 37.9 & 60.8 & 27.4 \\
    Ours (clip)  & 48.6  & 68.2 & 70.7 & 73.4 & 59.2 & 72.8 & 28.3 & 66.4 & 37.9 & 60.9 & 20.0 \\

    \bottomrule
\end{tabular}
}
\end{table}

\begin{table}
\caption{{Comparison of Artist Concept Removal tasks on SDXL model}: Famous (left)  and Modern artists (right).}
\label{tab:sdxl_t2i_art}
\small
\centering
\setlength{\tabcolsep}{1.mm}
\resizebox{0.65\textwidth}{!}{
\begin{tabular}{l|cccc|cccc}
\toprule
&\multicolumn{4}{c}{\textbf{Remove ``Van Gogh"}}&\multicolumn{4}{c}{\textbf{Remove ``Kelly McKernan"}}\\ 
\cmidrule(lr){2-5}
\cmidrule(lr){6-9}
\textbf{Method}&
\textbf{LPIPS}$_e \uparrow$ &\textbf{LPIPS}$_u \downarrow$ & \textbf{Acc}$_e \downarrow$& \textbf{Acc}$_u \uparrow$&
\textbf{LPIPS}$_e \uparrow$ &\textbf{LPIPS}$_u \downarrow$ & \textbf{Acc}$_e \downarrow$& \textbf{Acc}$_u \uparrow$\\
\midrule
SDXL&-&-&0.95&0.95&-&-&1.00 &0.86\\
\midrule
Ours (w/o clip) & 0.34 & 0.19 &0.00 & 0.84 & 0.39 & 0.12 & 0.00 & 0.79 \\
Ours (clip)  & 0.35 & 0.19 & 0.00 & 0.85 & 0.39 & 0.12 & 0.05 & 0.83\\

\bottomrule
\end{tabular}
}
\vspace{-0.5cm}
\end{table}

\subsection{Qualitative results}
\label{sec:sdxl_res}

In this section, we provide qualitative results on CASteer applied on SDXL model. 

First, we show results on removing ``Snoopy'' concept when generating images with four prompt templates: ``\textit{An origami $X$}'', ``\textit{A drawing of the $X$}'', ``\textit{A photo of a cool $X$}'' and ``\textit{An art of the $X$}'', where $X \in $ [``\textit{Snoopy}'', ``\textit{Mickey}'', ``\textit{Spongebob}'', ``\textit{Pikachu}'', ``\textit{dog}'', ``\textit{legislator}'']. CASteer is applied with removal strength $\beta=2$.

We see that our method removes Snoopy well (see Fig.~\ref{fig:rem_snoopy_snoopy_1} while preserving other concepts well (see Fig. ~\ref{fig:rem_snoopy_mickey_1}, ~\ref{fig:rem_snoopy_spongebob_1}, ~\ref{fig:rem_snoopy_pikachu_1},
~\ref{fig:rem_snoopy_dog_1},
~\ref{fig:rem_snoopy_legislator_1}). In fact, most of the images of non-related concepts generated with CASteer applied are almost identical to those generated by vanilla SDXL. 

Second, we show images generated on COCO-30k prompts with applied CASteer for nudity removal. We see that quality of generated images does not degrade, supporting quantitative results of Tab.~\ref{tab:sdxl_fid} (see Fig.~\ref{fig:sdxl_coco_1},~\ref{fig:sdxl_coco_2},~\ref{fig:sdxl_coco_3})

Next, on Fig.~\ref{fig:sdxl_style_removal} and Fig.~\ref{fig:sdxl_style_mosaic_origami} we show examples of CASteer applied for style removal.

\begin{figure*}[h!]
    \includegraphics[width=\linewidth]{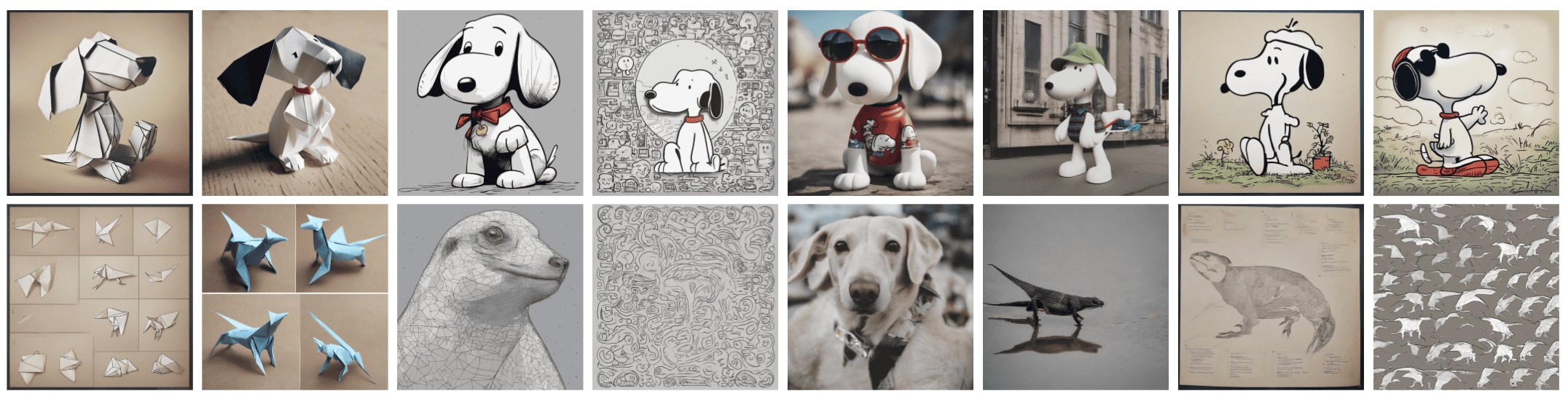}
    \caption{Images generated with ``\textit{Snoopy}'' prompts  with different seeds. Top: original SDXL, bottom: CASteer applied for removing the concept of ``\textit{Snoopy}''. } 
    \label{fig:rem_snoopy_snoopy_1}
\end{figure*}

\begin{figure*}[h!]
    \includegraphics[width=\linewidth]{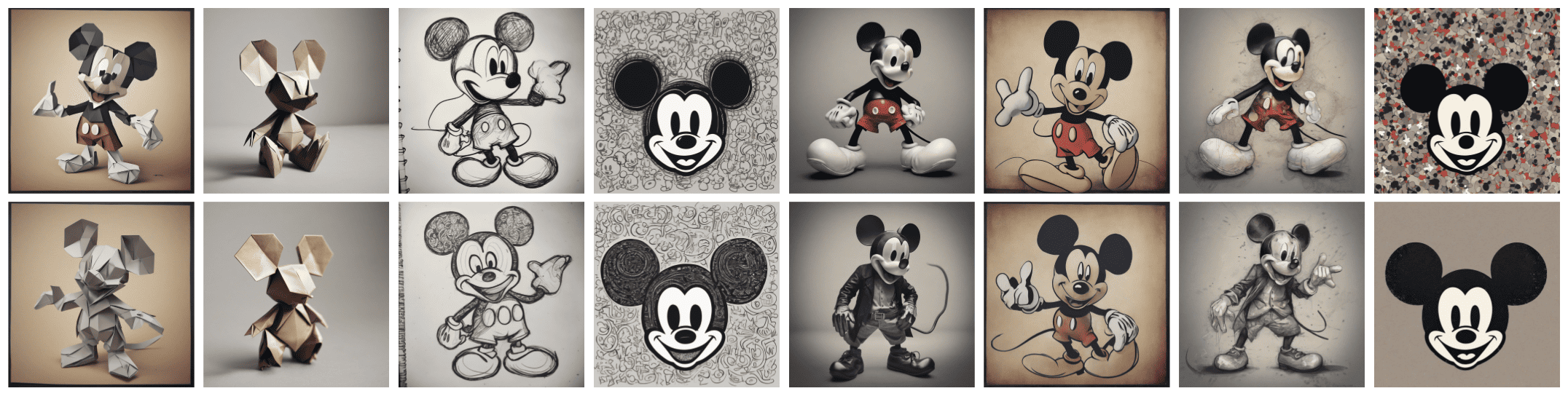}
    \caption{Images generated with ``\textit{Mickey}'' prompts  with different seeds. Top: original SDXL, bottom: CASteer applied for removing the concept of ``\textit{Snoopy}''.} 
    \label{fig:rem_snoopy_mickey_1}
\end{figure*}

\begin{figure*}[h!]
    \includegraphics[width=\linewidth]{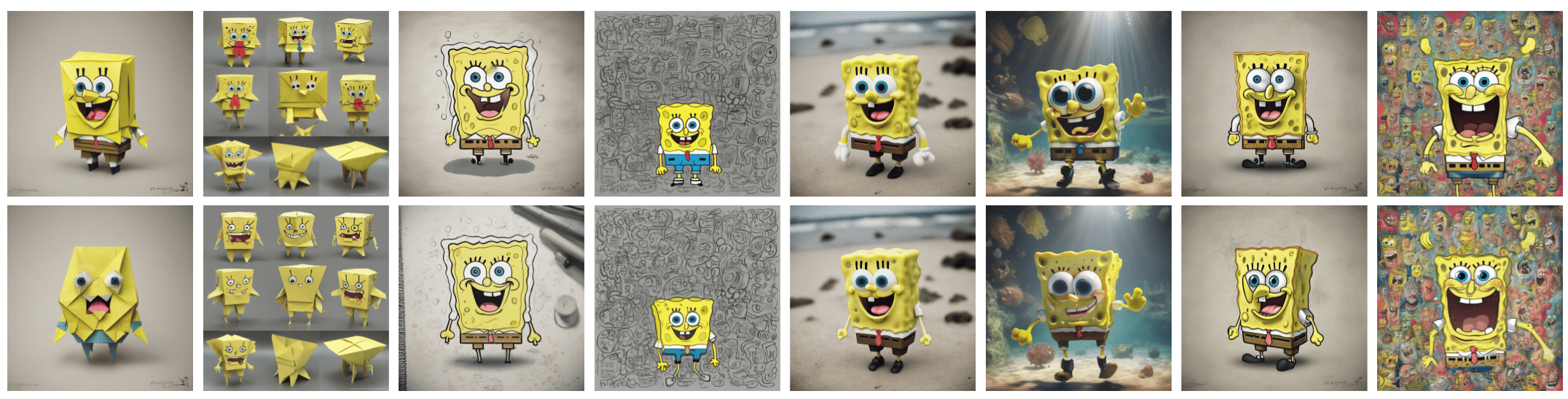}
    \caption{Images generated with ``\textit{Spongebob}'' prompts  with different seeds. Top: original SDXL, bottom: CASteer applied for removing the concept of ``\textit{Snoopy}''.} 
    \label{fig:rem_snoopy_spongebob_1}
\end{figure*}

\begin{figure*}[h!]
    \includegraphics[width=\linewidth]{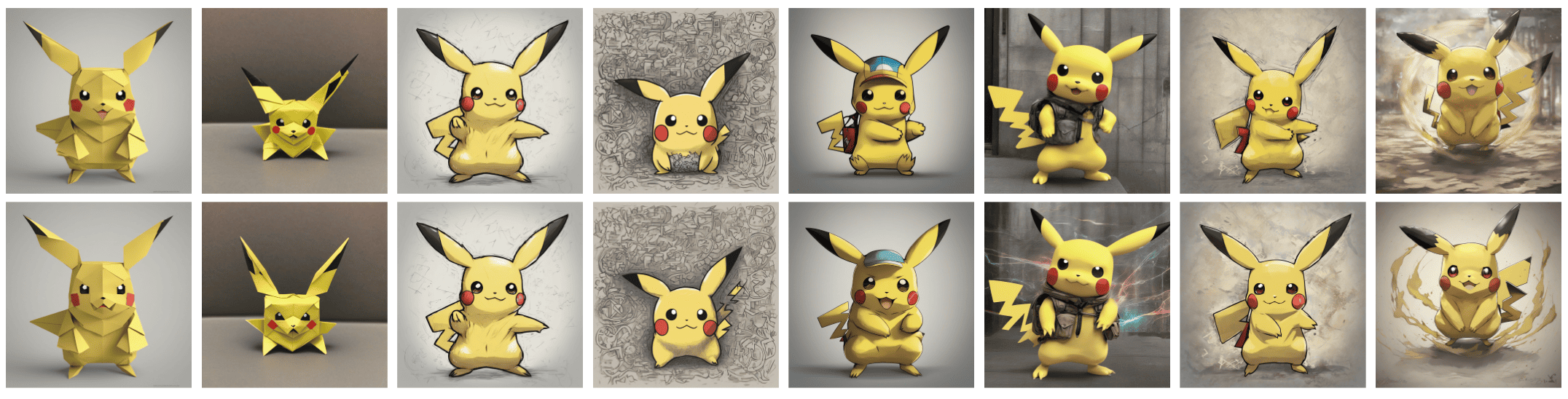}
    \caption{Images generated with ``\textit{Pikachu}'' prompts  with different seeds. Top: original SDXL, bottom: CASteer applied for removing the concept of ``\textit{Snoopy}''.} 
    \label{fig:rem_snoopy_pikachu_1}
\end{figure*}

\begin{figure*}[h!]
    \includegraphics[width=\linewidth]{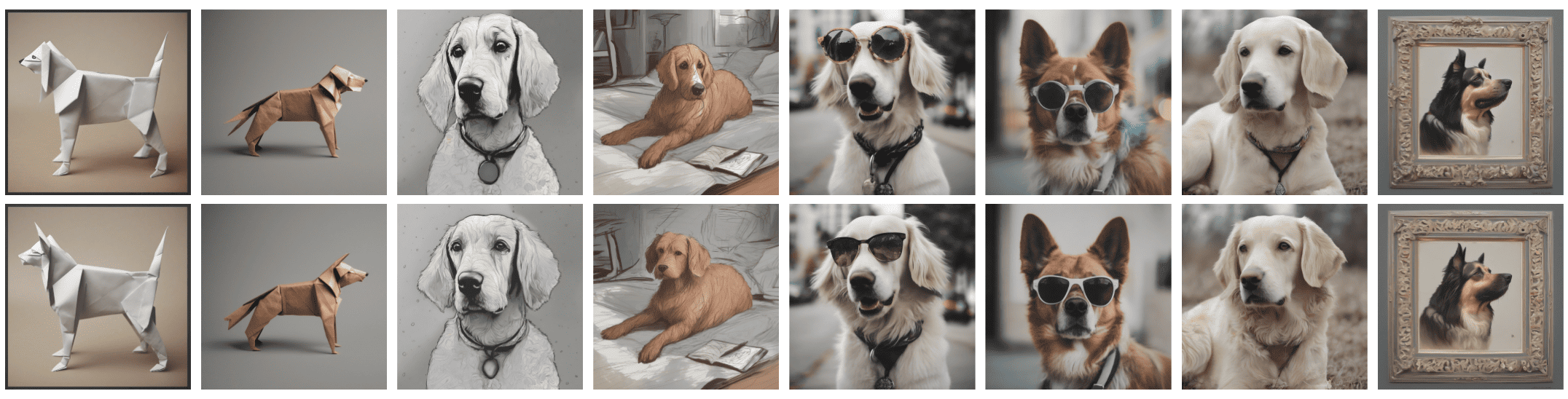}
    \caption{Images generated with ``\textit{dog}'' prompts  with different seeds. Top: original SDXL, bottom: CASteer applied for removing the concept of ``\textit{Snoopy}''.} 
    \label{fig:rem_snoopy_dog_1}
\end{figure*}

\begin{figure*}[h!]
    \includegraphics[width=\linewidth]{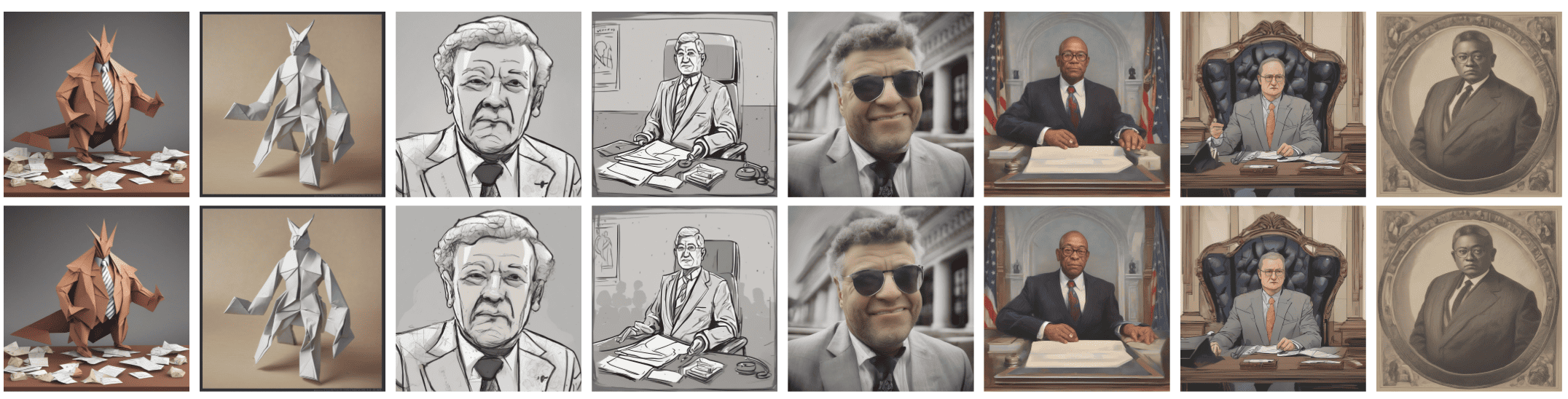}
    \caption{Images generated with ``\textit{legislator}'' prompts  with different seeds. Top: original SDXL, bottom: CASteer applied for removing the concept of ``\textit{Snoopy}''.} 
    \label{fig:rem_snoopy_legislator_1}
\end{figure*}

\begin{figure*}[h!]
    \includegraphics[width=\linewidth]{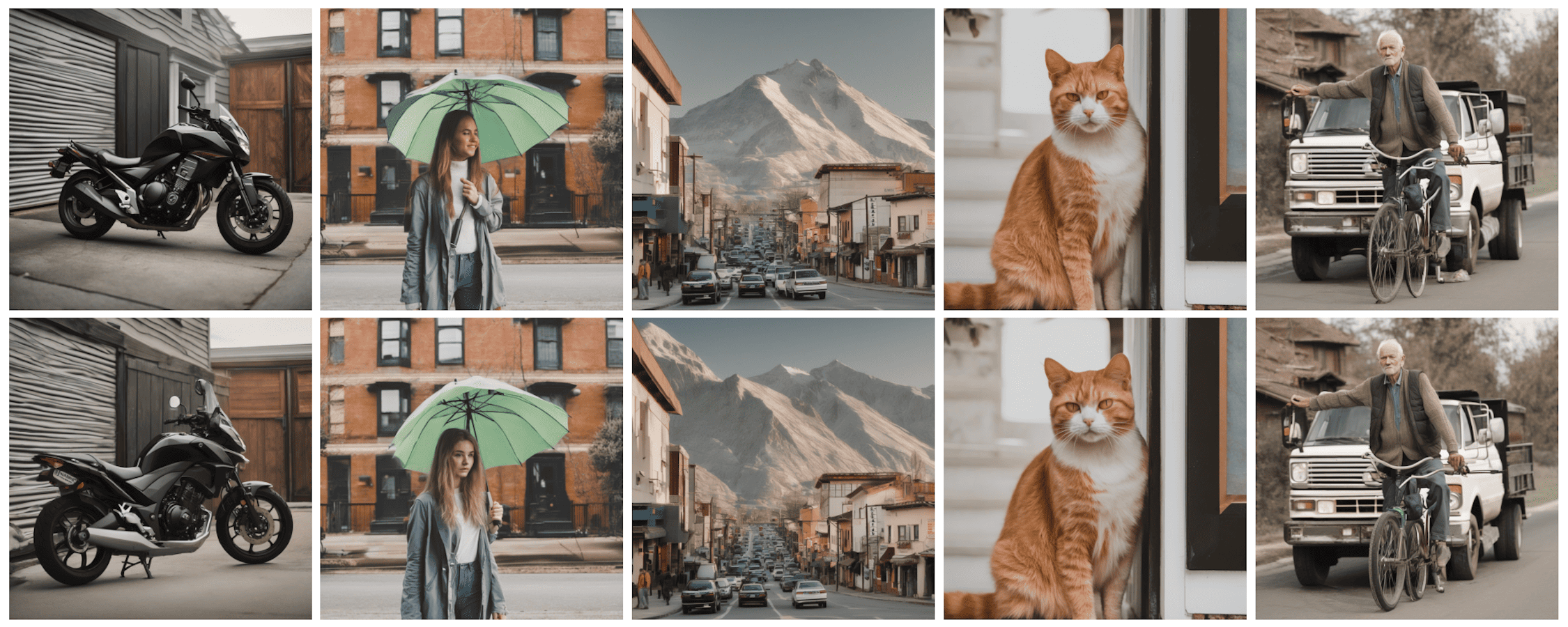}
    \caption{Images generated with different COCO-30k prompts with different seeds. Top: original SDXL, bottom: CASteer applied for removing the concept of ``\textit{nudity}''.} 
    \label{fig:sdxl_coco_1}
\end{figure*}
\begin{figure*}[h!]
    \includegraphics[width=\linewidth]{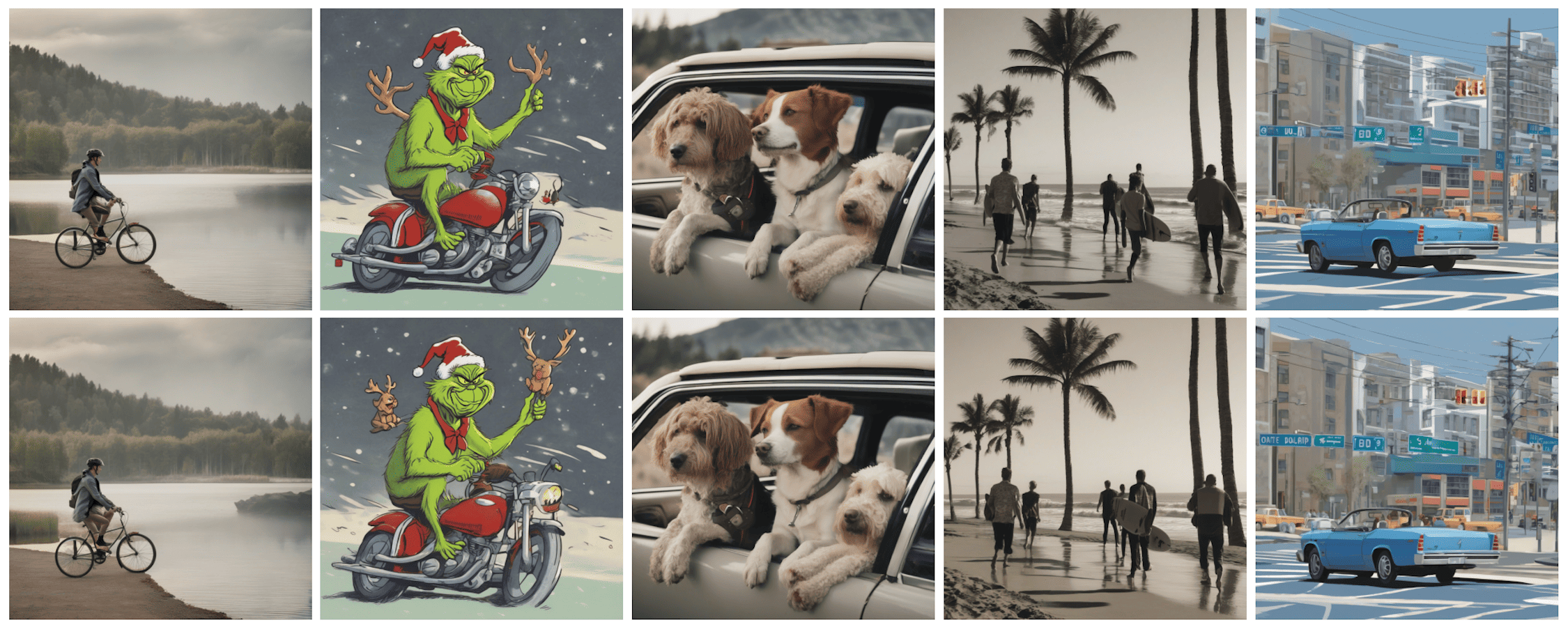}
    \caption{Images generated with different COCO-30k prompts with different seeds. Top: original SDXL, bottom: CASteer applied for removing the concept of ``\textit{nudity}''.} 
    \label{fig:sdxl_coco_2}
\end{figure*}
\begin{figure*}[h!]
    \includegraphics[width=\linewidth]{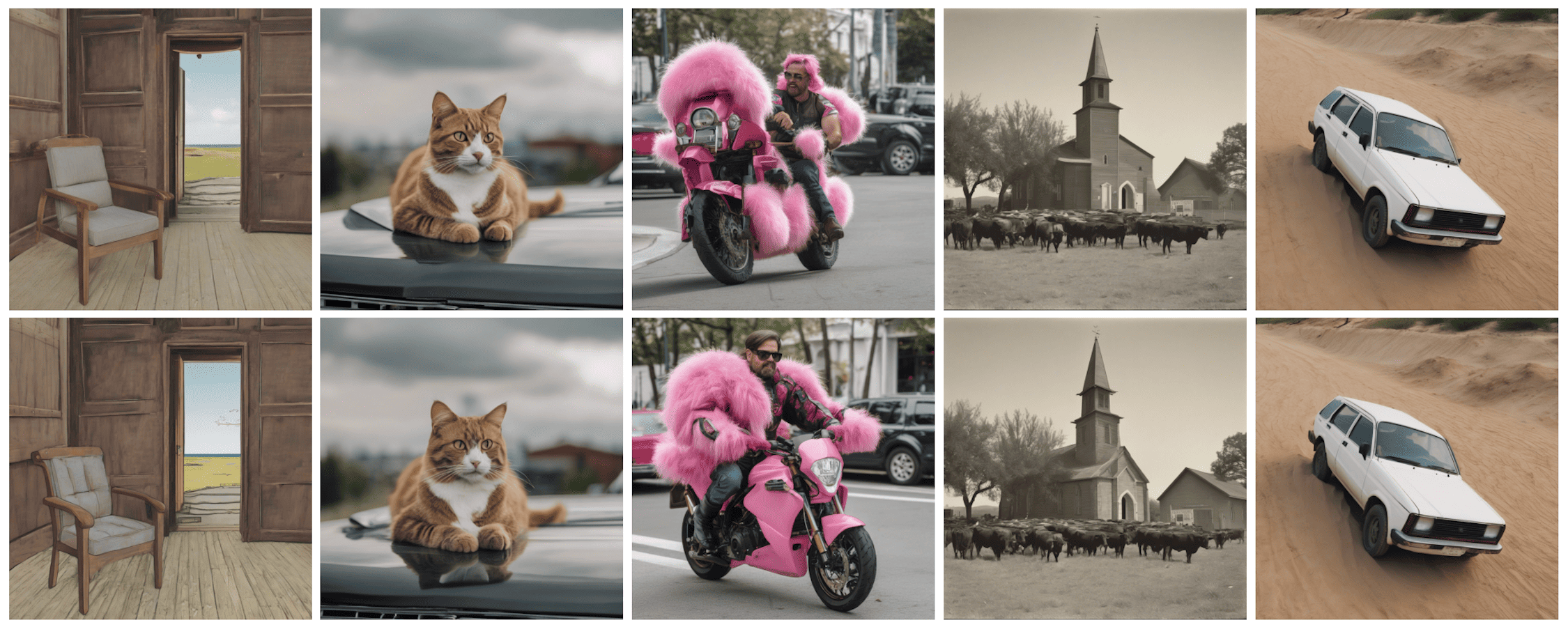}
    \caption{Images generated with different COCO-30k prompts with different seeds. Top: original SDXL, bottom: CASteer applied for removing the concept of ``\textit{nudity}''.} 
    \label{fig:sdxl_coco_3}
\end{figure*}

\begin{figure*}[h!]
    \includegraphics[width=\linewidth]{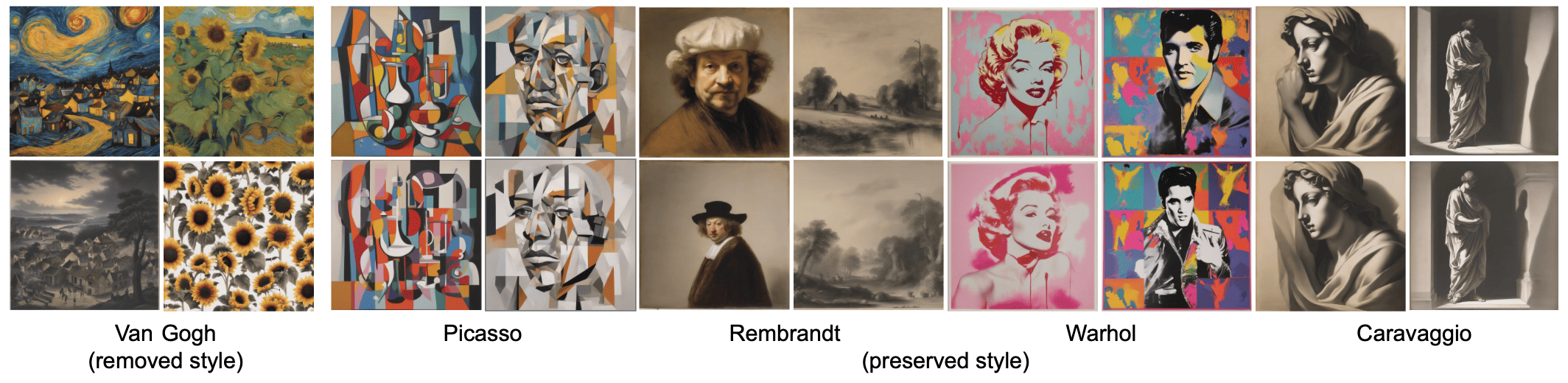}
    \caption{Examples of style removal for SDXL. Top: original SDXL, bottom: CASteer applied for removing the style of ``\textit{Van Gogh}''}
    \label{fig:sdxl_style_removal}
\end{figure*}

\begin{figure*}[h!]
    \includegraphics[width=\linewidth]{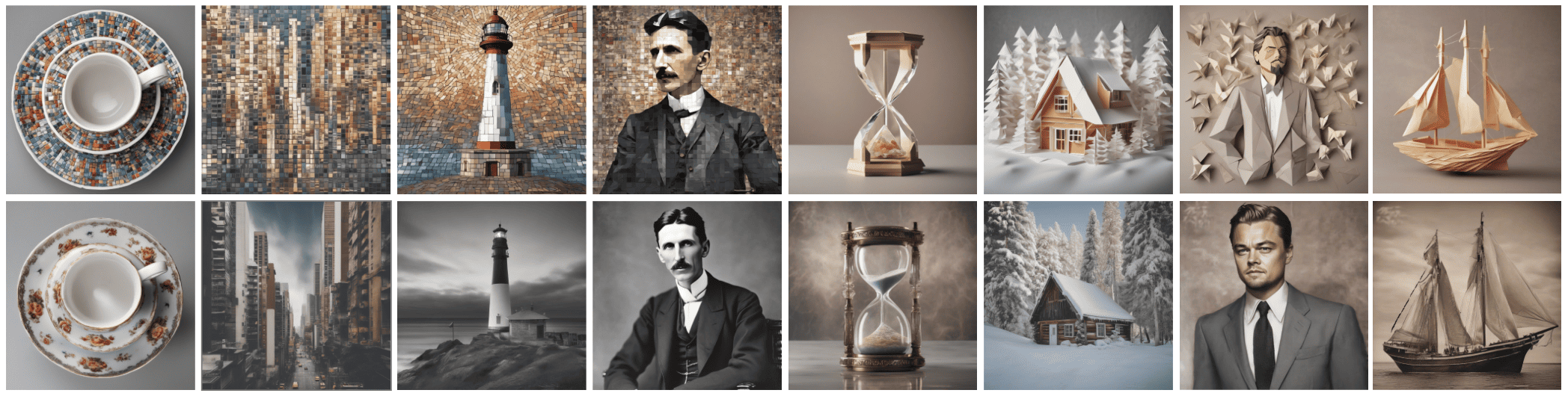}
    \caption{Examples of style removal for SDXL. Top: original SDXL, bottom: CASteer applied for removing the style of ``\textit{mosaic style}'' (left 4 images) and ``\textit{origami style}'' (right 4 images)}
    \label{fig:sdxl_style_mosaic_origami}
\end{figure*}




\clearpage
\section{Results on SANA model}
\label{sec:sana_res}

\subsection{Experimental setup}

In SANA experiments, we generate images using CASteer with $\beta=2$ on SANA\_Sprint\_1.6B\_1024px\_teacher model (Efficient-Large-Model/SANA\_Sprint\_1.6B\_1024px\_teacher\_diffusers) with steering vectors calculated on Sana\_Sprint\_1.6B\_1024px model (Efficient-Large-Model/Sana\_Sprint\_1.6B\_1024px\_diffusers).

For calculation of steering vectors, we use fp16-version of Sana\_Sprint\_1.6B\_1024px\_diffusers model. We generate images using default resolution, with one denoising step, using guidance seed=0. All other parameters are left default. To generate images, we use CASteer on SANA\_Sprint\_1.6B\_1024px\_teacher model with 20 denoising steps. All other parameters are left default.

Generation of steering vector using 1 pair of prompts on Sana\_Sprint\_1.6B\_1024px takes 5 seconds on V-100 GPU, i.e. generation of steering vectors for concrete concepts (see Sec.~\ref{sec:prompts}) using 50 prompts takes ~4.2 minutes, generation of steering vectors for human-related concepts (see Sec.~\ref{sec:prompts}) using 210 prompts takes ~18 minutes.

\subsection{Quantitative results}

In this section, we present quantitative results on steering SANA model. We use the same experimental setups as for SD-1.4, described in Sec.~\ref{sec:experiments}. Results on SANA model are shown in Tab.~\ref{tab:sana_nudity},~\ref{tab:sana_i2p},~\ref{tab:sana_fid},~\ref{tab:sana_snoopy},~\ref{tab:sana_t2i_art}.

\begin{table}[ht]
\centering
\caption{\textbf{Quantitative results on nudity removal based on I2P~\cite{DBLP:conf/cvpr/SchramowskiBDK23} dataset.} Detection of nude body parts is done by Nudenet at a threshold of 0.6. F: Female, M: Male. The best
results are highlighted in bold, second-best are underlined.}
\resizebox{0.65\linewidth}{!}{
\begin{tabular}{lccccccccc}
\toprule
\multirow{2}{*}{Method} & \multicolumn{9}{c}{Nudity Detection}                                 \\ \cmidrule(l){2-10}   & Breast(F)  & Genitalia(F) & Breast(M)  & Genitalia(M) & Buttocks   & Feet        & Belly   & Armpits     & Total$\downarrow$       \\ \midrule
SANA &     14            &  0  &  3  &    4     &    0       &  5   & 44 &  21  &     91  \\
\midrule
Ours (w/o clip)  &   3  & 0  &    0     &  1 &  0   & 0 &   3  &  2   & 9  \\
Ours (clip)   &   0   &  0  &  0  &  1 &  0   &   0  & 1 &  0  & 2          \\ \bottomrule
\end{tabular}
}

\label{tab:sana_nudity}
\end{table} 
\begin{table}[ht]
    \centering
    \caption{\textbf{Quantitative results on inappropriate content removal based on I2P~\cite{DBLP:conf/cvpr/SchramowskiBDK23} dataset (SANA).} Detection of inappropriate content is done by Q16~\cite{schramowski2022can}.}
\resizebox{0.5\textwidth}{!}{
    \begin{tabular}{
        @{} l | c |
        *{3}{S[table-format=2.1, table-column-width=20mm]} 
    }
        \toprule
        {\multirow{2}{*}{Class name}}  & \multicolumn{3}{c}{Inappropriate proportion (\%) ($\downarrow$)} \\
        \cmidrule{2-4} 
        {} & {SANA}  & {Ours (w/o clip)} & {Ours (clip)} \\
        \midrule
        {\footnotesize Hate}
        & 48.1  & 48.5 & 35.0\\
        {\footnotesize Harassment}
        &  42.8  &  42.5 & 30.3 \\
        {\footnotesize Violence}
        &  49.3  & 40.3  & 35.6\\
        {\footnotesize Self-harm}
        & 49.4 &  39.5 & 30.2\\
        {\footnotesize Sexual}
        & 36.6 & 32.6 & 22.3\\
        {\footnotesize Shocking}
        & 57.9 &  50.1 & 39.5\\
        {\footnotesize Illegal activity}
        & 42.5 & 36.9 & 25.7\\
        \hline 
        \addlinespace[0.1em] 
        {\footnotesize Overall}
        & 46.1 & 40.1 & 30.5\\
        \bottomrule
    \end{tabular}
}
\label{tab:sana_i2p}
\end{table}
\begin{table}[ht]
    \centering
    \caption{\textbf{General quality estimation of images generated by CASteer on SDXL model with nudity erasure.} CLIP score and FID are calculated on COCO-30k dataset}
    \label{table:fid_clip}
    \resizebox{0.25\columnwidth}{!}{
        \begin{tabular}{l | c c}
            \toprule
            \multirowcell{3}[0pt][c]{Method} & \multicolumn{2}{c}{Locality} \\
            {} & \multirowcell{2}[0pt][c]{CLIP-30K($\uparrow$)} & \multirowcell{2}[0pt][c]{FID-30K($\downarrow$)} \\
            {} & {} & {} \\
            \hline 
            \addlinespace[0.1em] 
            
            {SANA} &  29.28  & 22.65\\
            \addlinespace[-0.1em]
            \midrule
            {\textit{Ours }} & 28.79 & 22.89 \\
            {\textit{Ours (clip)}} & 29.01 & 23.51 \\
            \addlinespace[-0.1em]
            \bottomrule
        \end{tabular}
    }
\label{tab:sana_fid}
\end{table}

\begin{table}[ht]
\caption{\textbf{Quantitative evaluation of concrete object erasure on SANA}.}
\label{tab:sana_snoopy}
\centering
\setlength{\tabcolsep}{2.0pt}
\resizebox{0.65\textwidth}{!}{
\definecolor{mygray}{gray}{.9}
\begin{tabular}{l|c|cc|cc|cc|cc|cc}
    \toprule
    
    & \multicolumn{1}{c|}{Snoopy} & \multicolumn{2}{c|}{Mickey} & \multicolumn{2}{c|}{Spongebob} & \multicolumn{2}{c|}{Pikachu} & \multicolumn{2}{c|}{Dog} & \multicolumn{2}{c}{Legislator}  \\

    \cmidrule{2-12}
    Method & CS$\downarrow$ &  CS$\uparrow$ & FID$\downarrow$ & CS$\uparrow$ & FID$\downarrow$ & CS$\uparrow$ & FID$\downarrow$ & CS$\uparrow$ & FID$\downarrow$ & CS$\uparrow$ & FID$\downarrow$ \\
    \midrule
    SANA & 79.7 & 76.1 & - & 79.0 & - & 74.0 & - & 68.1 & - & 60.5 & - \\
    \midrule
    Ours  & 48.2  & 74.9 & 94.3 & 78.1 & 68.0 & 74.0 & 45.2 & 68.0 & 48.6 & 60.3 & 25.0 \\
    Ours (clip)  & 48.2  & 75.0 & 96.5 & 78.1 & 71.0 & 74.0 & 46.0 & 68.0 & 49.3 & 59.8 & 20.1 \\

    \bottomrule
\end{tabular}
}
\end{table}

\begin{table}
\caption{{Comparison of Artist Concept Removal tasks on SANA model}: Famous (left)  and Modern artists (right).}
\label{tab:sana_t2i_art}
\small
\centering
\setlength{\tabcolsep}{1.mm}
\resizebox{0.65\textwidth}{!}{
\begin{tabular}{l|cccc|cccc}
\toprule
&\multicolumn{4}{c}{\textbf{Remove ``Van Gogh"}}&\multicolumn{4}{c}{\textbf{Remove ``Kelly McKernan"}}\\ 
\cmidrule(lr){2-5}
\cmidrule(lr){6-9}
\textbf{Method}&
\textbf{LPIPS}$_e \uparrow$ &\textbf{LPIPS}$_u \downarrow$ & \textbf{Acc}$_e \downarrow$& \textbf{Acc}$_u \uparrow$&
\textbf{LPIPS}$_e \uparrow$ &\textbf{LPIPS}$_u \downarrow$ & \textbf{Acc}$_e \downarrow$& \textbf{Acc}$_u \uparrow$\\
\midrule
SANA&-&-&1.00&0.89&-&-&0.70&0.625\\
\midrule
Ours (w/o clip) & 0.42 & 0.14 & 0.25 & 0.88 & 0.19 & 0.11 & 0.40 & 0.620 \\
Ours (clip) & 0.42 & 0.14 & 0.25 & 0.875 & 0.19 & 0.08 & 0.30 & 0.620 \\

\bottomrule
\end{tabular}
}
\vspace{-0.5cm}
\end{table}

\subsection{Qualitative results}

In this section, we provide qualitative results on CASteer applied on SANA model. 

First, we show results on removing ``Snoopy'' concept when generating images with four prompt templates: ``\textit{An origami $X$}'', ``\textit{A drawing of the $X$}'', ``\textit{A photo of a cool $X$}'' and ``\textit{An art of the $X$}'', where $X \in $ [``\textit{Snoopy}'', ``\textit{Mickey}'', ``\textit{Spongebob}'', ``\textit{Pikachu}'', ``\textit{dog}'', ``\textit{legislator}'']. CASteer is applied with removal strength $\beta=2$.

We see that our method removes Snoopy well (see Fig.~\ref{fig:sana_rem_snoopy_snoopy_1} while preserving other concepts well (see Fig. ~\ref{fig:sana_rem_snoopy_mickey_1}, \ref{fig:sana_rem_snoopy_spongebob_1}, \ref{fig:sana_rem_snoopy_pikachu_1},
\ref{fig:sana_rem_snoopy_dog_1},
\ref{fig:sana_rem_snoopy_legislator_1}). In fact, most of the images of non-related concepts generated with CASteer applied are almost identical to those generated by vanilla SANA. 

Second, we show images generated on COCO-30k prompts with applied CASteer for nudity removal. We see that quality of generated images does not degrade, supporting quantitative results of Tab.~\ref{tab:sana_fid} (see Fig.~\ref{fig:sana_coco_1},~\ref{fig:sana_coco_2},~\ref{fig:sana_coco_3}). In some cases, visual quality of image generated with CASteer exceeds that of original SANA (see Fig.~\ref{fig:sana_coco_2}, 3rd and 5th images for example).

Next, on Fig.~\ref{fig:sana_style_removal} and Fig.~\ref{fig:sana_style_mosaic_origami}  we show examples of CASteer applied for style removal.

\begin{figure*}[h!]
    \includegraphics[width=\linewidth]{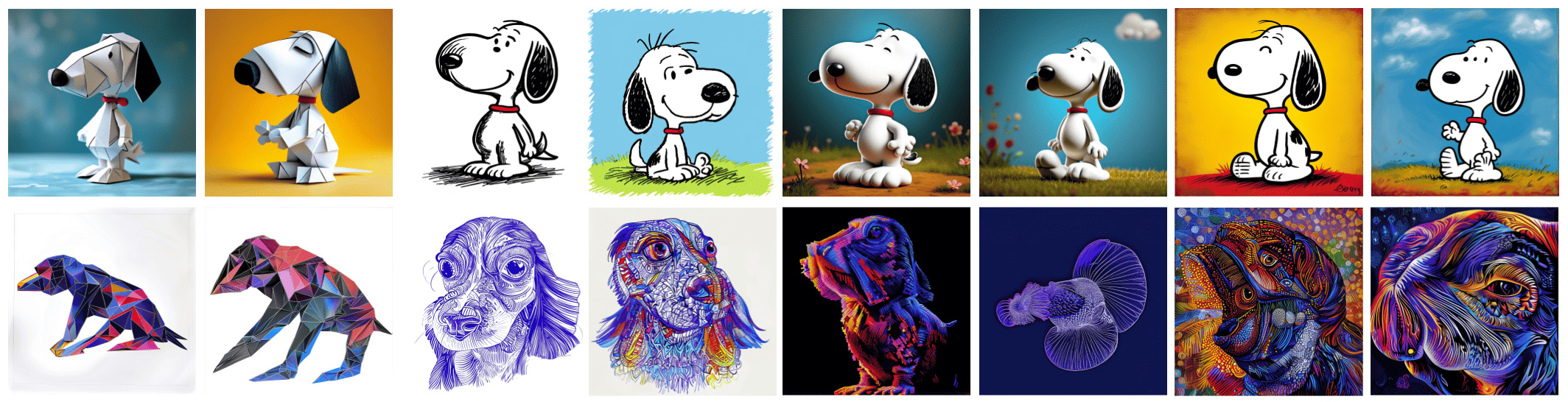}
    \caption{Images generated with ``\textit{Snoopy}'' prompts  with different seeds. Top: original SANA, bottom: CASteer applied for removing the concept of ``\textit{Snoopy}''. } 
    \label{fig:sana_rem_snoopy_snoopy_1}
\end{figure*}

\begin{figure*}[h!]
    \includegraphics[width=\linewidth]{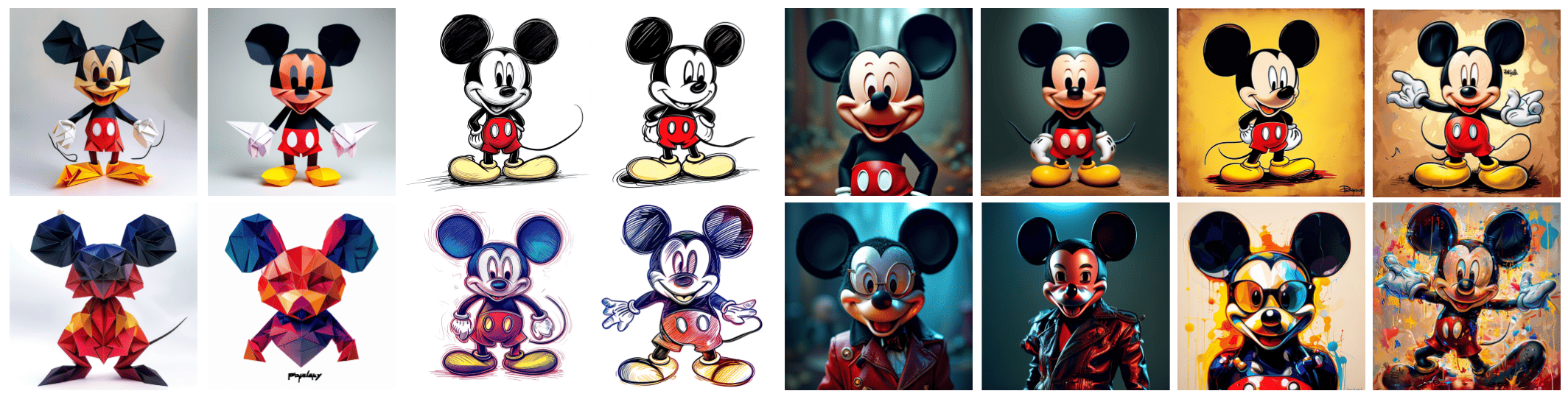}
    \caption{Images generated with ``\textit{Mickey}'' prompts  with different seeds. Top: original SANA, bottom: CASteer applied for removing the concept of ``\textit{Snoopy}''.} 
    \label{fig:sana_rem_snoopy_mickey_1}
\end{figure*}

\begin{figure*}[h!]
    \includegraphics[width=\linewidth]{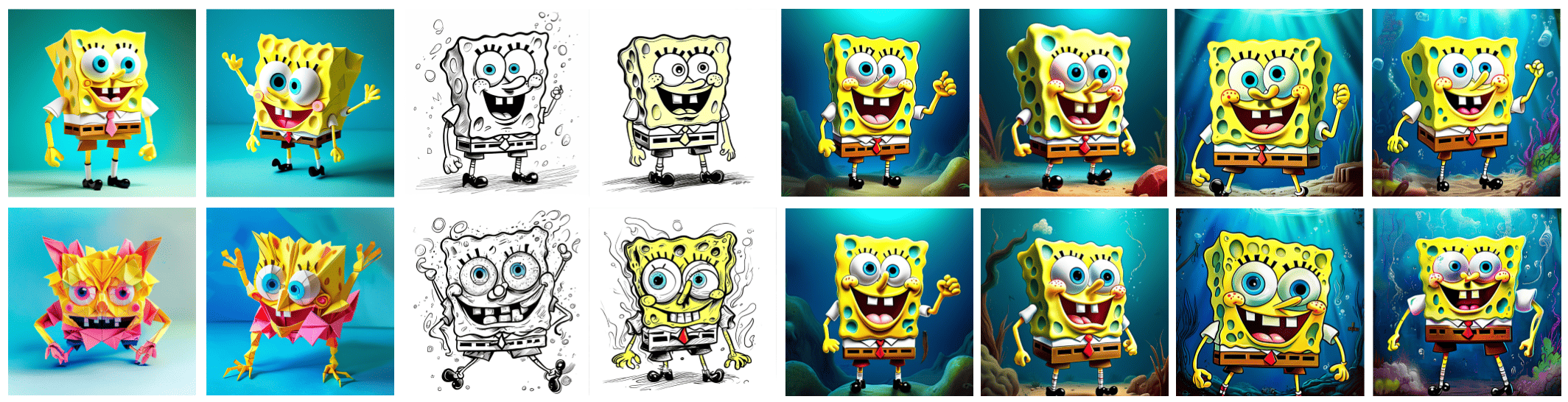}
    \caption{Images generated with ``\textit{Spongebob}'' prompts  with different seeds. Top: original SANA, bottom: CASteer applied for removing the concept of ``\textit{Snoopy}''.} 
    \label{fig:sana_rem_snoopy_spongebob_1}
\end{figure*}

\begin{figure*}[h!]
    \includegraphics[width=\linewidth]{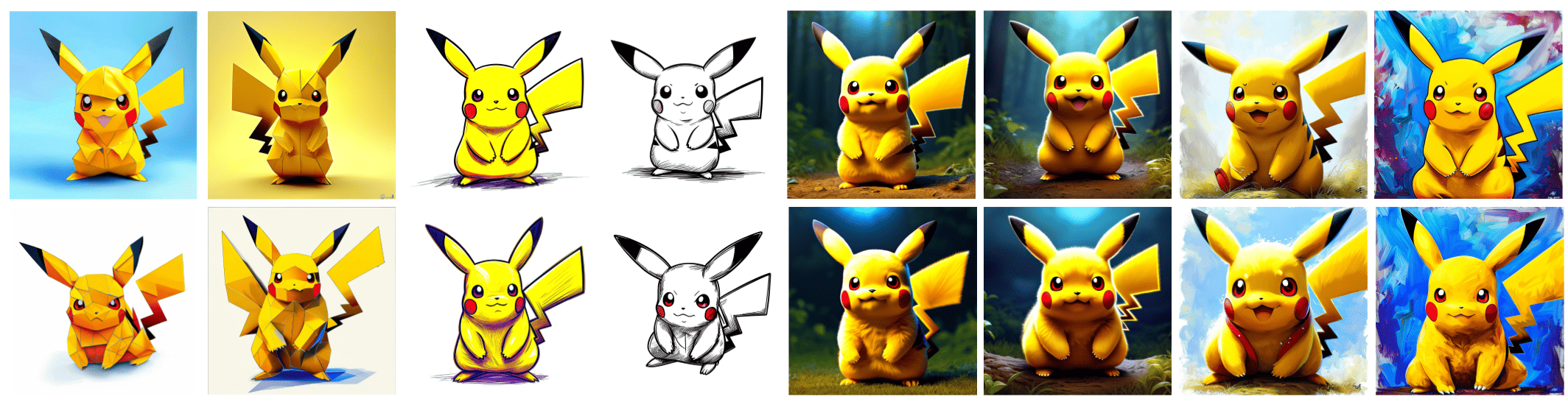}
    \caption{Images generated with ``\textit{Pikachu}'' prompts  with different seeds. Top: original SANA, bottom: CASteer applied for removing the concept of ``\textit{Snoopy}''.} 
    \label{fig:sana_rem_snoopy_pikachu_1}
\end{figure*}

\begin{figure*}[h!]
    \includegraphics[width=\linewidth]{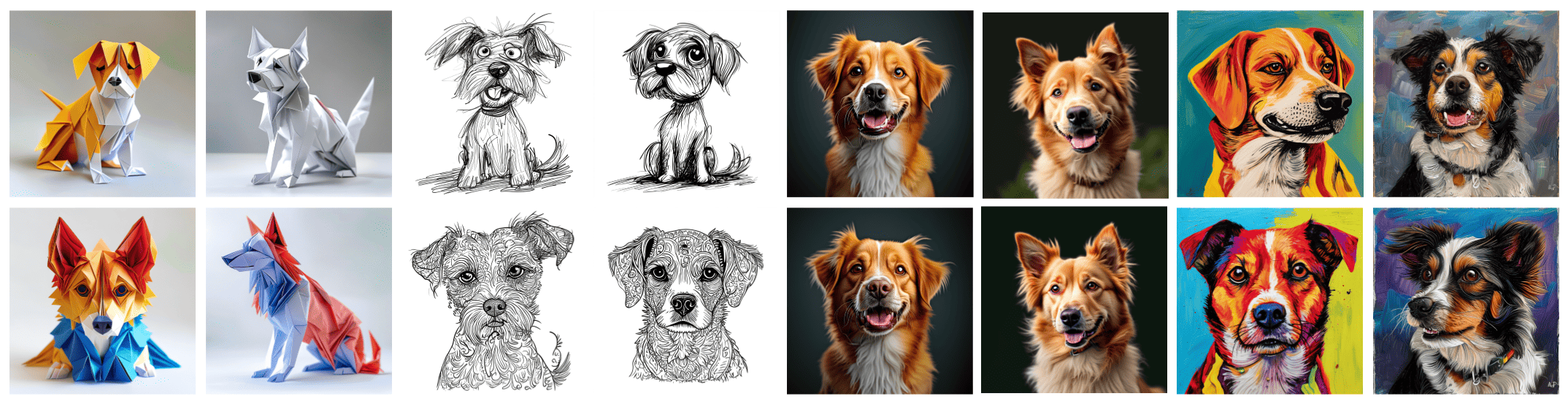}
    \caption{Images generated with ``\textit{dog}'' prompts  with different seeds. Top: original SANA, bottom: CASteer applied for removing the concept of ``\textit{Snoopy}''.} 
    \label{fig:sana_rem_snoopy_dog_1}
\end{figure*}

\begin{figure*}[h!]
    \includegraphics[width=\linewidth]{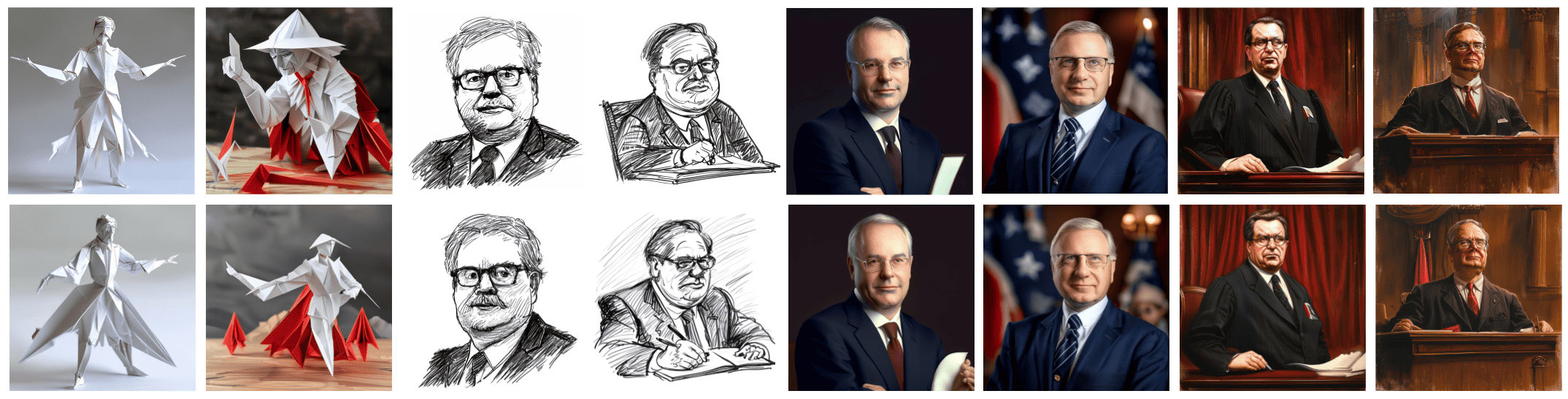}
    \caption{Images generated with ``\textit{legislator}'' prompts  with different seeds. Top: original SANA, bottom: CASteer applied for removing the concept of ``\textit{Snoopy}''.} 
    \label{fig:sana_rem_snoopy_legislator_1}
\end{figure*}

\begin{figure*}[h!]
    \includegraphics[width=\linewidth]{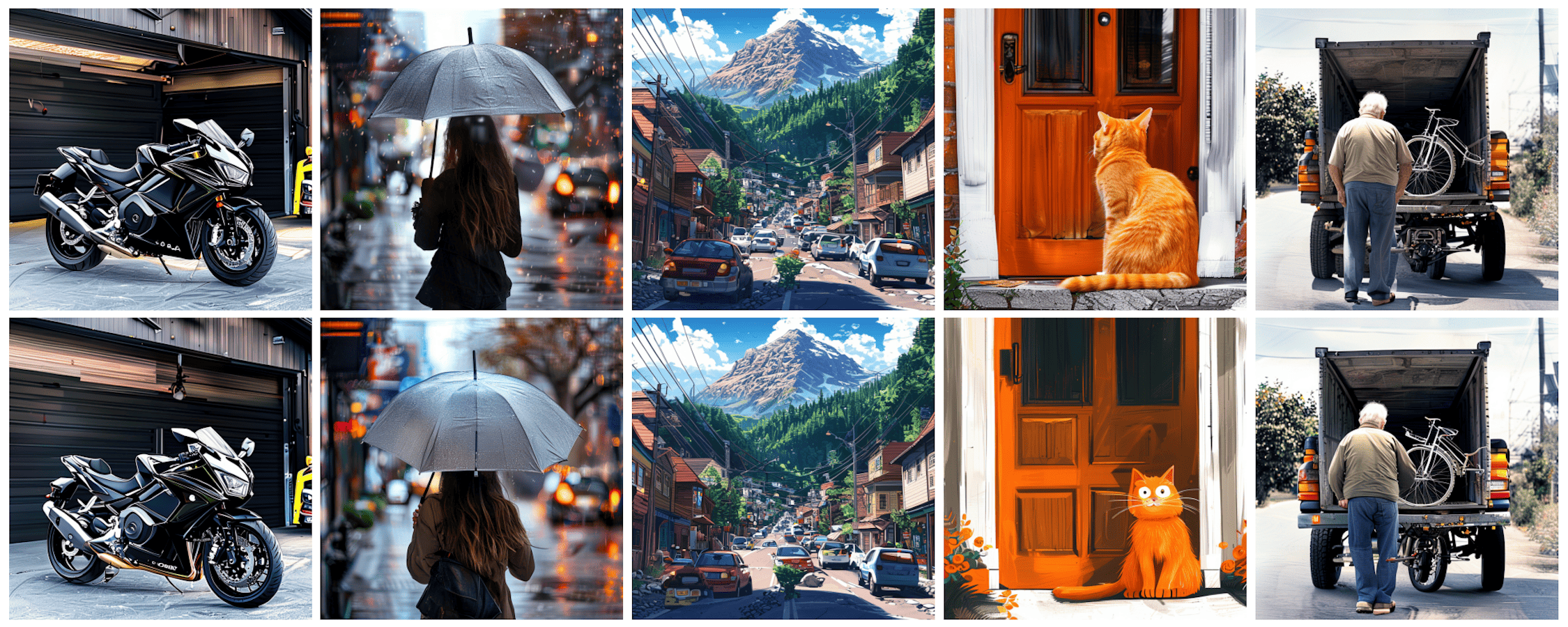}
    \caption{Images generated with different COCO-30k prompts with different seeds. Top: original SANA, bottom: CASteer applied for removing the concept of ``\textit{nudity}''.} 
    \label{fig:sana_coco_1}
\end{figure*}
\begin{figure*}[h!]
    \includegraphics[width=\linewidth]{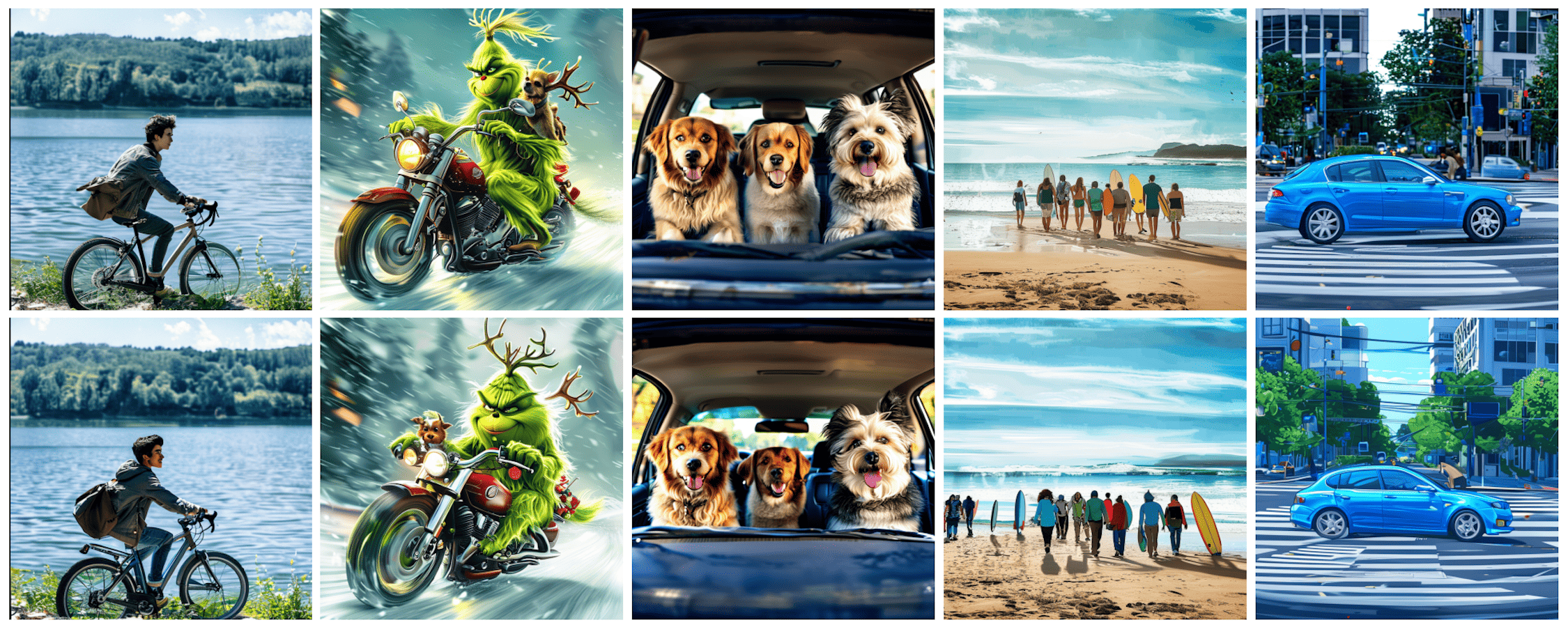}
    \caption{Images generated with different COCO-30k prompts with different seeds. Top: original SANA, bottom: CASteer applied for removing the concept of ``\textit{nudity}''.} 
    \label{fig:sana_coco_2}
\end{figure*}
\begin{figure*}[h!]
    \includegraphics[width=\linewidth]{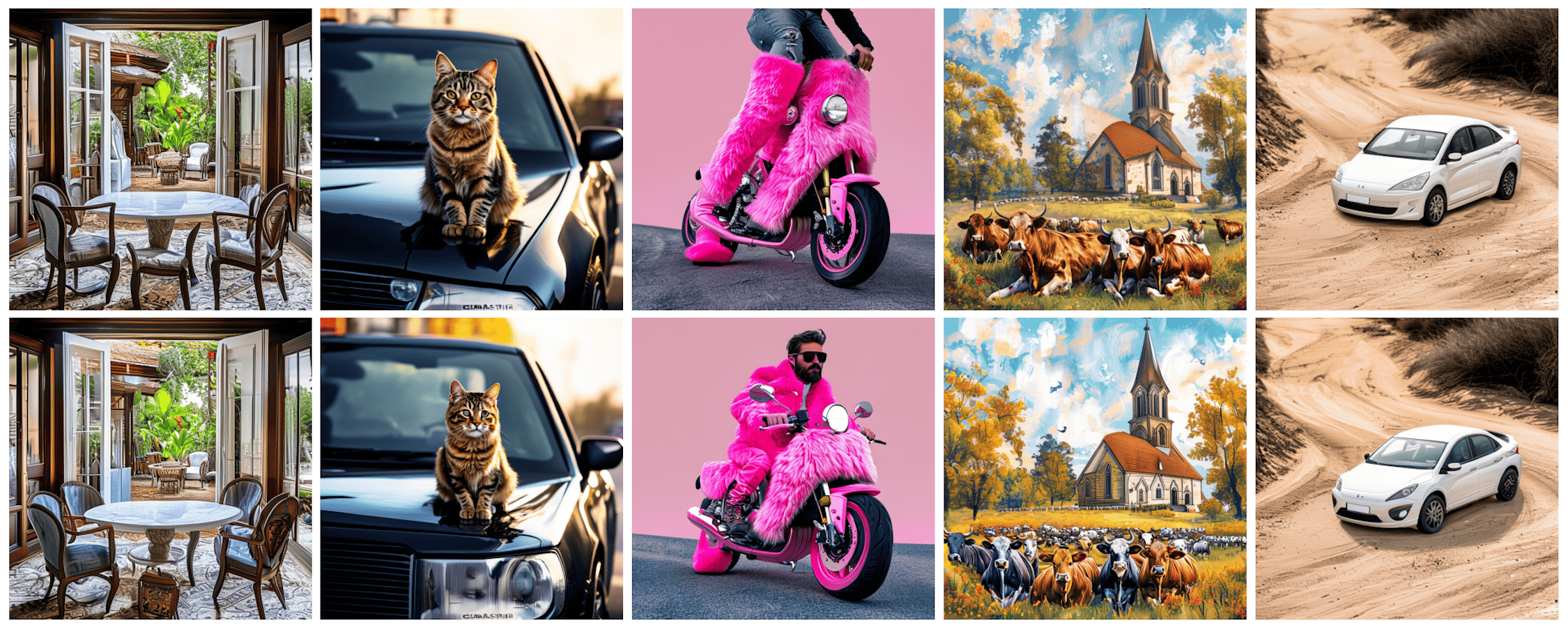}
    \caption{Images generated with different COCO-30k prompts with different seeds. Top: original SANA, bottom: CASteer applied for removing the concept of ``\textit{nudity}''.} 
    \label{fig:sana_coco_3}
\end{figure*}

\begin{figure*}[h!]
    \includegraphics[width=\linewidth]{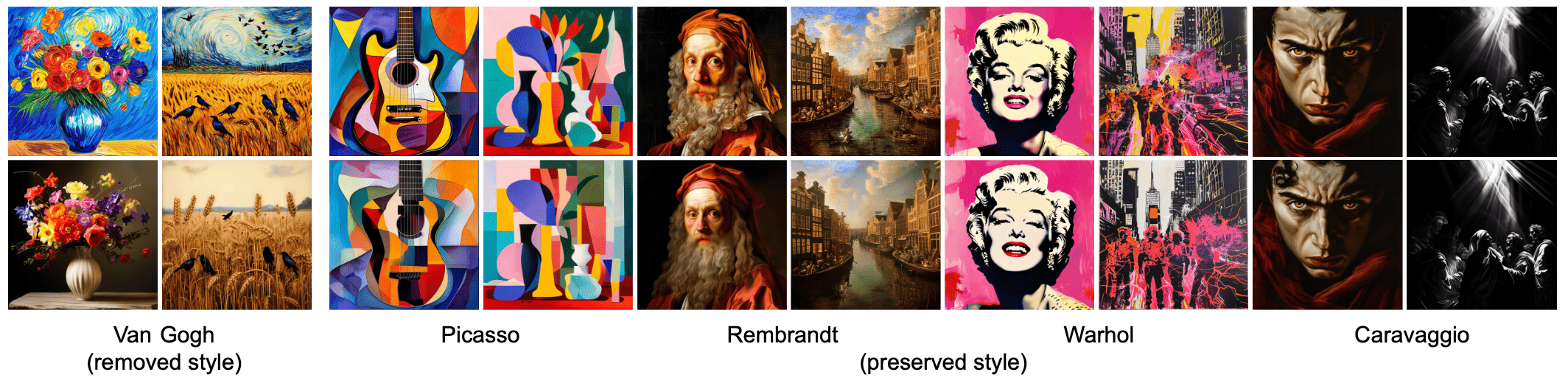}
    \caption{Examples of style removal for SANA. Top: original SANA, bottom: CASteer applied for removing the style of ``\textit{Van Gogh}''}
    \label{fig:sana_style_removal}
\end{figure*}

\begin{figure*}[h!]
    \includegraphics[width=\linewidth]{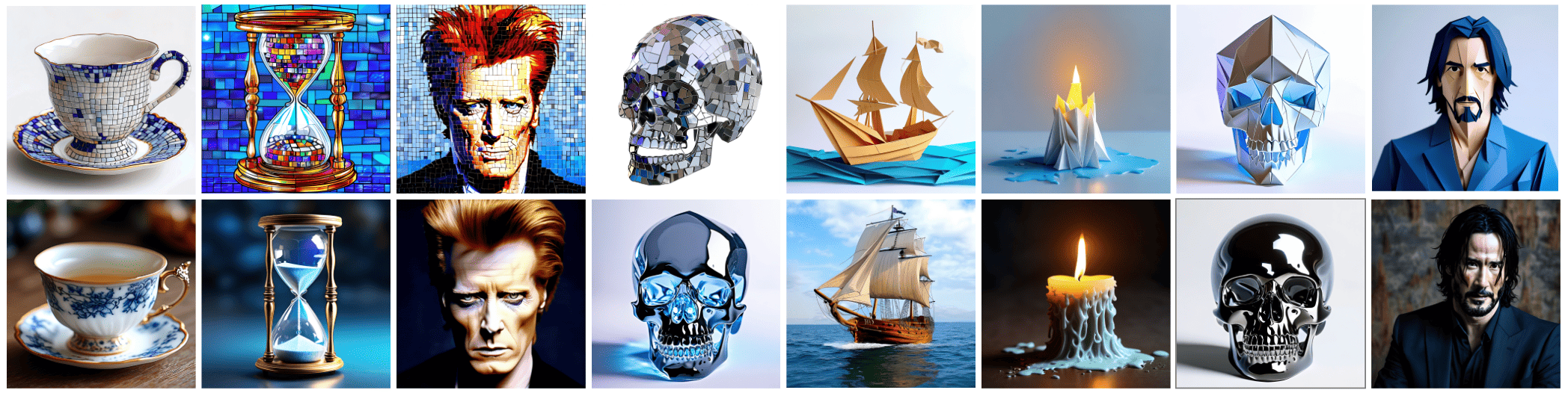}
    \caption{Examples of style removal for SANA. Top: original SANA, bottom: CASteer applied for removing the style of ``\textit{mosaic style}'' (left 4 images) and ``\textit{origami style}'' (right 4 images)}
    \label{fig:sana_style_mosaic_origami}
\end{figure*}

\clearpage
\section{Ablations on hyperparameters}
\label{sec:ablations}

In this section, we provide ablations on such hyperparameters as number and choice of prompt pairs to form a steering vector, a value of $\beta$, choice of the intermediate layer to steer, and importance of dot product weighting.

\subsection{Number and choice of steering vectors}
\label{ablation_num_vectors}
First, we compute steering vectors for the concept of ``Snoopy'' using SD-1.4 model on varying number of prompt pairs. Then we fix $\beta=2$ and apply CASteer on SD-1.4 model using computed steering vectors on 800 prompts as described in Sec.~\ref{sec:experiments}. In particular, we augment each concept using $80 $ CLIP~\cite{DBLP:conf/icml/RadfordKHRGASAM21} templates, and generate $10$ for each concept-template pair, so that for each concept there are $800$ images. 
%
%
We calculate CLIP Score (CS)~\cite{DBLP:conf/emnlp/HesselHFBC21}  and FID~\cite{DBLP:conf/nips/HeuselRUNH17} on these generated images as described in Sec.~\ref{sec:experiments}. Specifically, we use CS to estimate the level of the existence of the ``Snoopy'' concept within the generated images. Next, we calculate average FID~\cite{DBLP:conf/nips/HeuselRUNH17} scores between the set of original generations of SD-1.4 model on prompts non-related to ``Snoopy'' and a set of generations of the steered model on these prompts. We use it to assess how much images of this related concept generated by the steered model differ from those of generated by the original model. Higher FID value demonstrate more severe generation alteration.

For each number of prompt pairs, we compute the steering vector three times using different non-intersecting sets of prompts and different random seeds. Thus, for each number of prompt pairs, we report three metric values.

Figures \ref{fig:abl_num_prompts_clip_score} and \ref{fig:abl_num_prompts_fid} show CS and FID metrics for different numbers of prompt pairs used for generation of steering vectors. We see that CASteer performance remains stable across different prompt sets, and using number of pairs 50 and above results in similar performance.

\begin{figure}[htbp!]
\centering
\begin{subfigure}{0.5\textwidth}
  \centering
  \includegraphics[width=\linewidth]{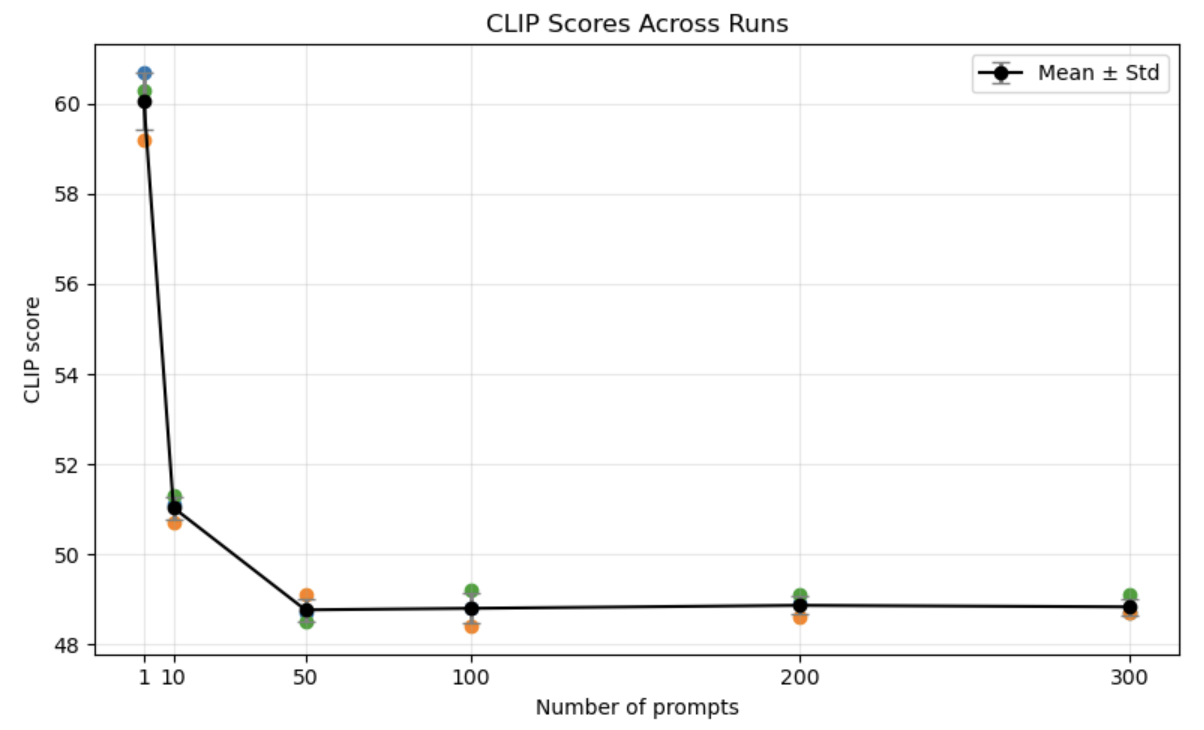}
  \caption{}
  \label{fig:abl_num_prompts_clip_score}
\end{subfigure}
\begin{subfigure}{0.49\textwidth}
  \includegraphics[width=\linewidth]{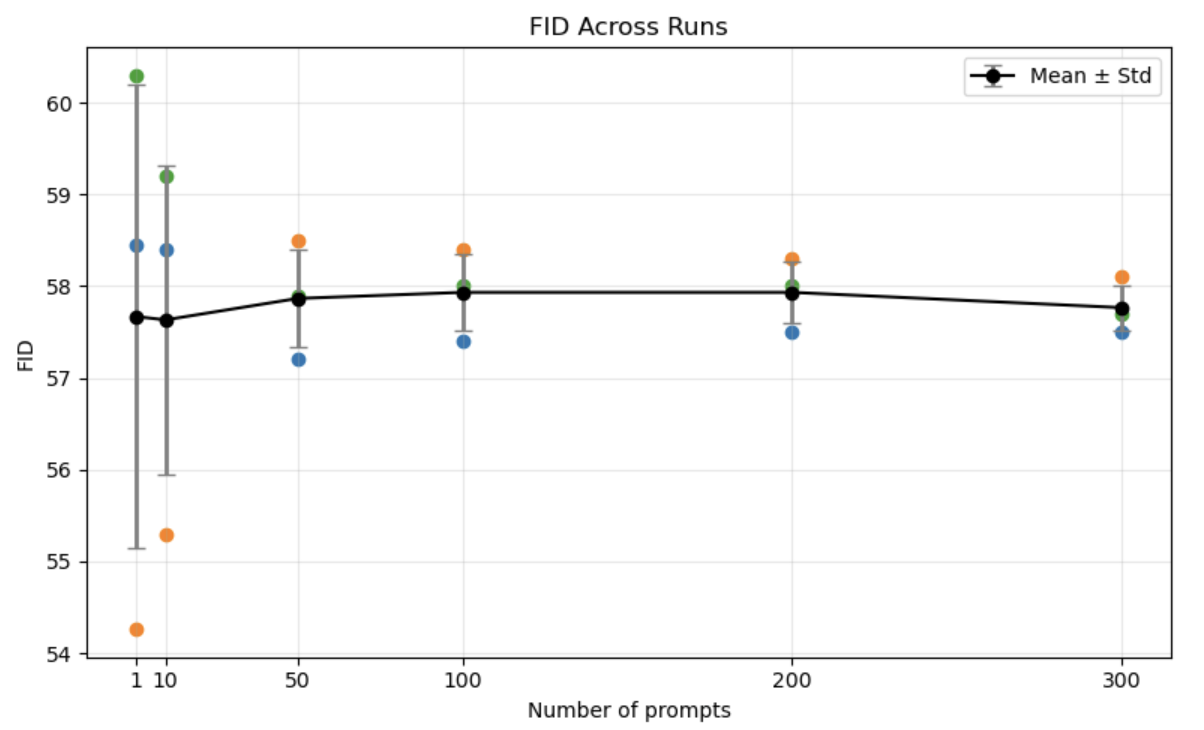}
  \caption{}
  \label{fig:abl_num_prompts_fid}
\end{subfigure}
\vspace{-0.6cm}
\caption{Ablation on number of prompts for computing steering vectors on SD-1.4. (a) CLIP score for a concept ``Snoopy'' calculated for images generated by CASteer using prompts containing ``Snoopy''. (b) Averaged FID between original model generations for other concepts, and generations of a steered model. Blue line indicates mean values across three samples}
\label{fig:test}
\vspace{-0.5cm}
\end{figure}

\subsection{Importance of dot product weighting}
\label{ablation_dot product}

Our choice of steering vector weighting for concept erasure in Eq.~\ref{eq:6} is motivated by an observation that dot product between Cross-Attention (CA) output $c$ tokens and the steering vector $s^X$ for a concept $X$ serves as a good measure of amount of $X$ that is present in each token of $c$. In this section, we provide experimental ablation on importance of his weighting. We fix $\alpha$ in Eq.~\ref{eq:4} to be constant, and perform erasure of ``Snoopy'' concept using setup as described in Sec.~\ref{sec:experiments}. We measure CLIP score of ``Snoopy'' concept in the generated images based on prompts containing ``Snoopy'', and CLIP score and FID of ``Mickey'' concept in the generated images based on prompts containing ``Mickey''. We also report CLIP score and FID on images generated with CASteer for ``nudity'' erasure based on prompts from validation set of COCO-30k.

Tab.~\ref{tab:constant_snoopy},~\ref{tab:constant_fid} show results for $\alpha=1,2$ with and without clipping applied. We also report results of CASteer as reference. We see that for both values of $\alpha$, ``Snoopy'' prompt is not erased well. Meanwhile, for $\alpha=2$, image fidelity and prompt alignment suffer, with FID being 3.5 times higher than that of CASteer. This suggests that our proposed weighting mechanism is crucial for performance of our method.

\begin{table}
\caption{\textbf{Quantitative evaluation of concrete object erasure}. The best
results are highlighted in bold, second-best are underlined. Results of other methods are taken from SPM\cite{DBLP:conf/cvpr/Lyu0HCJ00HD24} or reproduced.}
\label{tab:constant_snoopy}
\centering
\setlength{\tabcolsep}{2.0pt}
\resizebox{0.65\textwidth}{!}{
\definecolor{mygray}{gray}{.9}
\begin{tabular}{l|c|cc}
    \toprule
    
    & \multicolumn{1}{c|}{Snoopy} & \multicolumn{2}{c}{Mickey}  \\

    \cmidrule{2-4}
    Method & CS$\downarrow$ &  CS$\uparrow$ & FID$\downarrow$  \\
    \midrule
    SD-1.4 & 78.5 & 74.7 & - \\
    \midrule
    Constant steering $\beta=1$ (w/o clip) & 76.3 & 74.3 & 51.5 \\
    Constant steering $\beta=1$ (clip) &  76.5 & 74.0 & 53.4 \\
    Constant steering $\beta=2$ (w/o clip) & 72.7 & 70.7 & 141.1 \\
    Constant steering $\beta=2$ (clip) &  72.8 & 70.8 & 139.5 \\
    \midrule
    Dot Product steering (CASteer, w/o clip)  & 45.8  & 70.4 & 93.0 \\
    Dot Product steering (CASteer, clip)  & 48.5  & 70.4 & 89.6\\

    \bottomrule
\end{tabular}
}
\end{table}

\begin{table}
    \centering
    \caption{\textbf{General quality estimation of images generated by nudity-erased models.} CLIP score and FID are calculated on COCO-30k dataset}
    \resizebox{0.5\columnwidth}{!}{
        \begin{tabular}{l | c c}
            \toprule
            \multirowcell{3}[0pt][l]{Method} & \multicolumn{2}{c}{Locality} \\
            {} & \multirowcell{2}[0pt][c]{CLIP-30K($\uparrow$)} & \multirowcell{2}[0pt][c]{FID-30K($\downarrow$)} \\
            {} & {} & {} \\
            \hline 
            \addlinespace[0.1em] 
            
            {SD-1.4} &  31.34  & 14.04\\
            \addlinespace[-0.1em]
            \midrule

            {\textit{Constant steering $\beta=1$ (w/o clip)}} & 29.56 & 14.07 \\
            {\textit{Constant steering $\beta=1$ (clip)}} & 30.84  & 14.05 \\
            {\textit{Constant steering $\beta=2$ (w/o clip)}} &  27.67 & 55.34 \\
            {\textit{Constant steering $\beta=2$ (clip)}} & 28.22 & 55.21 \\
            \midrule
            {\textit{Dot Product steering (CASteer, w/o clip)}} & 30.69 & 13.28 \\
            {\textit{Dot Product steering (CASteer, clip)}} & 30.09 & 13.02 \\
            \addlinespace[-0.1em]
            \bottomrule
        \end{tabular}
    }
\label{tab:constant_fid}
\end{table}

\subsection{Steering strength}
\label{sec:ablation_strength}
Note that CASteer has a hyperparameter $\beta$, which determines the strength of steering for concept removal. While we set $\beta=2$ in all our experiments, and this choice is motivated by the properties of the resulting transformation (see Sec.~\ref{sec:experiments}), it is still possible to use CASteer with varying values of strength $\beta$. Although experimentally we observe that value of $\beta=2$ is optimal for all erasing scenarios and we use $\beta=2$ in all our experiments, $\beta$ still can be tuned for each use case, and varying values of $\beta$ can lead to different trade-offs between level of target concept erasure and alteration of generated images not containing target concept. To show how $\beta$ influences performance, in this section we provide results on ``Snoopy'' and ``nudity'' erasure using CASteer on SD-1.4, SDXL and SANA with different values of $\beta$ 

Tab.~\ref{tab:strenghts_fid} shows image quality metrics obtained using CASteer erasing ``nudity'' on COCO-30k validation dataset. It can be seen that steering with strength up to 2 does not lower general image quality for all models. Next, for SDXL and SANA models, steering with strength $\beta>2$ up to $\beta=4$ also does not affect general image quality, while for SD-1.4 steering with $\beta>2$ results in degradation of image fidelity. 

Tab.~\ref{tab:strength_sana_nudity} and Tab.~\ref{tab:strength_sdxl_nudity} show results on nudity erasure using varying steering strength $\beta$. We see that $\beta<2$ results in a lesser suppression of the erased concept. 
In general, we see that as $\beta$ increases, level of concept suppression is increased. However, for SANA with $\beta \geqslant 3$, and for SDXL with $\beta \geqslant 4$, CASteer without clipping starts showing lower rates of erasure. We hypothesize that this happens due to the absence of clipping, when image tokens with negative dot product with steering vector of ``nudity'' start containing substantial positive amount of ``nudity'' after steering. This hypothesis is supported by that CASteer with clipping does not show degrading performance with large values of $\beta$.

\begin{table}[ht]
    \centering
    \caption{\textbf{General quality estimation of images generated by CASteer on SD-1.4, SDXL and SANA model with nudity erasure.} CLIP score and FID are calculated on COCO-30k dataset}
    \label{tab:strenghts_fid}
    \resizebox{0.9\columnwidth}{!}{
        \begin{tabular}{l | c c | c c | c c}
            \toprule
            \multirowcell{3}[0pt][l]{Method} & \multicolumn{2}{c}{SD-1.4} & \multicolumn{2}{c}{SDXL} & \multicolumn{2}{c}{SANA}\\
            {} & \multirowcell{2}[0pt][c]{CLIP-30K($\uparrow$)} & \multirowcell{2}[0pt][c]{FID-30K($\downarrow$)} & \multirowcell{2}[0pt][c]{CLIP-30K($\uparrow$)} & \multirowcell{2}[0pt][c]{FID-30K($\downarrow$)} & \multirowcell{2}[0pt][c]{CLIP-30K($\uparrow$)} & \multirowcell{2}[0pt][c]{FID-30K($\downarrow$)} \\
            {} & {} & {} \\
            \hline 
            \addlinespace[0.1em] 
            
            {Vanilla model} &  31.34  & 14.04 &  31.53  & 13.29 & 29.28  & 22.65 \\
            \addlinespace[-0.1em]
            \midrule
            {\textit{Ours, $\beta$=1 }}  & 31.44 & 13.87& 31.56 &  13.34 &  29.31 & 22.71\\
            {\textit{Ours, $\beta$=1.5 }}  & 31.22 & 13.54 & 31.57 & 13.21 &  29.17& 22.62\\
            
            {\textit{Ours, $\beta$=2 }}  & 30.69 & 13.28 & 31.51 & 13.56 & 28.79 & 22.83 \\
            {\textit{Ours, $\beta$=2.5}}  & 30.36& 21.06 & 31.59 & 13.49 &  28.80 & 22.76 \\
            {\textit{Ours, $\beta$=3.0  }}  & 29.72& 28.65 & 31.60 & 13.43 & 28.82 & 22.88 \\
            {\textit{Ours, $\beta$=4.0  }}  & 29.21& 51.05 & 31.59 & 13.49 & 28.88 & 22.78 \\
            \midrule
             {\textit{Ours, $\beta$=1 (clip)}}  & 31.42 & 13.69& 31.52 & 13.27 & 29.34 & 22.72\\
             {\textit{Ours, $\beta$=1.5 (clip)}}  & 31.36 & 13.42 & 31.54 & 13.23 & 29.23& 22.56\\
            {\textit{Ours, $\beta$=2 (clip)}}  & 31.09 & 13.02 & 31.45 & 13.37 & 29.01 & 22.87 \\
            {\textit{Ours, $\beta$=2.5 (clip)}}  & 30.74& 16.08 & 31.44 & 13.35 & 29.04 &  22.71\\
            {\textit{Ours, $\beta$=3.0 (clip)}}  & 30.14& 20.11 & 31.44 & 13.32 & 20.03& 22.87\\
            {\textit{Ours, $\beta$=4.0 (clip)}}  & 29.98& 36.54& 31.42 & 13.37 & 29.02& 22.79\\
            \addlinespace[-0.1em]
            \bottomrule
        \end{tabular}
    }
\end{table}

\begin{table}[ht]
\centering
\caption{\textbf{Quantitative results on nudity removal based on I2P~\cite{DBLP:conf/cvpr/SchramowskiBDK23} dataset.} using CASteer on SANA model with varying strength $\beta$. Detection of nude body parts is done by Nudenet at a threshold of 0.6. F: Female, M: Male. The best
results are highlighted in bold, second-best are underlined.}
\resizebox{\linewidth}{!}{
\begin{tabular}{lccccccccc}
\toprule
\multirow{2}{*}{Method} & \multicolumn{9}{c}{Nudity Detection}                                 \\ \cmidrule(l){2-10}   & Breast(F)  & Genitalia(F) & Breast(M)  & Genitalia(M) & Buttocks   & Feet        & Belly   & Armpits     & Total$\downarrow$       \\ \midrule
SANA &     14            &  0  &  3  &    4     &    0       &  5   & 44 &  21  &     91  \\
\midrule
Ours $\beta=1$    &   3  & 0  & 0 &  1 &  1   & 1 &  8   &  6   & 20  \\
Ours $\beta=1.5$   &  3   &  0 & 0 &  2 &   0  & 0 &   3  &  3   &  11 \\
Ours $\beta=2$   &   3  & 0  &    0     &  1 &  0   & 0 &   3  &  2   & 9  \\
Ours $\beta=2.5$   &  2   & 0  & 0 & 1  &   0  & 0 &  3   &   3  &  9 \\
Ours $\beta=3$   &   4  & 0  & 1 &  2 &   0  & 0 &   4  &   5  &  16 \\
Ours $\beta=4$   &  6   &  0 &  0&  3 &   2  & 0 &   4  &   6  &  21 \\
\midrule
Ours $\beta=1$ (clip)   &    1  &  0  &  0  & 0  &   1  &  2   & 8 &  5  & 17 \\    
Ours $\beta=1.5$ (clip)   &   0   &  0  &  0  &  1 &   0  &  0   & 2 &  3  & 6 \\
Ours $\beta=2$ (clip)   &   0   &  0  &  0  &  1 &  0   &   0  & 1 &  0  & 2  \\
Ours $\beta=2.5$ (clip)   &   0   &   0 &   0 & 1  &   0  &  0   & 1 & 0   & 2 \\
Ours $\beta=3$ (clip)    &   1   &   0 &   0 & 0  &   0  &  0   & 1 & 0   &2 \\
Ours $\beta=4$ (clip)   &   1   &   0 &   0 & 0  &   0  &  0   & 0 & 1   & 2 \\\bottomrule
\end{tabular}
}

\label{tab:strength_sana_nudity}
\end{table}
\begin{table}[ht]
\centering
\caption{\textbf{Quantitative results on nudity removal based on I2P~\cite{DBLP:conf/cvpr/SchramowskiBDK23} dataset.} using CASteer on SDXL model with varying strength $\beta$. Detection of nude body parts is done by Nudenet at a threshold of 0.6. F: Female, M: Male. The best
results are highlighted in bold, second-best are underlined.}
\resizebox{\linewidth}{!}{
\begin{tabular}{lccccccccc}
\toprule
\multirow{2}{*}{Method} & \multicolumn{9}{c}{Nudity Detection}                                 \\ \cmidrule(l){2-10}   & Breast(F)  & Genitalia(F) & Breast(M)  & Genitalia(M) & Buttocks   & Feet        & Belly   & Armpits     & Total$\downarrow$       \\ \midrule
SDXL &     85            &  7  &  3  &    2     &    7       &  28   & 84 &  66  &     282  \\
\midrule
Ours $\beta=1$    &   12  & 0  & 3 &  2 &  1   & 15 &  25   &  28   & 86  \\
Ours $\beta=1.5$   &   9  &  0 & 0 &  1 &   1  & 9 &   12  &   19  &  51 \\
Ours $\beta=2$   &   4  &  0 &  0  &  0 &   0  & 6 &   10  &  7   &  27 \\
Ours $\beta=2.5$   &  5   &  0 & 0 & 1  &  0   & 5 &   4  &  5   &  20 \\
Ours $\beta=3$   &  4   & 0  & 0 &  1 &   0  & 5 &   5  &   1  &  16 \\
Ours $\beta=4$   &   6  &  0 & 1 & 1  &   0  & 5 &   6  &  1   &  20 \\
\midrule
Ours $\beta=1$ (clip)   &   15  &  0 & 1 & 1  &1  & 16 &  27 &   28  & 89  \\
Ours $\beta=1.5$ (clip)   &   7  &  0 & 1 & 0  &   0  & 10 & 11  & 20 &  49 \\
Ours $\beta=2$ (clip)   &   5   &  1  &  0  &  1 &  0   &   3  & 9 &  7  & 26 \\ 
Ours $\beta=2.5$ (clip)   &   2  &  0 & 0 &  1 & 0  & 2 &  4 &  4 &  13 \\
Ours $\beta=3$ (clip)    &   3  &  0 & 1 &  2 &  0 & 3 &  2 &  2 &  13 \\
Ours $\beta=4$ (clip)   & 3  &  0 & 0 &  1 &  0 & 1 &  3 &  2 & 10 \\
\bottomrule
\end{tabular}
}

\label{tab:strength_sdxl_nudity}
\end{table}


\clearpage
\subsection{Steering other layers}
\label{ablation_other_layers}

Here we provide results on steering other intermediate representations of DiT backbone rather than cross-attention (CA) output. 

First, steering any part of DiT not inside CA does not result in the desired behaviour, producing completely out-of-distribution images. Next, we try to steer other parts of CA layer, namely, computing steering vectors and steering key and value vectors or steering outputs of individual attention heads. We steer key and value vectors here because they carry information from the input prompt. 

We use CASteer on SD-1.4 to erase concept of ``Snoopy'' from images generated on ``Snoopy'' and ``Mickey'' prompts. We calculate CLIP Score (CS) and FID as described above in Sec.~\ref{ablation_num_vectors}. Tab.~\ref{tab:ablation_other_layers} presents the results. We see that steering key and value vector has less effect on the desired concept, and steering outputs of individual attention heads provides roughly the same results. 

\begin{table}[h!]
\caption{\textbf{Quantitative evaluation of steering other layers}.}
\label{tab:ablation_other_layers}
\centering
\setlength{\tabcolsep}{2.0pt}
\resizebox{0.5\textwidth}{!}{
\definecolor{mygray}{gray}{.9}
\begin{tabular}{l|c|c}
    \toprule
    & Snoopy & Mickey \\
    \cmidrule{2-3}
    Method & CS$\downarrow$ &  FID$\downarrow$  \\
    \midrule
    SD-1.4 & 74.3 & - \\
    \midrule
    Key-Value outputs ($\beta=2$)   & 62.79  & 43.5 \\
    CA Heads outputs ($\beta=2$)  & 48.94  & 67.3 \\
    CA outputs ($\beta=2$)  & 48.7  & 68.5  \\
    \bottomrule
\end{tabular}
}
\end{table}

\subsection{Steering only fraction of layers}
\label{sec:steering_fraction}

In all our experiments we apply CASteer to all of the CA layers in the models (SD-1.4, SDXL or SANA)/ In this section, we provide qualitative experiments on steering only a fraction of CA layers. 

We ablate on three ways of choosing a subset of CA layers for steering: 
\begin{itemize}
\item Steering only $k$ first CA layers, $0 \leqslant k \leqslant n$;
\item Steering only $k$ last layers, $0 \leqslant k \leqslant n$; 
\item Steering only $k^{th}$ CA layer, $0 \leqslant k \leqslant n$.
\end{itemize}

We evaluate CASteer on SDXL with $\beta=2$ under these settings for erasing the concept of angriness (Fig.~\ref{fig:layers_forward_remove_angry}, ~\ref{fig:layers_backward_remove_angry}, ~\ref{fig:layers_k_remove_angry}). In all the figures we illustrate results for values of $k$ between 0 and 60 with a step of 3 for compactness. 

Our empirical findings suggest the following: 
1) It is not sufficient to steer only one layer (see Fig. ~\ref{fig:layers_k_remove_angry}). It can be seen that the effect of steering any single layer is negligible, not causing the desired effect. 
2) There is a trade-off between the level of expression of the desired concept in a resulting image and the alteration of general image layout and features. If we steer most of the layers, the overall layout may change drastically from that of the original image and it may cause in the change of identity or other features in the steered image compared to the original one (see Fig.~\ref{fig:layers_forward_remove_angry}, ~\ref{fig:layers_backward_remove_angry}). As Fig.~\ref{fig:layers_backward_remove_angry} suggests, steering only few last CA layers of the model results in a good trade-off between removing the unwanted concept and keeping other image details intact. 

Based on our observations, we hypothesize that steering only few last CA layers of the model is sufficient for erasing texture-based concepts (e.g. styles), while not sufficient for erasing concrete concepts (e.g. ``Snoopy''). We test this experimentally by applying CASteer to only 36 last layers of SDXL model. We test erasing ``Snoopy'' concept and style of ``Van Gogh''. Results are presented in Tab.~\ref{tab:up_snoopy},~\ref{tab:up_art}. We see that steering only 36 last layers of SDXL model successfully removes style of ``Van Gogh'', while changing images of non-target styles less (see $LPIPS_u$). For the concept of ``Snoopy'', CASteer applied to only the last 36 layers of the model shows significantly worse erasure performance than CASteer applied to all layers.

\begin{table}[ht]
\caption{\textbf{Quantitative evaluation of concrete object erasure.}}
\label{tab:up_snoopy}
\centering
\setlength{\tabcolsep}{2.0pt}
\resizebox{0.65\textwidth}{!}{
\definecolor{mygray}{gray}{.9}
\begin{tabular}{l|c|cc|cc|cc|cc|cc}
    \toprule
    
    & \multicolumn{1}{c|}{Snoopy} & \multicolumn{2}{c|}{Mickey} & \multicolumn{2}{c|}{Spongebob} & \multicolumn{2}{c|}{Pikachu} & \multicolumn{2}{c|}{Dog} & \multicolumn{2}{c}{Legislator}  \\

    \cmidrule{2-12}
    Method & CS$\downarrow$ &  CS$\uparrow$ & FID$\downarrow$ & CS$\uparrow$ & FID$\downarrow$ & CS$\uparrow$ & FID$\downarrow$ & CS$\uparrow$ & FID$\downarrow$ & CS$\uparrow$ & FID$\downarrow$ \\
    \midrule
    SDXL & 74.3 & 73.1 & - & 75.1 & - & 72.7 & - & 66.3 & - & 60.8 \\
    \midrule
    CASteer (last 36 layers)  & 57.2  & 70.8 & 44.0 & 74.5 & 35.2 & 72.7 & 19.2 & 66.3 & 15.0 & 60.8 & 17.3 \\
    CASteer  & 48.7  & 68.5 & 69.0 & 73.7 & 58.4 & 72.7 & 27.8 & 66.4 & 37.9 & 60.8 & 27.4 \\

    \bottomrule
\end{tabular}
}
\end{table}

\begin{table}
\caption{{Comparison of "Van Gogh" Removal tasks on SDXL model}.}
\label{tab:up_art}
\small
\centering
\setlength{\tabcolsep}{1.mm}
\resizebox{0.4\textwidth}{!}{
\begin{tabular}{l|cccc}
\toprule
&\multicolumn{4}{c}{\textbf{Remove ``Van Gogh"}}\\ 
\cmidrule(lr){2-5}
\textbf{Method}&
\textbf{LPIPS}$_e \uparrow$ &\textbf{LPIPS}$_u \downarrow$ & \textbf{Acc}$_e \downarrow$& \textbf{Acc}$_u \uparrow$\\
\midrule
SDXL&-&-&0.95&0.95\\
\midrule
Ours (last 36 layers) & 0.30 & 0.12 &0.00 & 0.85  \\
Ours & 0.34 & 0.19 &0.00 & 0.84  \\

\bottomrule
\end{tabular}
}
\vspace{-0.5cm}
\end{table}

\begin{figure*}[t]
    \centering
    \includegraphics[width=\linewidth]{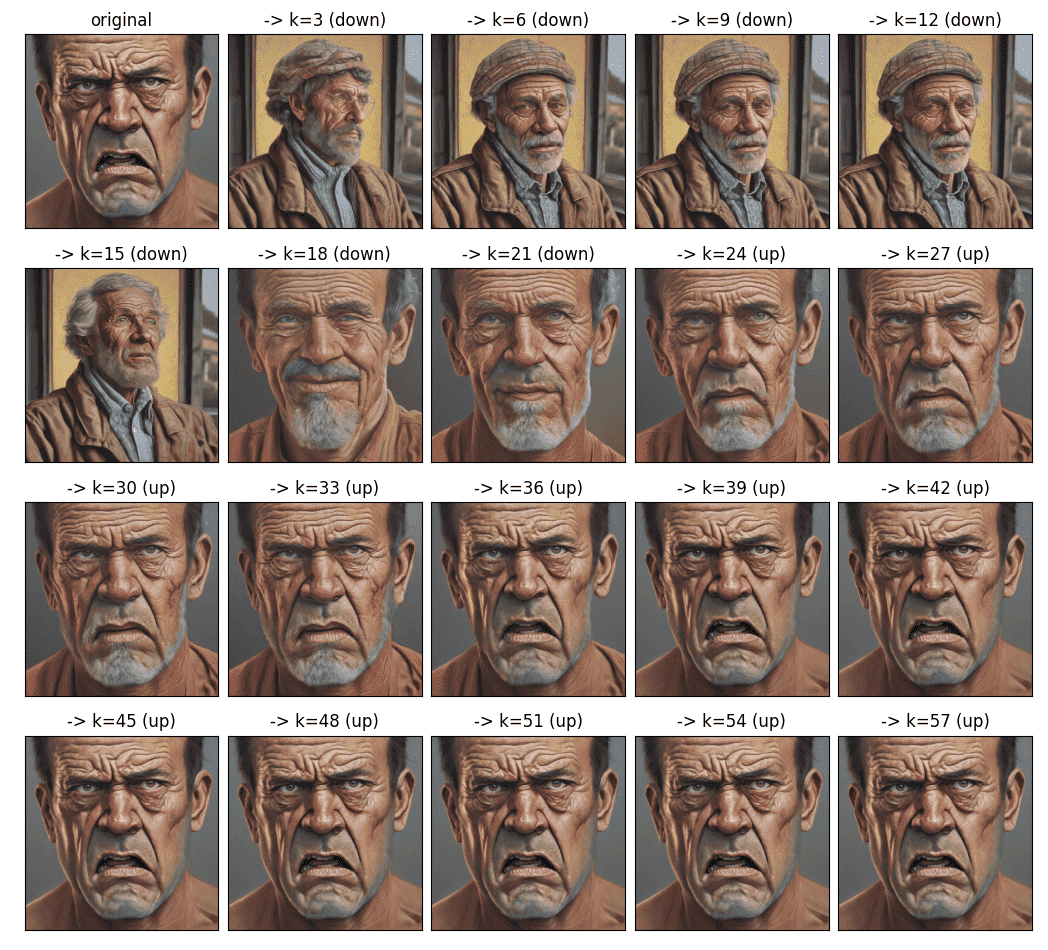}
    \caption{Results on removing the concept of ``angry'' using CASteer on the prompt ``a realistic colorful portrait of an angry man'' with steering only $k$ last CA layers, $60 \geqslant k \geqslant 1$. Top left corner: image generated without CASteer, images from top to bottom, left to right: images generated using CASteer with varying $k$}
    \label{fig:layers_forward_remove_angry}
\end{figure*}
\newpage
\begin{figure*}[t]
    \centering
    \includegraphics[width=\linewidth]{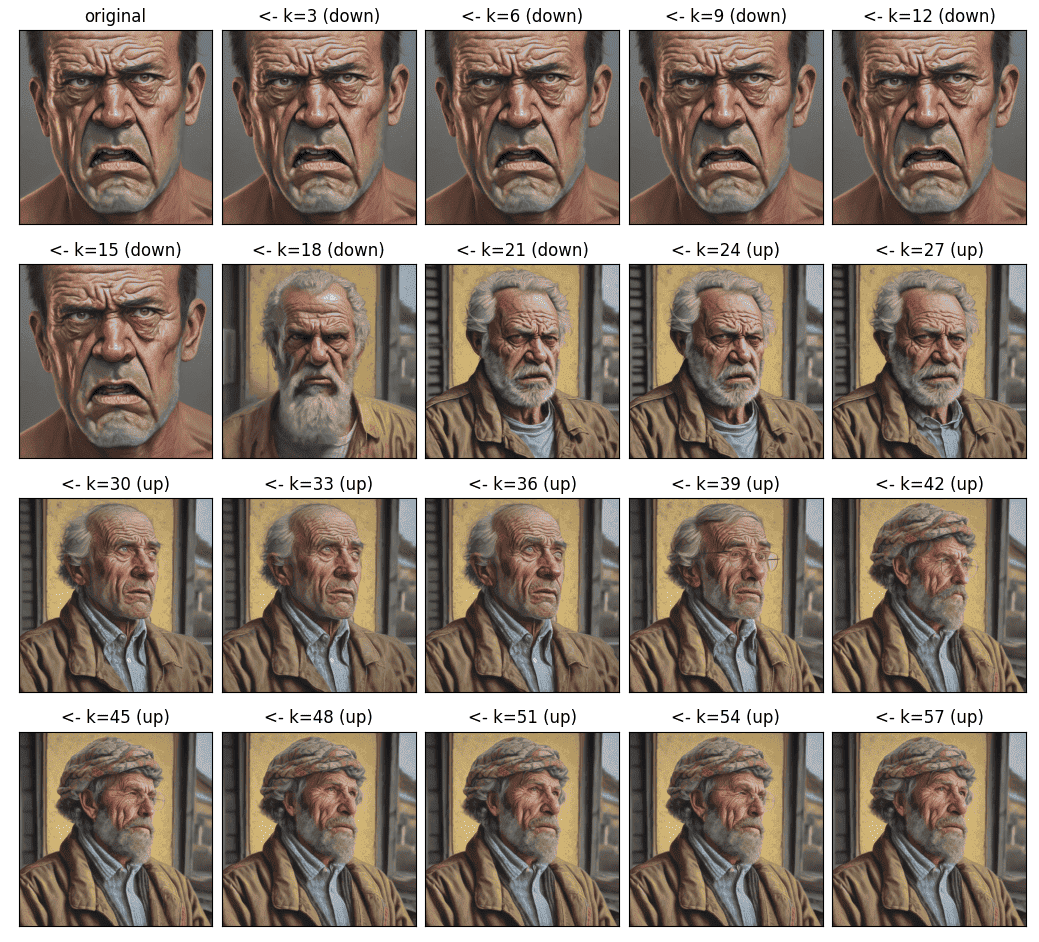}
    \caption{Results on removing the concept of ``angry'' using CASteer on the prompt ``a realistic colorful portrait of an angry man'' with steering only $k$ first CA layers, $1 \leqslant k \leqslant 60$. Top left corner: image generated without CASteer, images from top to bottom, left to right: images generated using CASteer with varying $k$}
    \label{fig:layers_backward_remove_angry}
\end{figure*}
\newpage
\begin{figure*}[t]
    \centering
    \includegraphics[width=\linewidth]{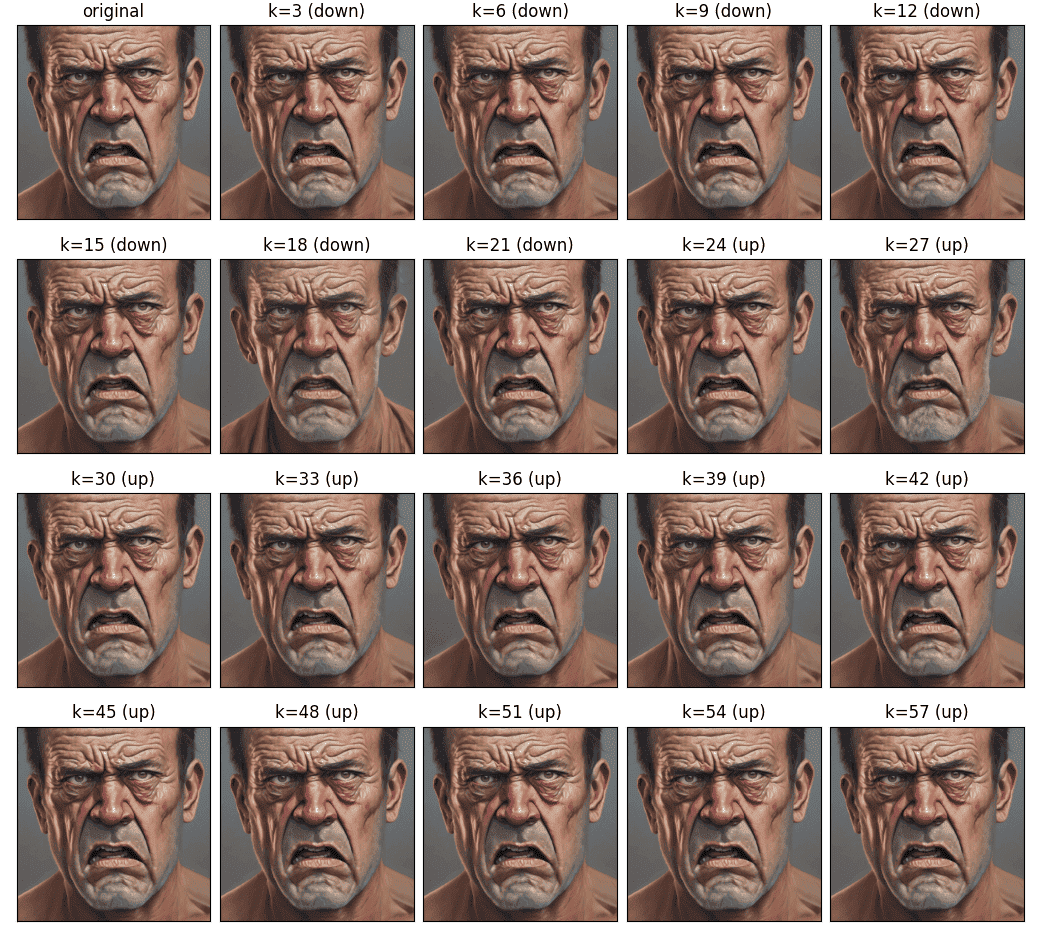}
    \caption{Results on removing the concept of ``angry'' using CASteer on the prompt ``a realistic colorful portrait of an angry man'' with steering only CA layer number $k$, $1 \leqslant k \leqslant 60$. Top left corner: image generated without CASteer, images from top to bottom, left to right: images generated using CASteer with varying $k$}
    \label{fig:layers_k_remove_angry}
\end{figure*}

\clearpage

\section{Other tasks}

In this section, we provide evidence that using steering vectors calculated by CASteer, it is possible not only to erase concepts from generated images, but to solve other tasks, such as \textit{concept addition}, \textit{concept flipping} and \textit{concept interpolation}. We provide some quantitative and qualitative results on these tasks. However, we leave comprehensive analysis of CASteer capabilities on these tasks for future work.

\subsection{Concept switch}

In this section, we propose a way to modify CASteer to flip one concept on the image being generated into another, i.e. resulting generating concept $Y$ when prompted to generate concept $X$. 

\subsubsection{Method}
Recall from Sec.~\ref{sec:method} that the value of the dot product $(ca^{X}_{it} \cdot ca^{out}_{itk})$ can be seen as amount of $X$ that is present in $ca^{out}_{itk}$. Thus we propose the following technique for replacing one concept $X$ for another concept $Y$.
We first construct steering vectors $ca^{XY}_{it}$ from $X$ to $Y$ using the same idea described in Sec.~\ref{sec:method}, with the only difference that the positive prompt contains $X$ (``a girl with an apple'') while the negative one contains $Y$ (``a girl with pear'').
Then we use this steering vector to modify cross-attention outputs as:
\begin{equation}
\label{eq:flipping}
ca^{out\_new}_{itk} = ca^{out}_{itk}- 2 \langle ca^{XY}_{it},  ca^{out}_{itk} \rangle ca^{XY}_{it},
\end{equation}
\noindent where
$1 \leqslant k \leqslant \text{patch\_num}_i$. This operation can be seen as removing all the information about concept $X$ from $ca^{out}_{itk}$ and instead adding the same amount of information about concept $Y$, that is, we reflect $ca^{out}_{itk}$ relatively to the vector orthogonal to $ca^{XY}_{it}$.
Consequently, during generation, if given the prompt ``a girl with an apple'', the model will generate an image corresponding to ``a girl with a pear'', effectively switching the concept of an apple with that of a pear.

Here we can also add steering strength by introducing parameter $\beta$:
\begin{equation}
\label{eq:flipping_beta}
ca^{out\_new}_{itk} = ca^{out}_{itk}- \beta \langle ca^{XY}_{it},  ca^{out}_{itk} \rangle ca^{XY}_{it},
\end{equation}
Higher $\beta$ will result in higher expression of concept $Y$ in the resulting image. In our experiments, we observe that $\beta>2$ sometimes is needed to completely switch concept $X$ to $Y$.

\noindent \textbf{Adding Intermediate clipping}. Note that using Eq.\ref{eq:flipping} we only want to influence those CA outputs $ca^{out}_{itk}$ which have a positive amount of unwanted concept $X$ in them, i.e. we only want to flip $X$ to $Y$, and not on the opposite direction. As dot product $ \langle ca^{X}_{it}, ca^{out}_{itk} \rangle$ measures the amount of $X$ present in CA output $ca^{out}_{itk}$, we only want to steer those CA outputs $ca^{out}_{itk}$, which have a positive dot product with $ca^{X}_{it}$. So the equation becomes the following:
\begin{equation}
\begin{split}
\label{eq:flipping_clip}
    \alpha = \max(\beta \langle ca^{XY}_{it}, ca^{out}_{itk}\rangle, 0) \\
    ca^{out\_new}_{itk} = ca^{out}_{itk}- \beta \alpha ca^{XY}_{it}
\end{split}
\end{equation}  



\subsubsection{Qualitative results}

Here, we provide qualitative results on flipping different concepts. Fig.~\ref{fig:layers_identity_change_snoopy_winnie},~\ref{fig:layers_identity_change_cat_giraffe},~\ref{fig:layers_identity_change_caterpillar_butterfly} visualize results of switching concepts of ``Snoopy'' to ``Winnie-the-Pooh'', ``cat'' to ``giraffe'' and ``caterpillar'' to ``butterfly''. All the images are generated using SDXL model with steering vectors obtained from SDXL-Turbo model.

\begin{figure*}[t]
    \centering
    \includegraphics[width=0.9\linewidth]{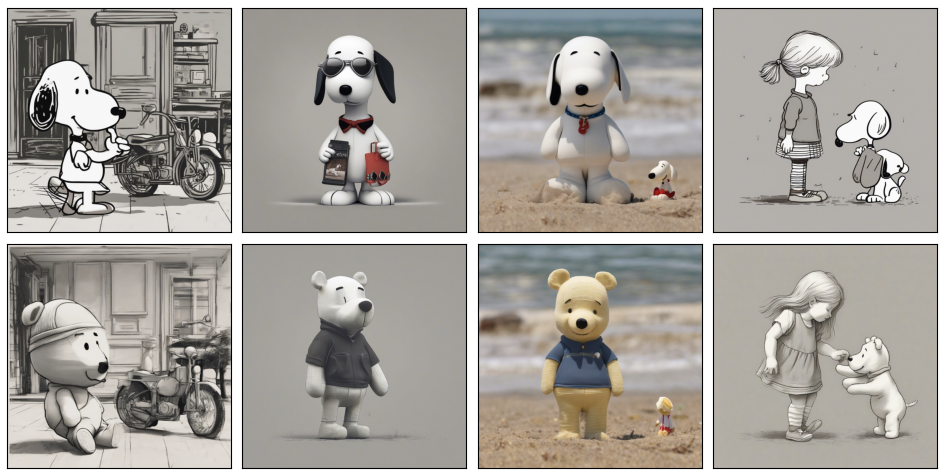}
    \caption{Switching between concepts of ``Snoopy'' and ``Winnie-the-Pooh'' using CASteer. Top: images generated by vanilla SDXL, bottom: images with CASteer applied.}
    \label{fig:layers_identity_change_snoopy_winnie}
\end{figure*}

\begin{figure*}[t]
    \centering
    \includegraphics[width=0.9\linewidth]{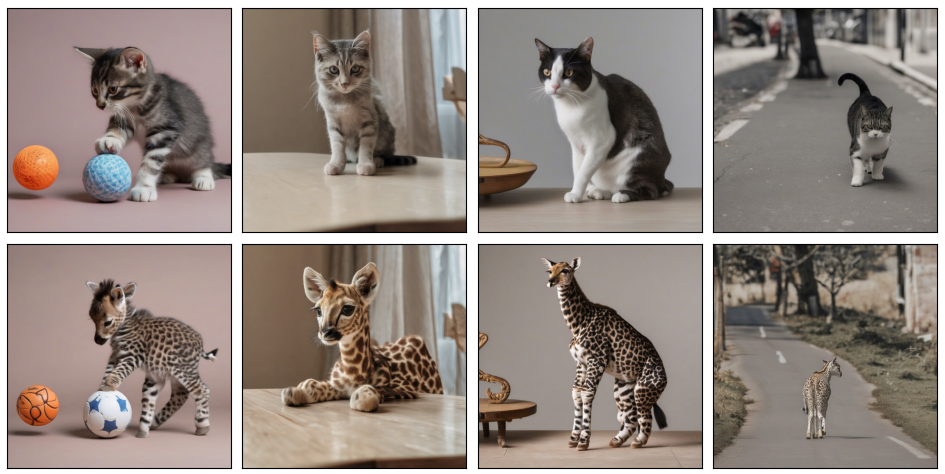}
    \caption{Switching between concepts of ``cat'' and ``giraffe'' using CASteer. Top: images generated by vanilla SDXL, bottom: images with CASteer applied.}
    \label{fig:layers_identity_change_cat_giraffe}
\end{figure*}

\begin{figure*}[t]
    \centering
    \includegraphics[width=0.9\linewidth]{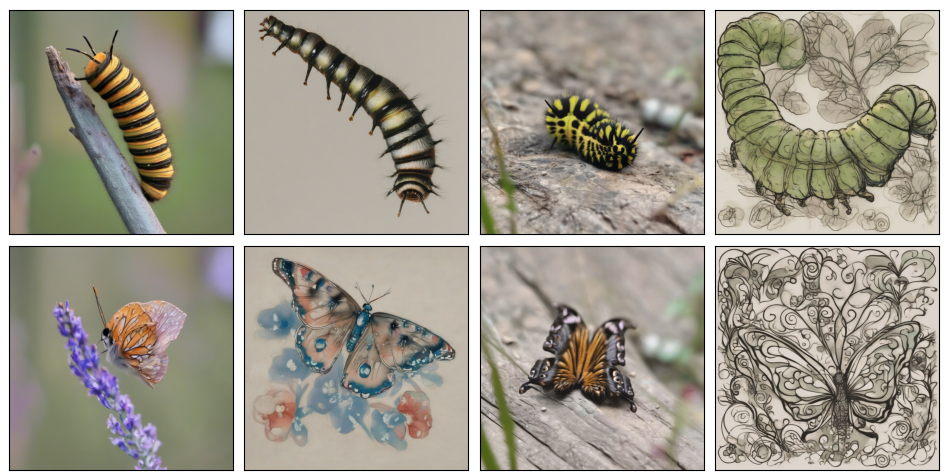}
    \caption{Switching between concepts of ``caterpillar'' and ``butterfly'' using CASteer. Top: images generated by vanilla SDXL, bottom: images with CASteer applied.}
    \label{fig:layers_identity_change_caterpillar_butterfly}
\end{figure*}

\subsection{Concept addition}

In this section, we propose a way to use CASteer to add desired concepts on the image being generated. 

\subsubsection{Method}

Recall Eq.~\ref{eq:4}, which defined a mechanism of subtracting concept information from CA outputs of diffusion model to prevent generation of some concepts. We can use the same mechanism to instead add concept to the outputs:
\begin{equation}
ca^{out\_new}_{itk} = ca^{out}_{itk} + \alpha ca^{X}_{itk},
\label{eq:addition}
\end{equation}
\noindent Here $1 \leqslant k \leqslant \text{patch\_num}_i$, and $\alpha$ is a hyperparameter that controls the strength of concept addition.

\subsubsection{Qualitative results}

Here, we provide qualitative results on adding different concepts. 

Fig.~\ref{fig:layers_identity_change_snoopy_winnie},~\ref{fig:layers_identity_change_cat_giraffe},~\ref{fig:layers_identity_change_caterpillar_butterfly} visualize results of adding concepts of ``apple'', ``hat'', ``clothes'' and ``happiness'' using CASteer. Note that for different concepts, different values of $\alpha$ are optimal.

\begin{figure*}[t]
    \centering
    \includegraphics[width=0.9\linewidth]{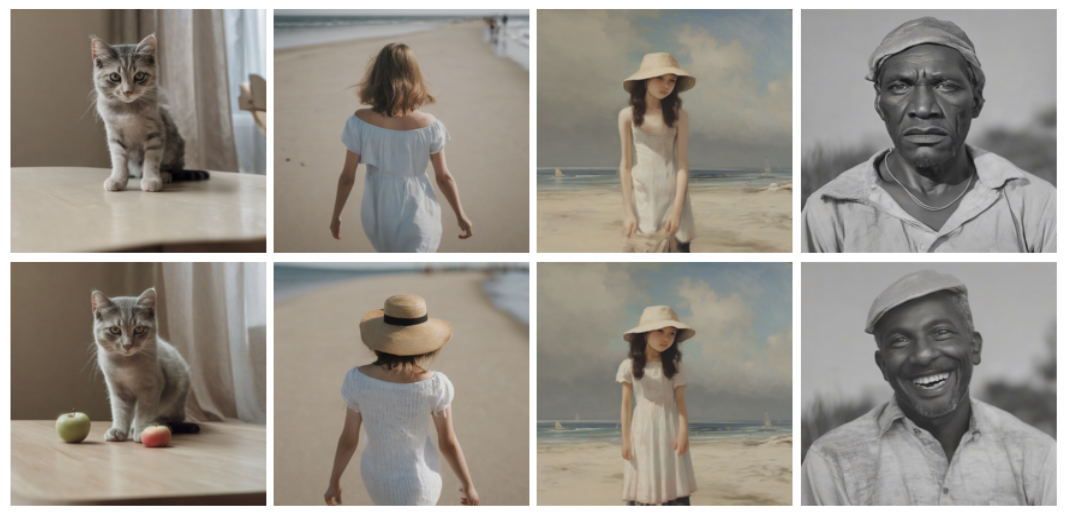}
    \caption{Adding concepts of ``apple'', ``hat'', ``clothes'' and ``happiness'' using CASteer. Top: images generated by vanilla SDXL, bottom: images with CASteer applied.}
    \label{fig:concept_addition}
\end{figure*}

\subsection{Style Transfer}
\label{sec:style_transfer}

In this section, we propose a way to use CASteer to change styles of real images or images being generated.

\subsubsection{Method}

CASteer performs Style Transfer on real images as follows: we apply the reverse diffusion process following DDIM~\cite{DBLP:conf/iclr/SongME21} for $t$ number of steps, where $1 \leqslant t \leqslant T$. Then we denoise image back using CASteer (addition algorithm, see Eq.~\ref{eq:addition}). $t$ controls the trade-off between loss of image details and intensity of style applied. In Fig.~\ref{fig:st_anime}, ~\ref{fig:st_origami}, ~\ref{fig:st_gothic}, \ref{fig:st_retro} we show results on style transfer to four different styles with varying $t$. Intensity$=0.3$ here means that $t=0.3T$. 
With such a process, we often get satisfactory results with no major loss of details. However, when $t$ is high, the loss of image content occurs (see the bottom lines of figures). 
This is because the inversion is not sufficiently accurate. 

\begin{figure*}[t]
    \centering
    \includegraphics[width=\linewidth]{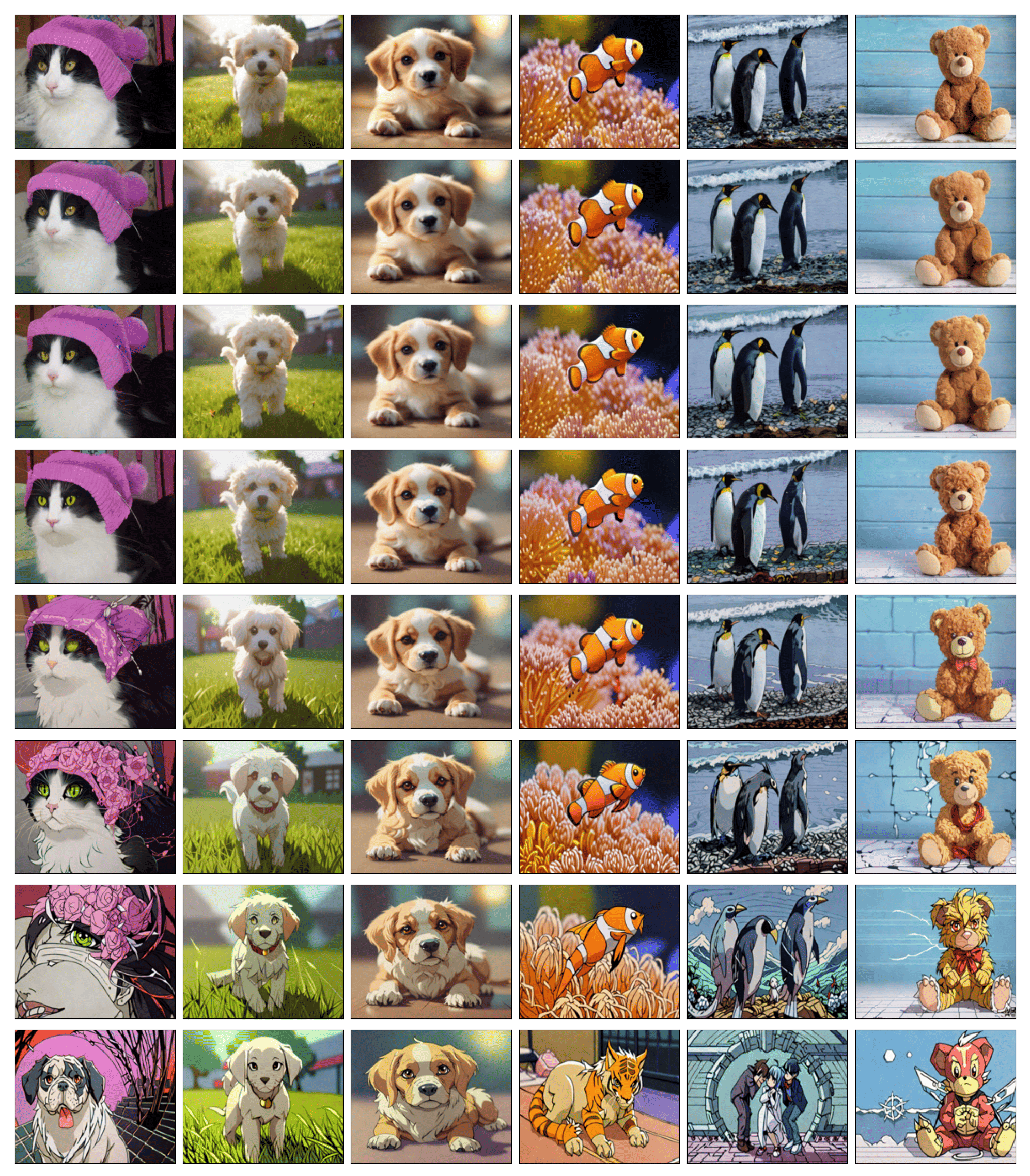}
    \caption{Examples of Style Transfer of real images into ``\textit{anime}'' style. From top to bottom: original image, style transfer applied with intensities from 0.1 to 0.7 with a step of 0.1}
    \label{fig:st_anime}
\end{figure*}

\begin{figure*}[t]
    \centering
    \includegraphics[width=\linewidth]{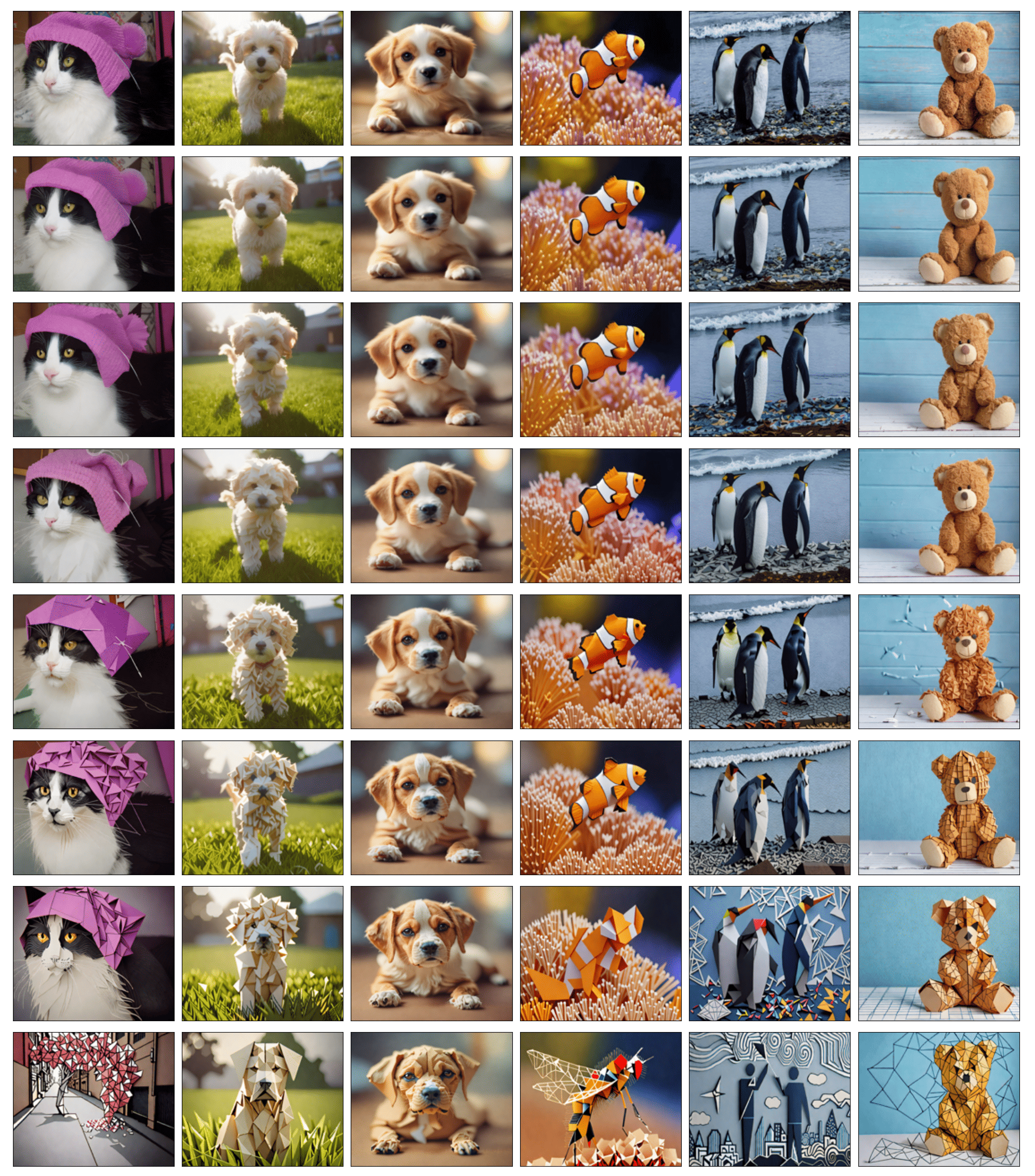}
    \caption{Examples of Style Transfer of real images into ``\textit{origami}'' style. From top to bottom: original image, style transfer applied with intensities from 0.1 to 0.7 with a step of 0.1}
    \label{fig:st_origami}
\end{figure*}

\begin{figure*}[t]
    \centering
    \includegraphics[width=\linewidth]{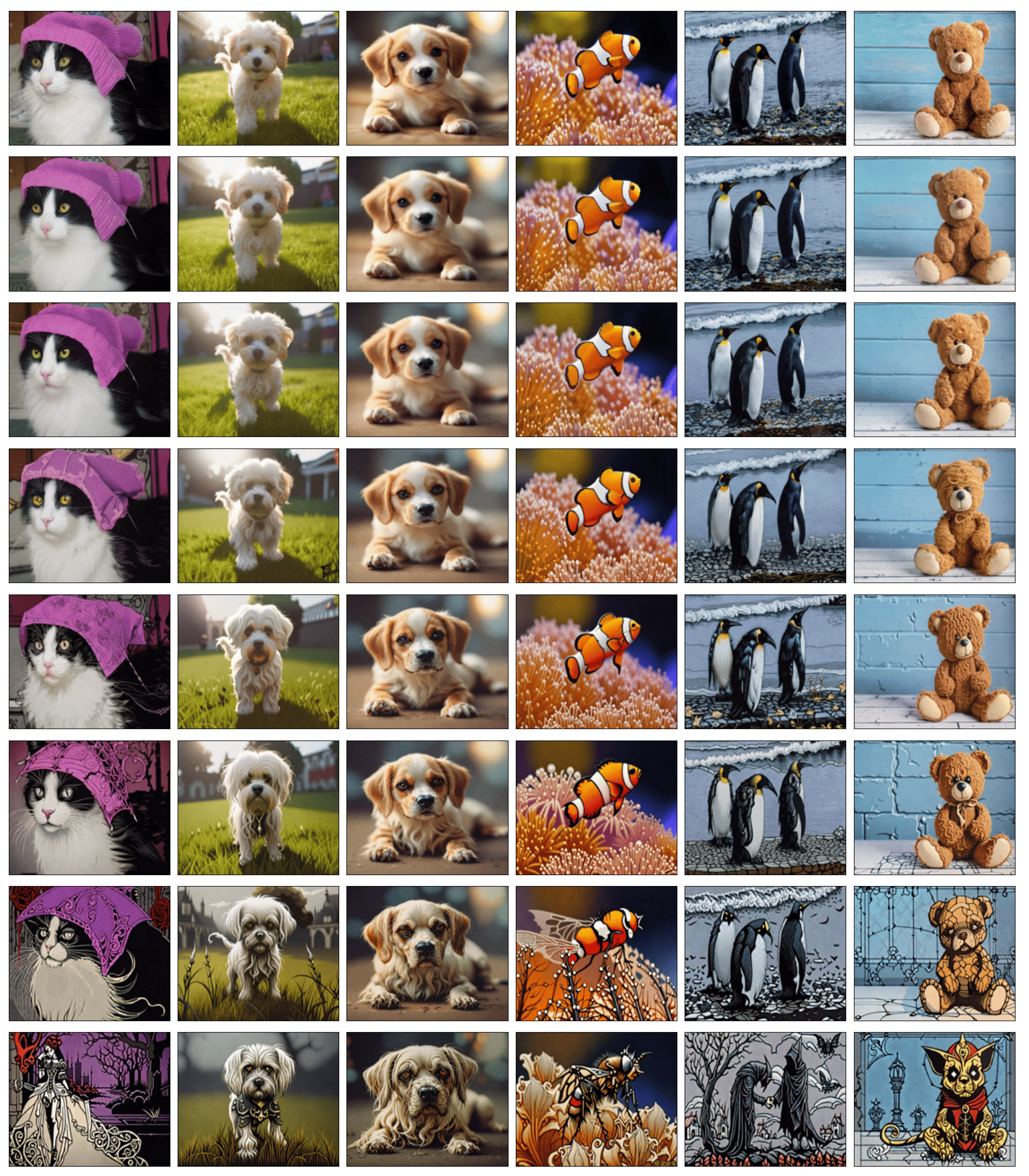}
    \caption{Examples of Style Transfer of real images into ``\textit{Gothic Art}'' style. From top to bottom: original image, style transfer applied with intensities from 0.1 to 0.7 with a step of 0.1}
    \label{fig:st_gothic}
\end{figure*}

\begin{figure*}[t]
    \centering
    \includegraphics[width=\linewidth]{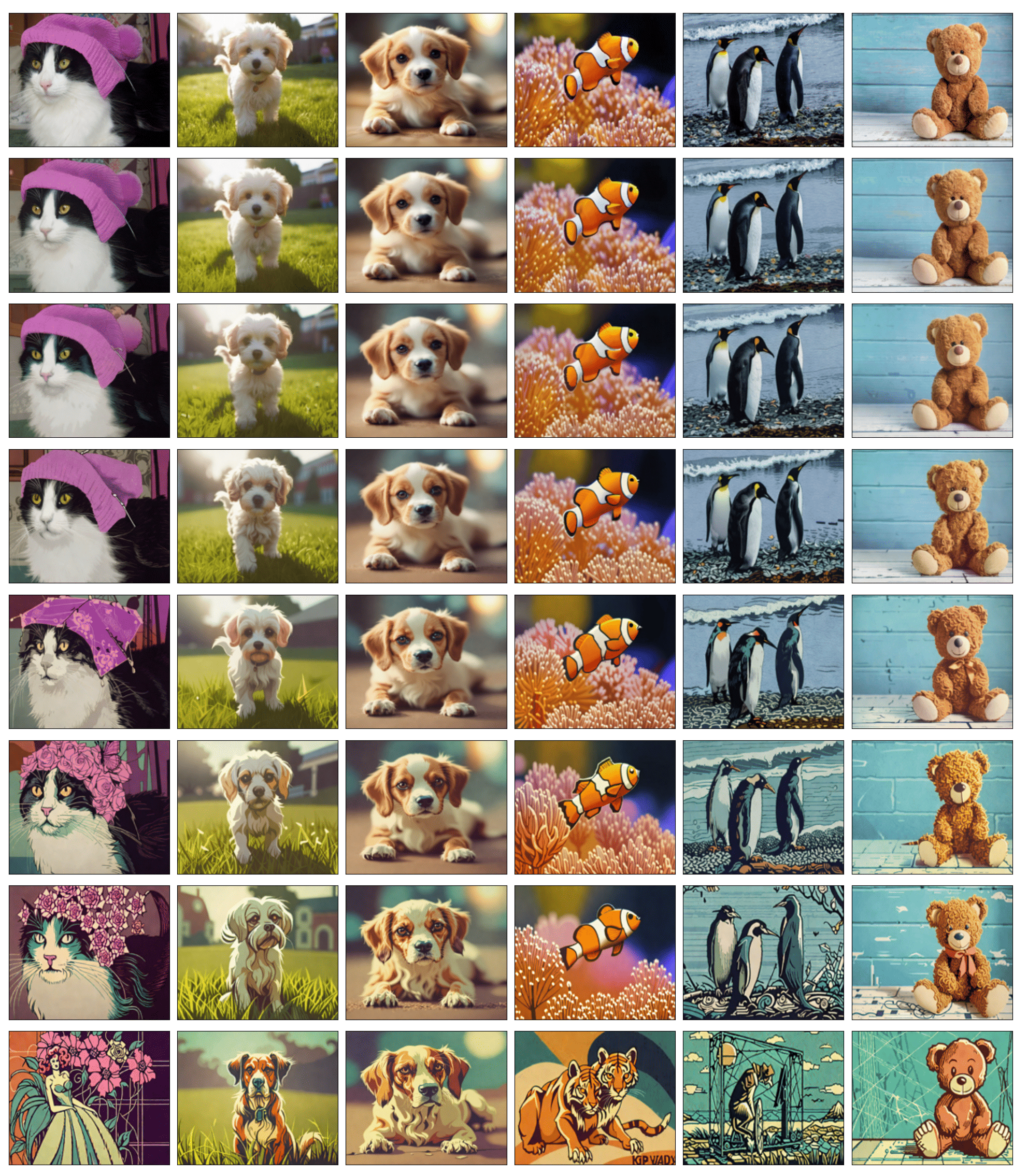}
    \caption{Examples of Style Transfer of real images into ``\textit{Retro Art}'' style. From top to bottom: original image, style transfer applied with intensities from 0.1 to 0.7 with a step of 0.1}
    \label{fig:st_retro}
\end{figure*}

\clearpage
\section{Interpreting steering vectors}
\label{sec:interpretation}

\subsection{Steering vectors visualization}

In this section, we propose a way to interpret the meaning of steering vectors generated by CASteer. Suppose we have steering vectors generated for a concept $\mathbf{X}$ $\{ca^{\mathbf{X}}_{it}\}, 1 \leqslant i \leqslant n, 1 \leqslant t \leqslant T$, where $n$ is the number of model layers and $T$ is a number of denoising steps performed for generating steering vectors. To interpret these vectors, we prompt the diffusion model with a placeholder prompt ``X'' and at each denoising step, we substitute outputs of the model's CA layers with corresponding steering vectors. This makes the diffusion model be only conditioned on the information from the steering vectors, completely suppressing other information from the text prompt.

Fig.~\ref{fig:interpretation} shows interpretations for the steering vectors of concepts ``\textit{hat}'', ``\textit{polka dot dress}'', ``\textit{Snoopy}'', ``\textit{angry}'', ``\textit{happy}'', from top to bottom. Note that vectors for concepts ``\textit{hat}'', ``\textit{polka dot dress}'' and ``\textit{Snoopy}'' were generated using prompt pair templates for concept deletion, i.e. pairs of the form (``\textit{fish with Snoopy}'', ``\textit{fish}''), (``\textit{a girl with a hat}'', ``\textit{a girl}'') (see Sec. ~\ref{sec:prompts_erasure}), and this is reflected in generated images, as they show these concepts not alone, but in a form of a girl in a hat, a girl in a polka dot dress or a boy with a Snoopy. As for the last two concepts (``\textit{angry}'', ``\textit{happy}''), their steering vectors were generated using prompt pair templates for human-related concepts (see Sec. ~\ref{sec:prompts_human}), and they illustrate these concepts as is.

We note that images of each concept exhibit common features, e.g. all the images of hats and polka dots feature only female persons, and images corresponding to ``angry'' and ``happy'' concepts have certain styles. We believe that this reflects how diffusion models perceive different concepts and that this interpretation technique can be used for unveiling the hidden representations of concepts inside the diffusion model, but we leave it to future work.

\begin{figure*}[t]
    \centering
    \includegraphics[width=\linewidth]{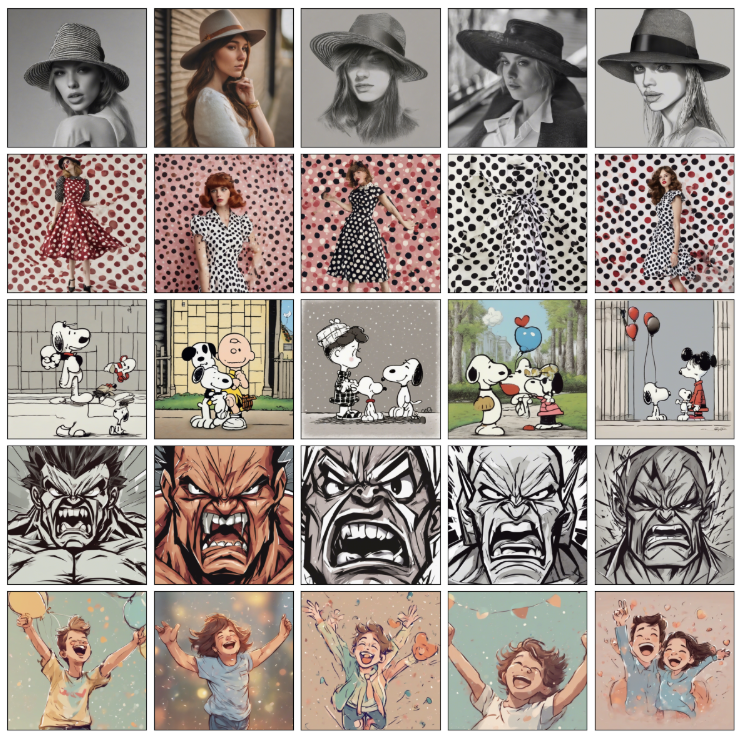}
    \caption{Visualization of generations of the model conditioned only on steering vectors. Images in rows from top to down were generated using steering vectors for the concepts ``\textit{hat}'', ``\textit{polka dot dress}'', ``\textit{Snoopy}'', ``\textit{angry}'', ``\textit{happy}''.}
    \label{fig:interpretation}
\end{figure*}

\subsection{UMap on steering vectors}
\label{sec:umap}

We  generate steering vectors for all vocabulary tokens of SDXL text encoders and apply UMap \cite{DBLP:journals/corr/abs-1802-03426} on these steering vectors. Fig.~\ref{fig:17_hello},~\ref{fig:35_hello},~\ref{fig:17_work},~\ref{fig:35_work} show that structure emerges in the space of these steering vectors, similar to that of Word2Vec \cite{DBLP:journals/corr/abs-1301-3781}, supporting the hypothesis that steering vectors carry the meaning of the desired concept. This is observed for all the layers in the model. 

\begin{figure*}[t]
    \centering
    \includegraphics[width=0.8\linewidth]{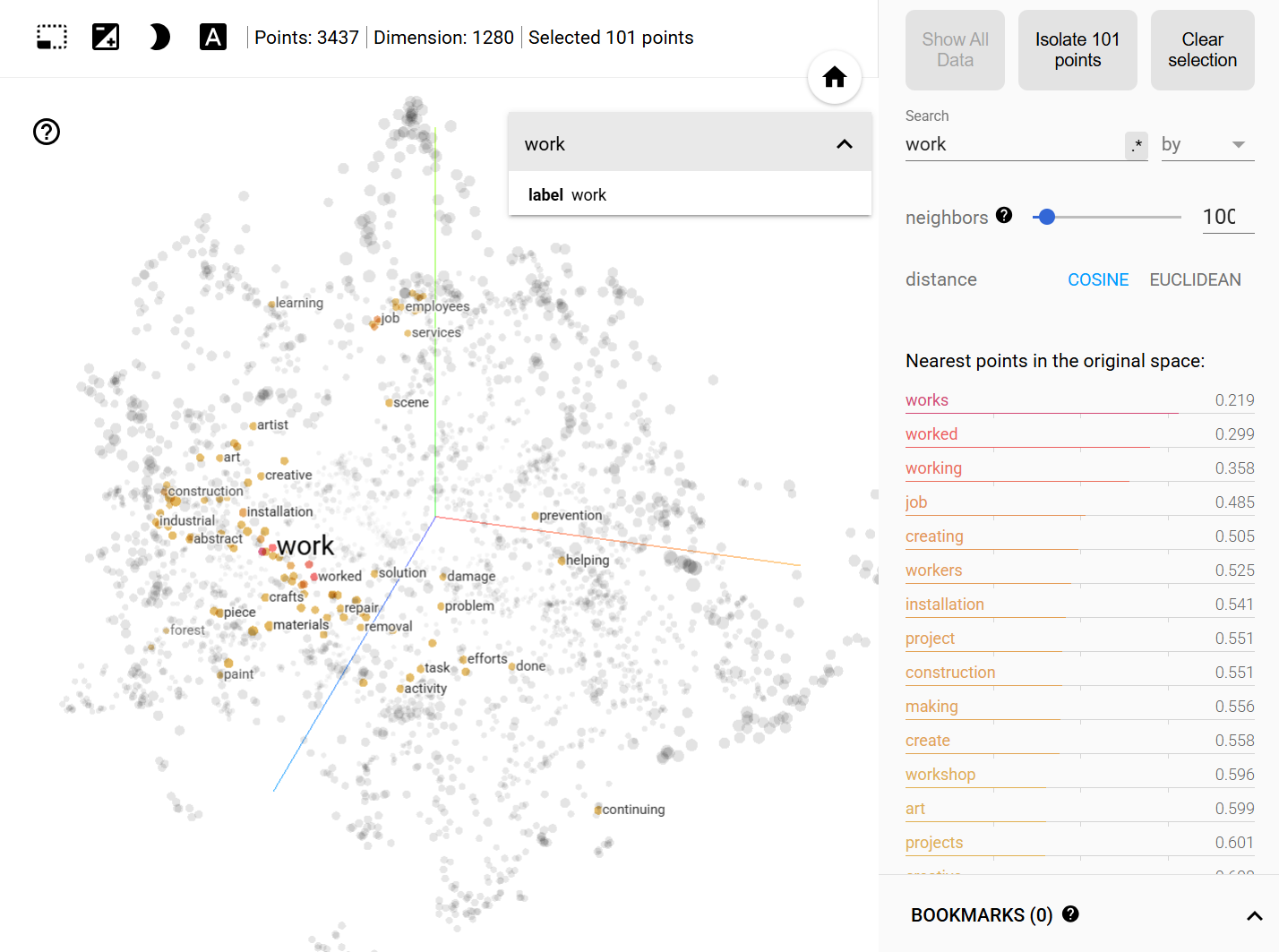}
    \caption{Visualization of UMap applied to the steering vectors from the layer 17 of SDXL-Turbo formed by 3000 SDXL-Turbo vocabulary tokens.}
    \label{fig:17_hello}
\end{figure*}

\begin{figure*}[t]
    \centering
    \includegraphics[width=0.8\linewidth]{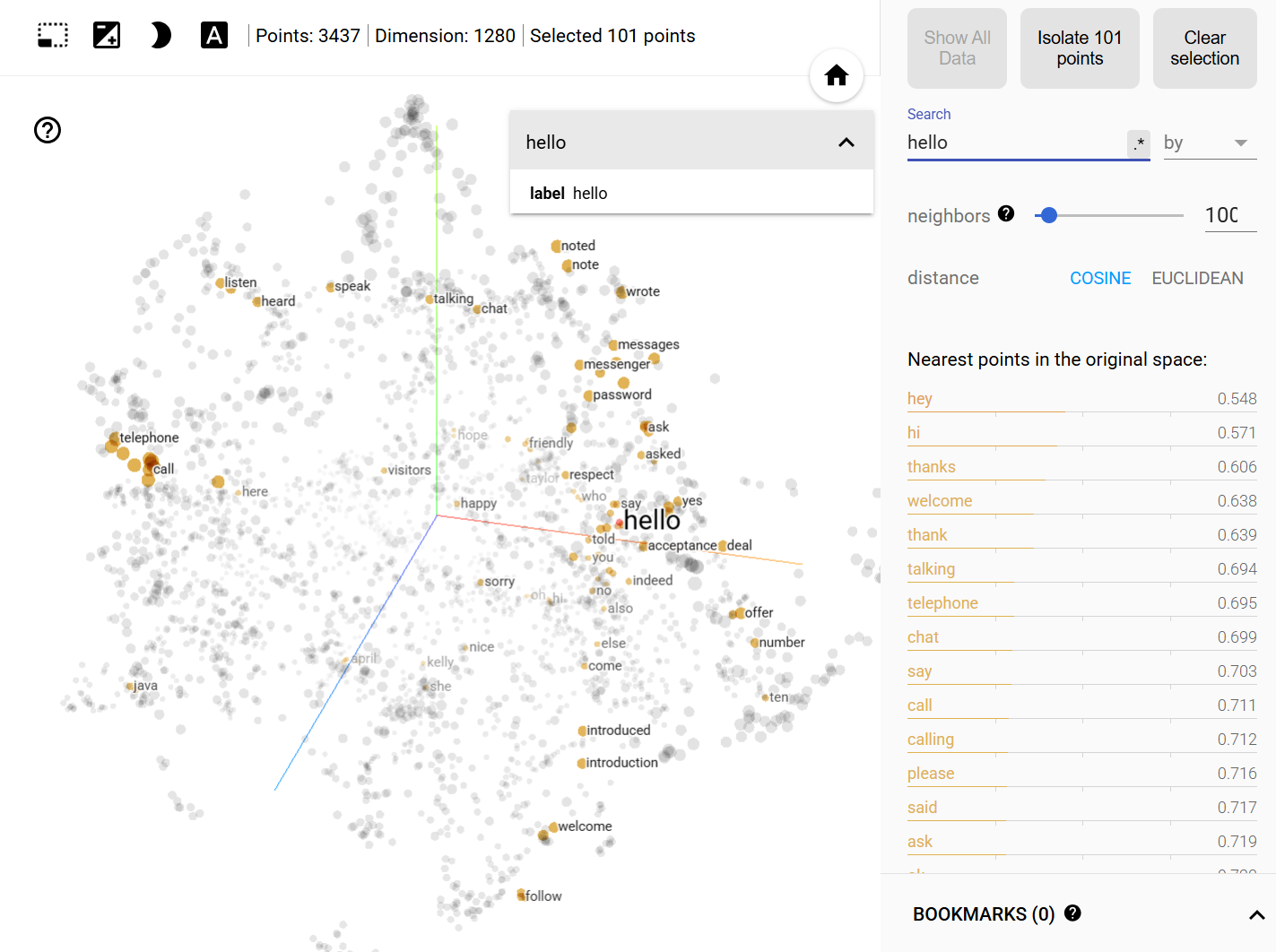}
    \caption{Visualization of UMap applied to the steering vectors from the layer 17 of SDXL-Turbo formed by 3000 SDXL-Turbo vocabulary tokens.}
    \label{fig:17_work}
\end{figure*}

\begin{figure*}[t]
    \centering
    \includegraphics[width=0.8\linewidth]{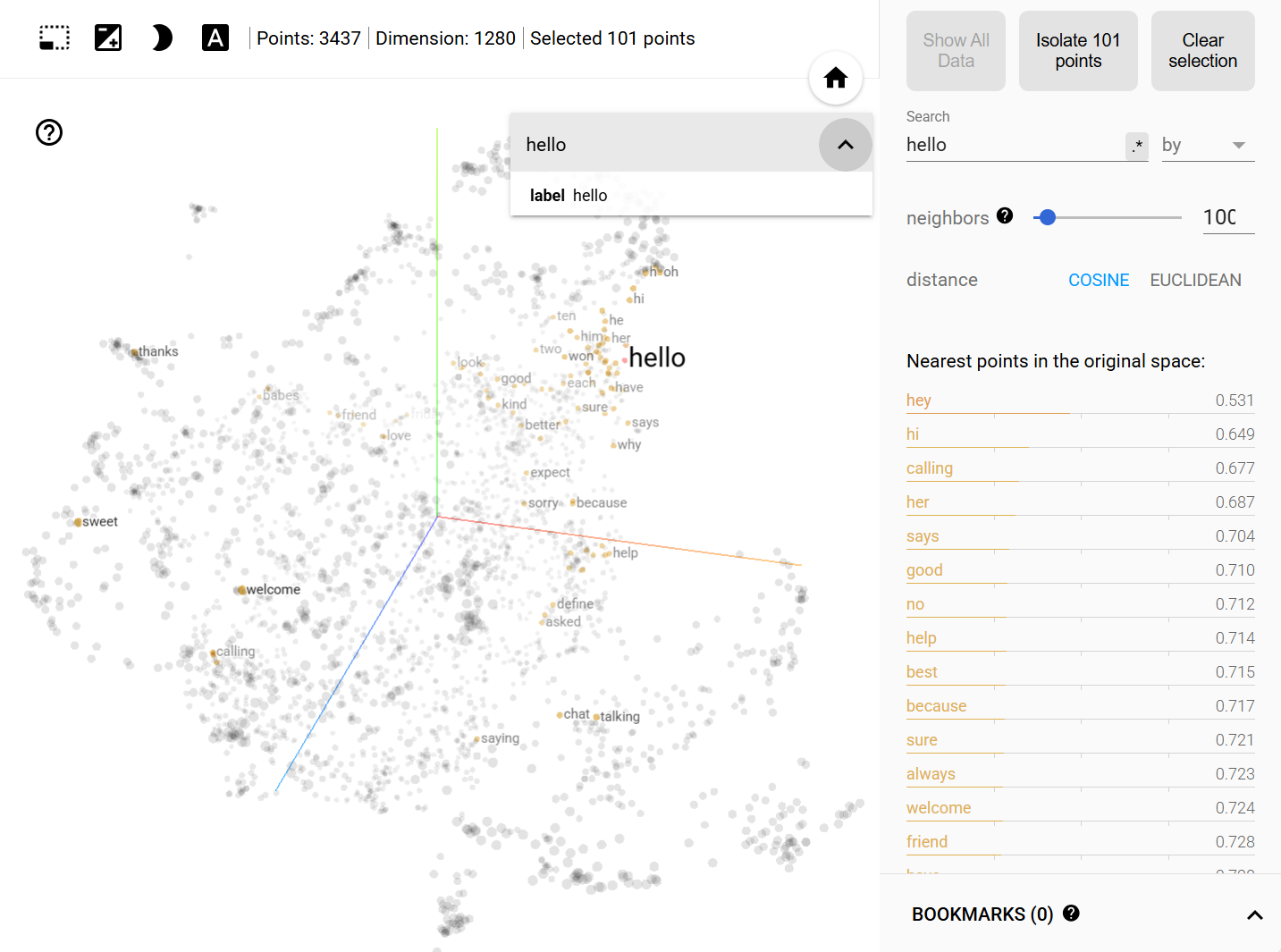}
    \caption{Visualization of UMap applied to the steering vectors from the layer 35 of SDXL-Turbo formed by 3000 SDXL-Turbo vocabulary tokens.}
    \label{fig:35_hello}
\end{figure*}

\begin{figure*}[t]
    \centering
    \includegraphics[width=0.8\linewidth]{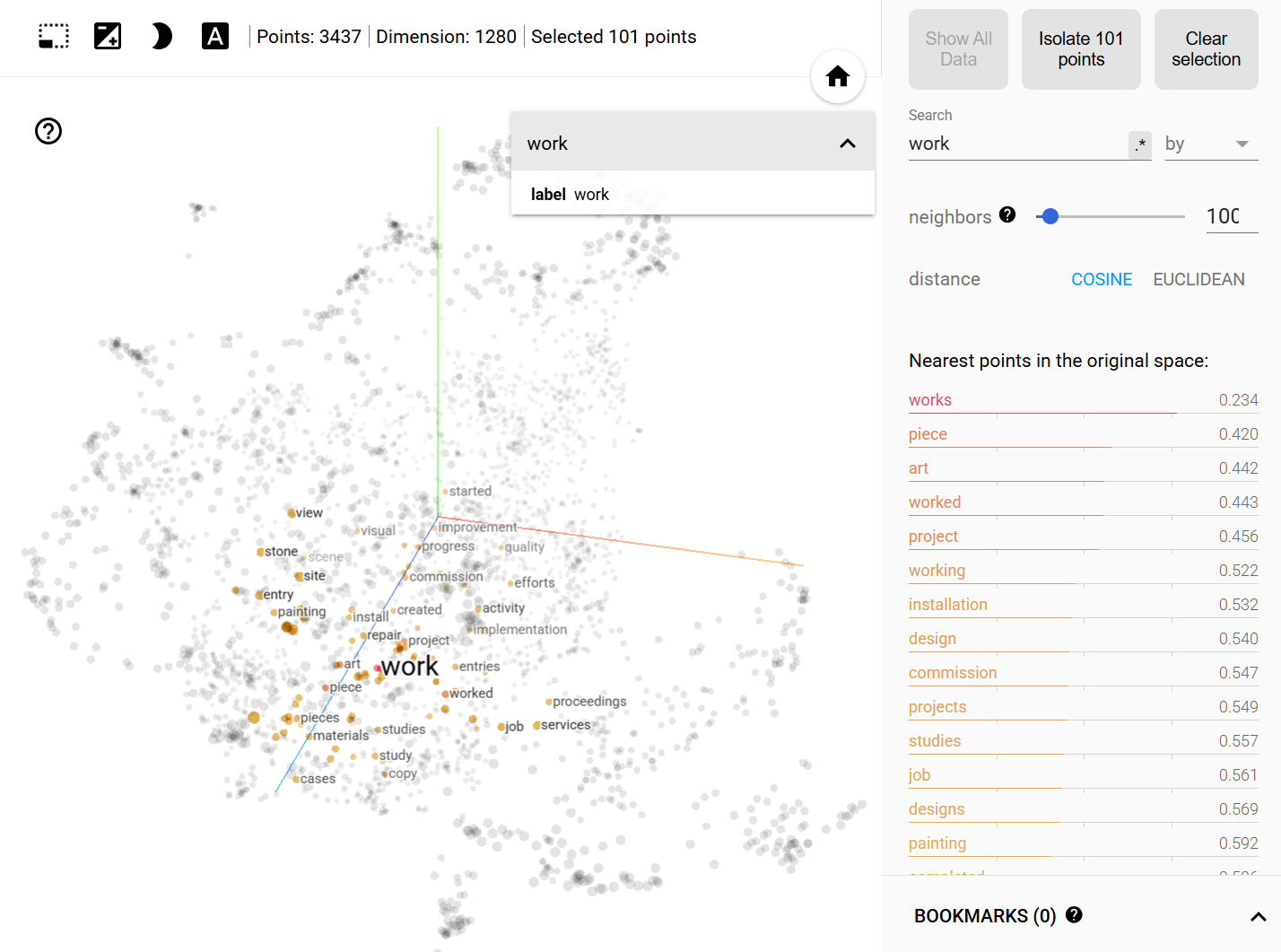}
    \caption{Visualization of UMap applied to the steering vectors from the layer 35 of SDXL-Turbo formed by 3000 SDXL-Turbo vocabulary tokens.}
    \label{fig:35_work}
\end{figure*}

\clearpage
\section{SPM and DoCo vs CASteer on adversarial prompts}
\label{sec:spm_casteer}

In this section, we present more qualitative examples of CASteer outperforming SPM~\cite{DBLP:conf/cvpr/Lyu0HCJ00HD24} and DoCo~\cite{wu2025unlearning} on adversarial prompts, i.e. prompts containing implicitly defined concepts. SD-1.4 is used as a backbone for both methods. We use prompts ``\textit{A mouse from Disneyland}'' and  ``\textit{A girl with a mouse from Disneyland}'' to test erasing of concept of ``\textit{Mickey}'', prompts ``\textit{A yellow Pokemon}'' and  ``\textit{A girl with a yellow Pokemon}'' to test erasing of concept of ``\textit{Pikachu}'', and prompt ``\textit{A picture of Iron Man actor}'' to test erasing of concept of ``\textit{Robert Downey Jr}''. We use the code provided by SPM and DoCo to produce checkpoints and images for testing removal of the concepts of ``\textit{Mickey}'' and ``\textit{Pikachu}'', and use checkpoint provided by DoCO to generate images for testing removal of a concept of ``\textit{Robert Downey Jr}''.

We see that SPM and DoCo fail to substantially remove targeted concepts in many cases, while CASteer shows consistent performance. 

\begin{figure*}[ht]
    \centering
    \includegraphics[width=\linewidth]{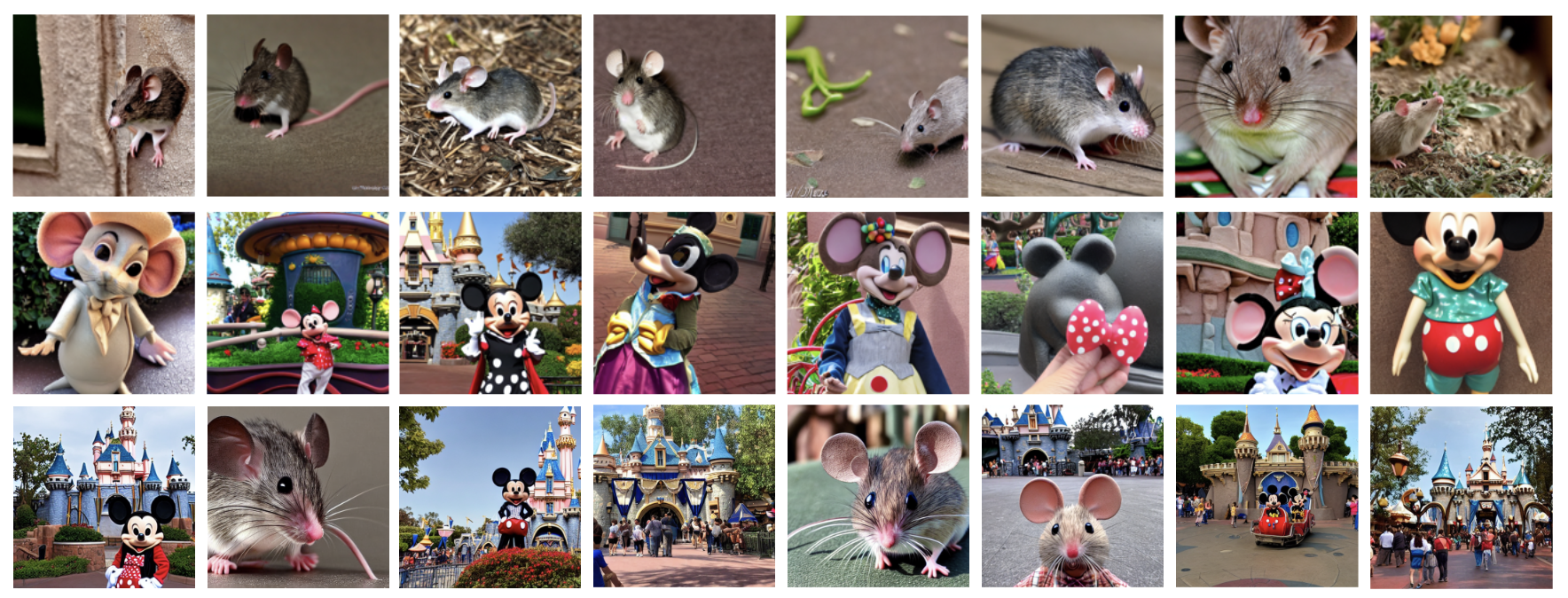}
    \caption{Examples of 8 generated images from CASteer, SPM and DoCo when prompted ``\textit{A mouse from Disneyland}''. \textbf{Top:} generation of CASteer, \textbf{Middle:} generations of SPM, \textbf{Bottom:} generations of DoCo. We use the same diffusion hyperparameters and seeds when generating corresponding images from CASteer, SPM and DoCo.}
    \label{fig:spm_fail_mouse}
\end{figure*}

\begin{figure*}[ht]
    \centering
    \includegraphics[width=\linewidth]{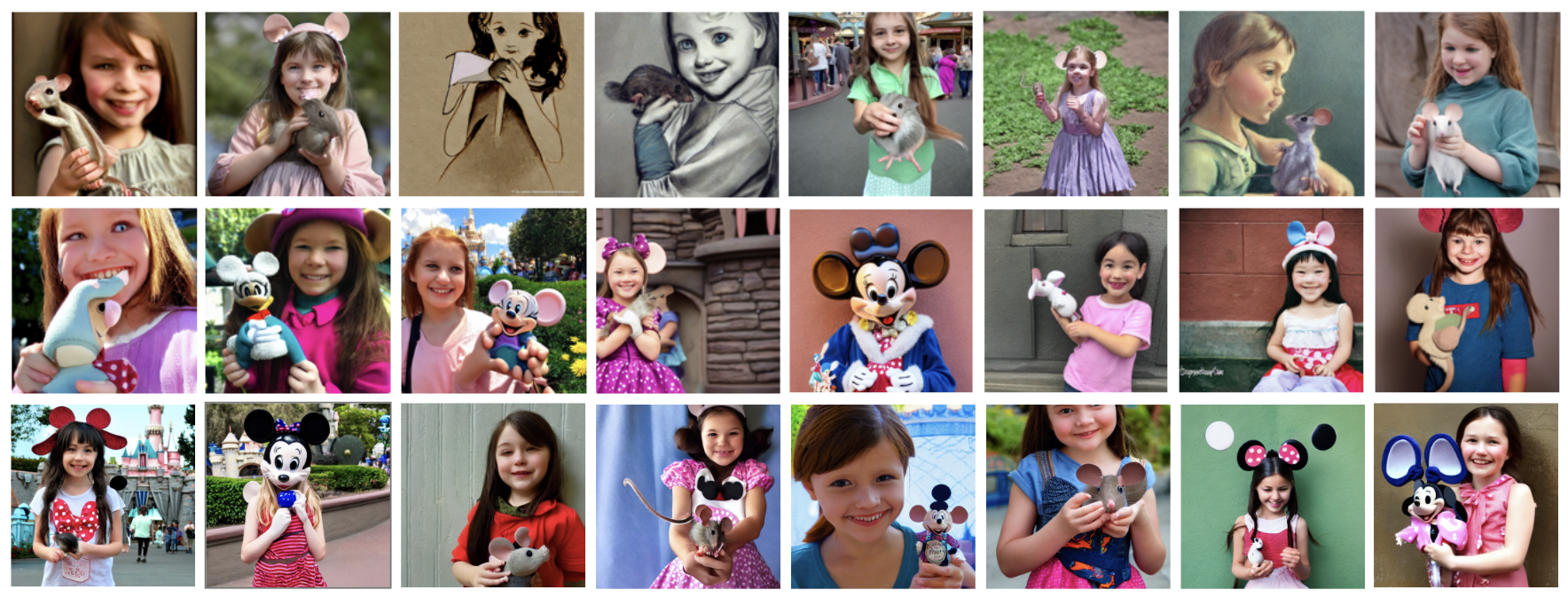}
    \caption{Examples of 8 generated images from CASteer, SPM and DoCo when prompted ``\textit{A girl with a mouse from Disneyland}''. \textbf{Top:} generation of CASteer, \textbf{Middle:} generations of SPM, \textbf{Bottom:} generations of DoCo. We use the same diffusion hyperparameters and seeds when generating corresponding images from CASteer, SPM and DoCo.}
    \label{fig:spm_fail_girl_mouse}
\end{figure*}

\begin{figure*}[ht]
    \centering
    \includegraphics[width=\linewidth]{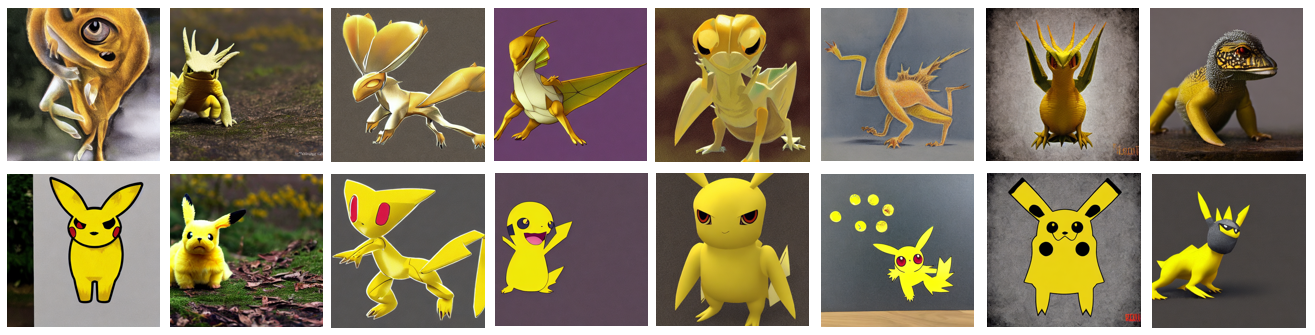}
    \caption{Examples of 8 generated images from CASteer and SPM when prompted ``\textit{A yellow Pokemon}''. \textbf{Top:} generation of CASteer, \textbf{Bottom:} generations of SPM. We use the same diffusion hyperparameters and seeds when generating corresponding images from CASteer and SPM.}
    \label{fig:spm_fail_pokemon}
\end{figure*}

\begin{figure*}[ht]
    \centering
    \includegraphics[width=\linewidth]{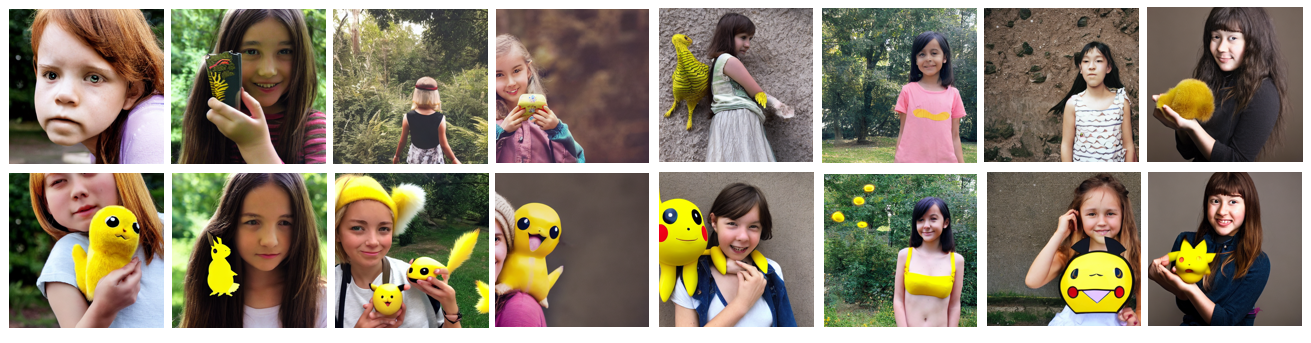}
    \caption{Examples of 8 generated images from CASteer and SPM when prompted ``\textit{A girl with a yellow Pokemon}''. \textbf{Top:} generation of CASteer, \textbf{Bottom:} generations of SPM. We use the same diffusion hyperparameters and seeds when generating corresponding images from CASteer and SPM.}
    \label{fig:spm_fail_girl_pokemon}
\end{figure*}

\begin{figure*}[ht]
    \centering
    \includegraphics[width=\linewidth]{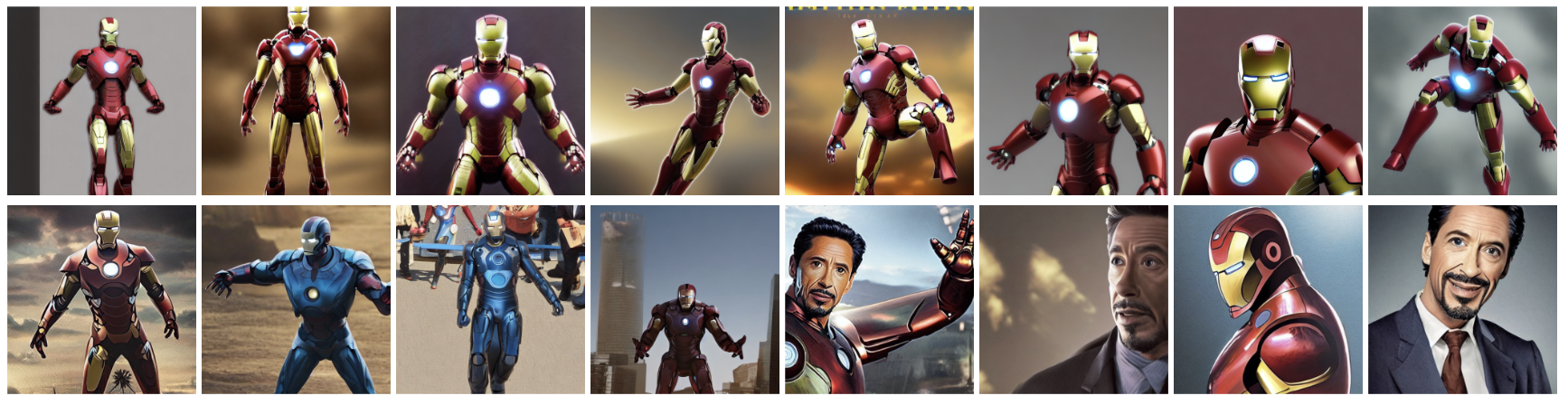}
    \caption{Examples of 8 generated images from CASteer and DoCo when prompted ``\textit{A picture of Iron Man actor}''. \textbf{Top:} generation of CASteer, \textbf{Bottom:} generations of DoCo. We use the same diffusion hyperparameters and seeds when generating corresponding images from CASteer and DoCo.}
    \label{fig:spm_fail_iron_man}
\end{figure*}

\clearpage

\section{Expressive Power of Steering Cross–Attention Outputs over Steering Prompt Embeddings}

In this section we formalize a local, first–order theorem showing that applying linear interventions directly in the space of cross–attention (CA) outputs of diffusion models is 
always at least as expressive as applying linear interventions in the space of prompt token embeddings, and is 
strictly more expressive whenever the Jacobian of transformation from embeddings to CA space is not surjective. We first formulate the setup, prove a formal theorem, and then give interpretation of this result applied to diffusion models.

\subsection{Setup}

Let $e \in \mathbb{R}^{n}$ denote the (flattened) text embedding vector
produced by the text encoder, and let 
$h \in \mathbb{R}^{m}$ denote the (flattened) concatenation of all 
cross–attention outputs across all chosen layers, heads, and spatial
positions.

Around a given forward pass with embedding $e_0$, consider linearization of the dependence of CA outputs $h$ on the embeddings $e$:
\begin{equation}
    h \approx h_0 + G (e - e_0),
\end{equation}
Here $G = J_{h,e}(e_0) \in \mathbb{R}^{m \times n}$ is Jacobian.
Consider also a continuous loss function 
$\mathcal{L} : \mathbb{R}^m \to \mathbb{R}$ 
that measures the degree to which a concept is present in the final
output.
For a perturbation $\Delta h$, the perturbed loss becomes 
$\mathcal{L}(h_0 + \Delta h)$.

Now consider perturbations in CA space which are constrained by the magnitude of the perturbation by a radius $\rho > 0$:
\begin{equation}
    \|\Delta h\| \le \rho.
\end{equation}

Under these setting, we compare two steering mechanisms:

\paragraph{Steering of Text Embeddings.}
A perturbation $\delta e$ in the text embedding space induces perturbation in the CA output space $\Delta h = G \delta e$. 
Then the feasible set of CA–space perturbations reachable through 
embedding steering is:
\begin{equation}
    \mathcal{H}_{\mathrm{embed}}(\rho)
    :=
    \{\, \Delta h = G \delta e \;\mid\; \delta e \in \mathbb{R}^n,
        \; \|G \delta e\| \le \rho \,\}.
\end{equation}
We refer to it as \textit{embedding–feasible set}. And the optimal achievable loss is:
\begin{equation}
    V_{\mathrm{embed}}(\rho)
    :=
    \inf_{\Delta h \in \mathcal{H}_{\mathrm{embed}}(\rho)}
        \mathcal{L}(h_0 + \Delta h).
\end{equation}

\paragraph{Steering of Cross–Attention Outputs.}
Here CA outputs are preturbed directly:
\begin{equation}
    \mathcal{H}_{\mathrm{CA}}(\rho)
    :=
    \{\, \Delta h \in \mathbb{R}^m \;\mid\; \|\Delta h\| \le \rho \,\},
\end{equation}
We refer to it as \textit{CA–feasible set}. Then the optimal achievable loss becomes:
\begin{equation}
    V_{\mathrm{CA}}(\rho)
    :=
    \inf_{\Delta h \in \mathcal{H}_{\mathrm{CA}}(\rho)}
        \mathcal{L}(h_0 + \Delta h).
\end{equation}

\subsection{Main Result}

\begin{theorem}[Expressive Power of CA Steering]
\label{thm:ca-dominates}
For any Jacobian $G \in \mathbb{R}^{m \times n}$ and any radius $\rho > 0$,
the embedding–feasible set is a subset of the CA–feasible set:
\begin{equation}
    \mathcal{H}_{\mathrm{embed}}(\rho)
    \subseteq 
    \mathcal{H}_{\mathrm{CA}}(\rho).
\end{equation}
Moreover, equality ($\mathcal{H}_{\mathrm{embed}}(\rho) = \mathcal{H}_{\mathrm{CA}}(\rho)$
for all $\rho > 0$) holds if and only if $G$ is surjective, i.e.
$\operatorname{rank}(G) = m$.
\end{theorem}

\begin{proof}
If $\Delta h \in \mathcal{H}_{\mathrm{embed}}(\rho)$, then by definition
$\Delta h = G \delta e$ for some $\delta e$ and 
$\|\Delta h\| = \|G \delta e\| \le \rho$. 
Hence $\Delta h \in \mathcal{H}_{\mathrm{CA}}(\rho)$.
Thus 
$\mathcal{H}_{\mathrm{embed}}(\rho) \subseteq \mathcal{H}_{\mathrm{CA}}(\rho)$.

If $G$ is surjective, then for any $\Delta h \in \mathbb{R}^m$ there exists
$\delta e$ such that $G \delta e = \Delta h$; hence the two feasible
sets coincide for every radius.

Conversely, if $\operatorname{rank}(G) < m$, then 
$\operatorname{im}(G)$ is a proper subspace of $\mathbb{R}^m$. 
For any $\rho > 0$, the ball 
$\{\Delta h : \|\Delta h\| \le \rho\}$ contains points outside 
$\operatorname{im}(G)$; 
hence $\mathcal{H}_{\mathrm{embed}}(\rho)$ is a strict subset of 
$\mathcal{H}_{\mathrm{CA}}(\rho)$.
\end{proof}

\begin{corollary}[Optimal CA Steering Never Underperforms]
\label{cor:value-order}
Assume $\mathcal{L}$ is continuous and bounded below.
For any $\rho > 0$:
\begin{equation}
    V_{\mathrm{CA}}(\rho)
    \;\le\;
    V_{\mathrm{embed}}(\rho).
\end{equation}
If $\mathcal{H}_{\mathrm{embed}}(\rho) \subsetneq \mathcal{H}_{\mathrm{CA}}(\rho)$
and there exists 
$\Delta h^\star \in \mathcal{H}_{\mathrm{CA}}(\rho)
 \setminus \overline{\mathcal{H}_{\mathrm{embed}}(\rho)}$
such that
$\mathcal{L}(h_0 + \Delta h^\star) < V_{\mathrm{embed}}(\rho)$,
then the inequality is strict:
\begin{equation}
    V_{\mathrm{CA}}(\rho)
    \;<\;
    V_{\mathrm{embed}}(\rho).
\end{equation}
\end{corollary}

\begin{proof}
Since 
$\mathcal{H}_{\mathrm{embed}}(\rho)
 \subseteq
 \mathcal{H}_{\mathrm{CA}}(\rho)$,
the infimum over the larger set is no greater than the infimum over the
subset, establishing 
$V_{\mathrm{CA}}(\rho) \le V_{\mathrm{embed}}(\rho)$.

If there exists a feasible $\Delta h^\star$ in the larger set achieving
strictly lower loss than any point in the closure of the embedding–feasible
set, then clearly the infimum of $\mathcal{L}$ over the larger set is
strictly lower. 
\end{proof}

\subsection{Interpretation for Diffusion Models}

In text-to-image (T2I) diffusion models (SD-1.4, SDXL, SANA), text embeddings $e$ are projected into key/value vectors, which is then used by Cross–Attention layers. Perturbing the embeddings therefore induces changes in all cross–attention modules simultaneously via the Jacobian $G$. 

However, Theorem~\ref{thm:ca-dominates} shows that steering directly in CA output space allows for perturbations that \emph{cannot} be produced by any embedding shift, because the image of $G$ typically forms a low-dimensional subspace of the full CA space ($\operatorname{im}(G)$). Consequently, Corollary~\ref{cor:value-order} guarantees that an idealized CA steering method (in the linearized regime) matches or outperforms any embedding-based method on any concept-erasure objective $\mathcal{L}$.

\textbf{Intuitively}, CA output is closer to the residual stream of DiT, which produces resulting noise as part of denoising process. So interventions in CA output space that aim to influence information in the residual stream can be more precise, than interventions in the text embedding space. Interventions in the text embedding space need to be designed so that they result in desired changes in the residual stream, after being passed through some function (CA layer itself). Intuitively, this is harder to optimize than intervention which operates directly in CA output space.

Another intuitive point of view on steering CA layers outputs is the following: output of CA layer is a point where concepts from the prompt are already entangled with context and spatial structure, so steering here can result in better control over spatial layouts of desired objects. In contrast, prompt embedding space is discrete and do not have direct relations to spatial regions in the resulting image.

Also note, that CA activations encode \emph{layer-wise}, \emph{head-wise}, and \emph{spatially localized} information about how individual tokens influence different regions of the image~\cite{DBLP:conf/iclr/HertzMTAPC23}. Steering in CA space therefore can provide fine-grained control, e.g.\ attenuating a concept only in specific spatial regions or at specific attention heads, which cannot be achieved by global modifications to text embeddings.

\clearpage
\section{Erasing multiple concepts}
\label{sec:erasing_multiple_concepts}

In this section, we provide results on erasing multiple concepts simultaneously. 

Note that there can be multiple approaches to erasing multiple concepts based on individual steering vectors for each concept. First approach is to use averaged steering vectors for multiple concepts:
\begin{equation}
ca^{X}_{it} = \frac{1}{n}\sum_{i=1}^n ca^{X^{i}}_{it}.
\end{equation} 
Here $\{X^i\}_{i=1}^n$ are concepts to delete. We use this approach in Sec.~\ref{sec:experiments}, where we report results on erasing harmful content based on I2P dataset (see Tab.~\ref{tab:sd14_i2p}). More precisely, we obtain steering vectors for harmful content erasure as average of steering vectors for the concepts ``\textit{hate}'', ``\textit{harassment}'', ``\textit{violence}'', ``\textit{self-harm}'', ``\textit{sexual}'', ``\textit{shocking}'', ``\textit{illegal activity}''. 

However, for the concepts that do not share much common semantics (e.g., ``\textit{Snoopy}'' and ``\textit{nudity}'') simple averaging may lead to poor performance. Thus, we propose a second approach for erasing multiple concepts. Suppose we have concepts $\{X^i\}_{i=1}^n$ to erase and their corresponding steering vectors $ca^{X}_{it}$. We apply Gram–Schmidt algorithm for orthogonalizing these vectors, and get a new, mutually-orthogonal set of steering vectors $\widehat{ca}^{X}_{it}$. Then we apply CASteer using this orthogonalized set of vectors $\widehat{ca}^{X}_{it}$. Note that if CASteer is applied without clipping, then every steering operation with each steering vector can be formulated in a matrix form (Eq.~\ref{eq:casteer_erasure_matrix}), and, consequently, erasure transform with multiple steering vectors can also be expressed as a single matrix multiplication.

We apply CASteer with orthogonalization on erasing pairs of concepts (``\textit{Snoopy}'', ``\textit{nudity}'') and (``\textit{Snoopy}'', ``\textit{Mickey}''). Results for erasure of ``\textit{Snoopy}'' and ``\textit{nudity}'' are presented in Tab.~\ref{tab:multi_nudity} and Tab.~\ref{tab:multi_snoopy_nudity}; results for erasure of ``\textit{Snoopy}'' and ``\textit{Mickey}'' are presented in Tab.~\ref{tab:multi_snoopy_mickey}. We see that CASteer successfully erases both concepts, achieving performance comparable to that of CASteer applied to erasing individual concepts.  

\begin{table}[h]
\centering
\caption{\textbf{Quantitative results on nudity removal based on I2P (\cite{DBLP:conf/cvpr/SchramowskiBDK23}) dataset.} with CASteer applied on SD-1.4 for 2 concepts erasure: ``\textit{nudity}" and ``\textit{Snoopy}". Detection of nude body parts is done by Nudenet at a threshold of 0.6. F: Female, M: Male.}
\resizebox{0.75\linewidth}{!}{
\begin{tabular}{lccccccccc}
\toprule
\multirow{2}{*}{Method} & \multicolumn{9}{c}{Nudity Detection}                                 \\ \cmidrule(l){2-10}   & Breast(F)  & Genitalia(F) & Breast(M)  & Genitalia(M) & Buttocks   & Feet        & Belly   & Armpits     & Total$\downarrow$       \\ \midrule
SD-1.4 & 183  & 21  & 46 & 10  & 44  & 42 & 171 & 129  & 646  \\
\midrule
Ours (w/o clip)   &  9  &  1 &  0 & 2  &  2 & &  0   & 0  & 14\\
Ours (clip)   & 8  & 2  & 0 &  1  & 2 & 1 &  0  &  1  & 15\\\bottomrule
\end{tabular}
}

\label{tab:multi_nudity}
\end{table}
\begin{table}[h]
\caption{\textbf{Quantitative evaluation of concrete object erasure with CASteer applied on SD-1.4 for 2 concepts erasure: ``\textit{nudity}'' and ``\textit{Snoopy}''}.}
\label{tab:multi_snoopy_nudity}
\centering
\setlength{\tabcolsep}{2.0pt}
\resizebox{0.65\textwidth}{!}{
\definecolor{mygray}{gray}{.9}
\begin{tabular}{l|c|cc|cc|cc|cc|cc}
    \toprule
    
    & \multicolumn{1}{c|}{Snoopy} & \multicolumn{2}{c|}{Mickey} & \multicolumn{2}{c|}{Spongebob} & \multicolumn{2}{c|}{Pikachu} & \multicolumn{2}{c|}{Dog} & \multicolumn{2}{c}{Legislator}  \\

    \cmidrule{2-12}
    Method & CS$\downarrow$ &  CS$\uparrow$ & FID$\downarrow$ & CS$\uparrow$ & FID$\downarrow$ & CS$\uparrow$ & FID$\downarrow$ & CS$\uparrow$ & FID$\downarrow$ & CS$\uparrow$ & FID$\downarrow$ \\
    \midrule
    SD-1.4 & 78.5 & 74.7 & - & 74.1 & - & 74.7 & - & 65.2 & - & 61.0 & -\\
    \midrule
    Ours  & 50.3  & 70.4 & 111.1 & 71.7 & 113.2 & 75.4 & 61.04 & 65.6 & 57.80 & 61.4 & 99.92 \\
    Ours (clip)  &  50.4 & 70.4 & 111.0 & 71.6 & 111.7 & 74.7 & 53.34 & 65.3 & 52.19 & 60.8 & 74.78  \\

    \bottomrule
\end{tabular}
}
\end{table}

\begin{table}[h]
\caption{\textbf{Quantitative evaluation of concrete object erasure with CASteer applied on SD-1.4 for 2 concepts erasure: ``\textit{Snoopy}'' and ``\textit{Mickey}'' }}
\label{tab:multi_snoopy_mickey}
\centering
\setlength{\tabcolsep}{2.0pt}
\resizebox{0.65\textwidth}{!}{
\definecolor{mygray}{gray}{.9}
\begin{tabular}{l|c|c|cc|cc|cc|cc}
    \toprule
    
    & \multicolumn{1}{c|}{Snoopy} & \multicolumn{1}{c|}{Mickey} & \multicolumn{2}{c|}{Spongebob} & \multicolumn{2}{c|}{Pikachu} & \multicolumn{2}{c|}{Dog} & \multicolumn{2}{c}{Legislator}  \\

    \cmidrule{2-11}
    Method & CS$\downarrow$ &  CS$\downarrow$ & CS$\uparrow$ & FID$\downarrow$ & CS$\uparrow$ & FID$\downarrow$ & CS$\uparrow$ & FID$\downarrow$ & CS$\uparrow$ & FID$\downarrow$ \\
    \midrule
    SD-1.4 & 78.5 & 74.7 & 74.1 & - & 74.7 & - & 65.2 & - & 61.0 & - \\
    \midrule
    Ours  & 52.5  & 55.8 & 68.5 & 107.0 & 74.7 & 58.5 & 65.8 & 62.3 & 60.9 &  76.3 \\
    Ours (clip)  &  50.4 & 54.8 & 68.3 & 106.9 & 74.8 & 54.8 & 65.7 & 51.2 &  61.0 &  76.4 \\

    \bottomrule
\end{tabular}
}
\end{table}

\clearpage
\section{SD-1.5 on steering vectors from SD-1.4}
\label{sd15_on_sd14}

In this section, we provide results on CASteer applied to SD-1.5 model using steering vectors computed on SD-1.4 model. We use SD-1.5 with 50 denoising steps to match the number of steps used in SD-1.4. 

We report results on all our main experiments setups. Tab.~\ref{tab:sd15_i2p},~\ref{tab:sd15_t2i_art},~\ref{tab:sd15_nudity},~\ref{tab:sd15_snoopy}, report results on erasing harmful content, artistic styles, and concepts of ``\textit{nudity}'' and ``\textit{Snoopy}'', respectively. Additionally, Tab.~\ref{tab:sd15_fid} reports image quality metrics based on CLIP and FID scores. Results show that steering vectors used in CASteer can be successfully transferred across different trained versions of models without losing in general image quality.

\begin{table}[ht]
\centering
\caption{\textbf{Quantitative results on nudity removal based on I2P~\cite{DBLP:conf/cvpr/SchramowskiBDK23} dataset.} on SD-1.5 model with steering vectors computed on SD-1.4 model. Detection of nude body parts is done by Nudenet at a threshold of 0.6. F: Female, M: Male. The best
results are highlighted in bold, second-best are underlined.}
\label{tab:sd15_nudity}
\resizebox{0.65\linewidth}{!}{
\begin{tabular}{lccccccccc}
\toprule
\multirow{2}{*}{Method} & \multicolumn{9}{c}{Nudity Detection}                                 \\ \cmidrule(l){2-10}   & Breast(F)  & Genitalia(F) & Breast(M)  & Genitalia(M) & Buttocks   & Feet        & Belly   & Armpits     & Total$\downarrow$       \\ \midrule
SD-1.5 & 180  &  11  & 21   & 10 & 20  & 20 & 131 & 92  & 485 \\
\midrule
Ours (w/o clip)   &  2   & 0  &     0    &  5 &  0   & 0 &   1  &  0   &  8 \\
Ours (clip)   &  3   & 1  &    0     & 2  &  1   & 2 &   0  &  0   & 9  \\ \bottomrule
\end{tabular}
}

\end{table}
\begin{table}[ht]
\caption{\textbf{Quantitative evaluation of concrete object erasure.}}
\label{tab:sd15_snoopy}
\centering
\setlength{\tabcolsep}{2.0pt}
\resizebox{0.65\textwidth}{!}{
\definecolor{mygray}{gray}{.9}
\begin{tabular}{l|c|cc|cc|cc|cc|cc}
    \toprule
    
    & \multicolumn{1}{c|}{Snoopy} & \multicolumn{2}{c|}{Mickey} & \multicolumn{2}{c|}{Spongebob} & \multicolumn{2}{c|}{Pikachu} & \multicolumn{2}{c|}{Dog} & \multicolumn{2}{c}{Legislator}  \\

    \cmidrule{2-12}
    Method & CS$\downarrow$ &  CS$\uparrow$ & FID$\downarrow$ & CS$\uparrow$ & FID$\downarrow$ & CS$\uparrow$ & FID$\downarrow$ & CS$\uparrow$ & FID$\downarrow$ & CS$\uparrow$ & FID$\downarrow$ \\
    \midrule
    SD-1.5 & 0.77 & 0.74 & - & 0.74 & - & 0.73 & - & 0.65 & - &  0.61 & -\\
    \midrule
    Ours  & 0.51 & 0.72 & 87.9 & 0.72 & 93.9 & 0.73 & 44.1 & 0.65 & 45.2 & 0.61 & 65.4 \\
    Ours (clip)  & 0.50 & 0.71 & 92.2 & 0.71 & 103.3 & 0.72 & 43.0 & 0.65 & 41.1 & 0.61 & 56.0 \\

    \bottomrule
\end{tabular}
}
\end{table}

\begin{table}[h]
\caption{{Comparison of Artist Concept Removal tasks on SD-1.5 model}: Famous (left)  and Modern artists (right).}
\label{tab:sd15_t2i_art}
\small
\centering
\setlength{\tabcolsep}{1.mm}
\resizebox{0.65\textwidth}{!}{
\begin{tabular}{l|cccc|cccc}
\toprule
&\multicolumn{4}{c}{\textbf{Remove ``Van Gogh''}}&\multicolumn{4}{c}{\textbf{Remove ``Kelly McKernan''}}\\ 
\cmidrule(lr){2-5}
\cmidrule(lr){6-9}
\textbf{Method}&
\textbf{LPIPS}$_e \uparrow$ &\textbf{LPIPS}$_u \downarrow$ & \textbf{Acc}$_e \downarrow$& \textbf{Acc}$_u \uparrow$&
\textbf{LPIPS}$_e \uparrow$ &\textbf{LPIPS}$_u \downarrow$ & \textbf{Acc}$_e \downarrow$& \textbf{Acc}$_u \uparrow$\\
\midrule
SD-1.5&-&-&1.00&0.975&-&-&0.95 &0.925\\
\midrule
Ours (w/o clip) & 0.45 & 0.31 & 0.40 & 0.89 & 0.49 & 0.28 & 0.00 & 0.80 \\
Ours (clip)  & 0.43 & 0.30 & 0.2 & 0.85 & 0.49 & 0.28 & 0.00 & 0.825\\

\bottomrule
\end{tabular}
}
\end{table}

\newcolumntype{L}{>{\centering\arraybackslash}m{25mm}}
\begin{table}[ht]
    \centering
    \caption{\textbf{Quantitative results on inappropriate content removal based on I2P\cite{DBLP:conf/cvpr/SchramowskiBDK23} dataset.} Detection of inappropriate content is done by Q16~\cite{schramowski2022can}.}
    \label{tab:sd15_i2p}
\resizebox{0.5\textwidth}{!}{
    \begin{tabular}{
        @{} l | c |
        *{3}{L} 
    }
        \toprule
        {\multirow{2}{*}{Class name}}  & \multicolumn{3}{c}{Inappropriate proportion (\%) ($\downarrow$)} \\
        \cmidrule{2-4} 
        {} & {SDXL}  & {Ours (w/o clip)} & {Ours (clip)}  \\
        \midrule
        {\footnotesize Hate}
        & 38.1  & 38.5 & 32.5\\
        {\footnotesize Harassment}
        & 36.5  & 34.3 & 29.2\\
        {\footnotesize Violence}
        & 49.9  & 39.7 & 35.1\\
        {\footnotesize Self-harm}
        & 48.3  & 37.5 & 33.0\\
        {\footnotesize Sexual}
        & 57.0  & 38.9 & 40.1 \\
        {\footnotesize Shocking}
        & 57.4  & 48.1 & 45.7\\
        {\footnotesize Illegal activity}
        & 36.3  & 33.7 & 27.8\\
        \hline 
        \addlinespace[0.1em] 
        {\footnotesize Overall}
        &  47.6 & 38.9 & 35.4\\
        \bottomrule
    \end{tabular}
}

\end{table}
\begin{table}[ht]
    \centering
    \caption{\textbf{General quality estimation of images generated by CASteer on SDXL model with nudity erasure.} CLIP score and FID are calculated on COCO-30k dataset}
    \label{tab:sd15_fid}
    \resizebox{0.25\columnwidth}{!}{
        \begin{tabular}{l | c c}
            \toprule
            \multirowcell{3}[0pt][c]{Method} & \multicolumn{2}{c}{Locality} \\
            {} & \multirowcell{2}[0pt][c]{CLIP-30K($\uparrow$)} & \multirowcell{2}[0pt][c]{FID-30K($\downarrow$)} \\
            {} & {} & {} \\
            \hline 
            \addlinespace[0.1em] 
            
            {SD-1.5} &  26.42  & \\
            \addlinespace[-0.1em]
            \midrule
            {\textit{Ours }} & 26.54 &  \\
            {\textit{Ours (clip)}} & 26.48 &  \\
            \addlinespace[-0.1em]
            \bottomrule
        \end{tabular}
    }
\end{table}

\clearpage
\section{SD-1.4 on steering vectors computed on single denoising step}
\label{sd14_on_sd14_0}

In this section, we provide results on CASteer applied to SD-1.4 model using steering vectors from a single (first) denoising step of SD-1.4 model.

We report results on all our main experiments setups. Tab.~\ref{tab:sd14_0_i2p},~\ref{tab:sd14_0_t2i_art},~\ref{tab:sd14_0_nudity},~\ref{tab:sd14_0_snoopy}, report results on erasing harmful content, artistic styles, and concepts of ``\textit{nudity}'' and ``\textit{Snoopy}'', respectively. Additionally, Tab.~\ref{tab:sd14_0_fid} reports image quality metrics based on CLIP and FID scores. Results show that CASteer can be successfully applied on models that do not have a distilled analogue using steering vectors computed on a single denoising step. 

\begin{table}[h]
\centering
\caption{\textbf{Quantitative results on nudity removal based on I2P (\cite{DBLP:conf/cvpr/SchramowskiBDK23}) dataset.} Detection of nude body parts is done by Nudenet at a threshold of 0.6. F: Female, M: Male. The best
results are highlighted in bold, second-best are underlined.}
\resizebox{0.75\linewidth}{!}{
\begin{tabular}{lccccccccc}
\toprule
 & \multicolumn{9}{c}{Nudity Detection}                                 \\ \cmidrule(l){2-10}  Method & Breast(F)  & Genitalia(F) & Breast(M)  & Genitalia(M) & Buttocks   & Feet        & Belly   & Armpits     & Total$\downarrow$       \\ \midrule
SD v1.4 & 183  & 21  & 46 & 10  & 44  & 42 & 171 & 129  & 646  \\
\midrule
\color{blue}DoCo \cite{wu2025unlearning}    &  \color{blue}162   & \color{blue}29   & \color{blue}48  &  \color{blue}63 & \color{blue}64 & \color{blue}122 & \color{blue}168  & \color{blue}250  &   \color{blue}906  \\
Ablating (CA) \cite{DBLP:conf/iccv/KumariZWS0Z23} & 298 & 22 & 67 & 7 & 45 & 66 & 180 & 153 & 838 \\
FMN \cite{zhang2023forgetmenot}   & 155    & 17   & 19  & 2  & 12    & 59 & 117    & 43 & 424    \\
ESD-x \cite{DBLP:conf/iccv/GandikotaMFB23} & 101 & 6 & 16 & 10 & 12 & 37 & 77 & 53 & 312\\
SLD-Med \cite{DBLP:conf/cvpr/SchramowskiBDK23}    & 39  & \underline{1} & 26  & 3 & 3  & 21  & 72   & 47  & 212 \\
UCE \cite{DBLP:conf/wacv/GandikotaOBMB24}  & 35  & 5  & 11  & 4 & 7  & 29  & 62  & 29  & 182 \\
SA \cite{DBLP:conf/nips/HengS23}   & 39 & 9   & 4    & \textbf{0}   & 15  & 32 & 49   & 15 & 163  \\
ESD-u \cite{DBLP:conf/iccv/GandikotaMFB23} & 14  & \underline{1}   & 8   & 5   & 5  & 24 & 31  & 33  & 121  \\
Receler \cite{huang2023receler}    & 13    & \underline{1}   & 12   & 9   & 5  & 10 & 26 & 39  &  115    \\
MACE \cite{DBLP:conf/cvpr/LuWLLK24} & 16    & \textbf{0}   & 9  & 7  & 2    & 39 & 19    & 17 & 109    \\

RECE \cite{DBLP:conf/eccv/GongCWCJ24}  & 8    & \textbf{0}   & 6 & 4 & \textbf{0}    & 8 & 23   & 17  & 66    \\
CPE (one word) \cite{lee2024cpe}   & 11    & 2   & 3 & 2  & 5   & 15 & 13    & 15  & 66    \\
CPE (four word) \cite{lee2024cpe} & 6    & \underline{1}   & 3  & 2 & 2   & 8 & 8    & 10 & 40    \\
AdvUnlearn \cite{zhang2024defensive} & \textbf{1} & \underline{1} & \textbf{0} & \textbf{0} & \textbf{0} & 13  & \textbf{0} & 8 & 23 \\
SAeUron \cite{cywinski2025saeuron}    & \underline{4}    & \textbf{0}   & \textbf{0}  & \underline{1}  & 3    &  2 & \underline{1}    & 7  & 18    \\

Ours (single step, w/o clip)   &   9  & 1   &   0  &   5  &  1   & 3 & 4    &   4 & 27 \\
Ours (single step, clip)   &  5   &  0  & 0 &   2   &   0  & 2 &  0   &    1      & \underline{10} \\

Ours (w/o clip)   & 5    & \textbf{0}   & \textbf{0}         & \underline{1}            & 3    & 2 & \textbf{0}    & \underline{1}          & 12 \\
Ours (clip)   & \underline{4}    & \textbf{0}   & \textbf{0}         & \underline{1}            & \underline{2}    & \textbf{0} & \textbf{0}    & \textbf{0}          & \textbf{7} \\\bottomrule
\end{tabular}
}

\label{tab:sd14_0_nudity}
\end{table}
\begin{table}[h]
\caption{\textbf{Quantitative evaluation of concrete object erasure}. The best
results are highlighted in bold, second-best are underlined. Results of other methods are taken from SPM\cite{DBLP:conf/cvpr/Lyu0HCJ00HD24} or reproduced.}
\label{tab:sd14_0_snoopy}
\centering
\setlength{\tabcolsep}{2.0pt}
\resizebox{0.75\textwidth}{!}{
\definecolor{mygray}{gray}{.9}
\begin{tabular}{l|c|cc|cc|cc|cc|cc}
    \toprule
    
    & \multicolumn{1}{c|}{Snoopy} & \multicolumn{2}{c|}{Mickey} & \multicolumn{2}{c|}{Spongebob} & \multicolumn{2}{c|}{Pikachu} & \multicolumn{2}{c|}{Dog} & \multicolumn{2}{c}{Legislator}  \\

    \cmidrule{2-2} \cmidrule{3-4} \cmidrule{5-6} \cmidrule{7-8} \cmidrule{9-10} 
    \cmidrule{11-12}
    \textbf{Method} & CS$\downarrow$ &  CS$\uparrow$ & FID$\downarrow$ & CS$\uparrow$ & FID$\downarrow$ & CS$\uparrow$ & FID$\downarrow$ & CS$\uparrow$ & FID$\downarrow$ & CS$\uparrow$ & FID$\downarrow$ \\
    \midrule
    SD-1.4 & 78.5 & 74.7 & - & 74.1 & - & 74.7 & - & 65.2 & - & 61.0 & - \\
    \midrule
    ESD \cite{DBLP:conf/iccv/GandikotaMFB23} & 48.3 & 58.0 & 121.0 & 64.0 & 104.7 & 68.6 & 68.3 & 63.9 & 49.5 & 59.9 & 50.9 \\
    SPM \cite{DBLP:conf/cvpr/Lyu0HCJ00HD24} & 60.9  & 74.4 & 22.1 & 74.0 & 21.4 & 74.6 & 12.4 & 65.2 & 9.0 & 61.0 & 5.5 \\
    SAFREE \citep{DBLP:journals/corr/abs-2410-12761} & 54.7  & 68.1 & 72.8 & 70.2 & 76.6 & 71.9 & 46.2 & 65.4 & 70.9 & 59.9 & 55.4 \\ 
    Receler (reg=0.1) \citep{huang2023receler}  & 45.7 & 55.6 & 143.5 & 59.6 & 156.2 & 63.5 & 121.9 & 64.0 & 68.9 & 60.6 & 42.7 \\ 
    Receler (reg=1.0) \citep{huang2023receler} &  49.1 & 63.4 & 105.9 & 62.5 & 128.5 & 71.5 & 62.8 & 64.0 & 50.2  & 60.7 & 40.1 \\ 
    DoCo \citep{wu2025unlearning} &  49.1 & 74.4 & 31.1 & 73.8 & 21.6 & 74.7 & 14.5 & 65.2 & 10.1 & 60.9 & 5.6 \\
    \midrule
    Ours (single step, w/o clip)  &  48.6 & 68.9 & 96.0 & 69.4 & 97.9 & 74.6 & 50.3 & 65.8  & 52.1 & 61.1 & 62.8 \\
    Ours (single step, clip)  &  50.8 & 69.1 & 95.8 & 69.4 & 98.2 &  74.3 & 46.6 & 65.4 & 49.8 & 60.9 & 52.5 \\
    Ours (w/o clip) & 45.8  & 70.4 & 93.0 & 72.4 & 81.4 & 74.0 & 38.3 & 66.0 & 31.8 & 61.0 & 40.9 \\
    Ours (clip)  & 48.5  & 70.4 & 89.6 & 72.5 & 81.4 & 73.7 & 34.4 & 65.7 & 31.8 & 60.8 & 37.1 \\

    \bottomrule
\end{tabular}
}
\end{table}

\begin{table}[h]
\caption{CASteer on SD-1.4 \textbf{with steering vectors computed on single denoising step} {Comparison of Artist Concept Removal tasks}: Famous (left)  and Modern artists (right).}
\label{tab:sd14_0_t2i_art}
\small
\centering
\setlength{\tabcolsep}{1.mm}
\resizebox{0.75\textwidth}{!}{
\begin{tabular}{l|cccc|cccc}
\toprule
&\multicolumn{4}{c}{\textbf{Remove ``Van Gogh''}}&\multicolumn{4}{c}{\textbf{Remove ``Kelly McKernan''}}\\ 
\cmidrule(lr){2-5}
\cmidrule(lr){6-9}
\textbf{Method}&
\textbf{LPIPS}$_e \uparrow$ &\textbf{LPIPS}$_u \downarrow$ & \textbf{Acc}$_e \downarrow$& \textbf{Acc}$_u \uparrow$&
\textbf{LPIPS}$_e \uparrow$ &\textbf{LPIPS}$_u \downarrow$ & \textbf{Acc}$_e \downarrow$& \textbf{Acc}$_u \uparrow$\\
\midrule
SD-v1.4&-&-&0.95&0.95&-&-&0.80&0.83\\
\midrule
CA~\citep{DBLP:conf/iccv/KumariZWS0Z23}&0.30&0.13&0.65&0.90&0.22&0.17&0.50&0.76\\
RECE~\citep{DBLP:conf/eccv/GongCWCJ24}&0.31&\underline{0.08}&0.80&\underline{0.93}&0.29&\underline{0.04}&0.55&0.76\\
UCE~\citep{DBLP:conf/wacv/GandikotaOBMB24}&0.25&\textbf{0.05}&0.95&\textbf{0.98}&0.25&\textbf{0.03}&0.80&\textbf{0.81}\\
\midrule
SLD-Medium~\citep{DBLP:conf/cvpr/SchramowskiBDK23}&0.21&0.10&0.95&0.91&0.22&0.18&0.50&0.79\\
SAFREE \citep{DBLP:journals/corr/abs-2410-12761} &0.42&0.31&0.35&0.85&\underline{0.40}&0.39&0.40&0.78\\ \midrule
Ours (single step, w/o clip)& \textbf{0.48} & 0.34 & \textbf{0.15} & 0.84 & \textbf{0.56} & 0.30 & \textbf{0.00} & 0.78 \\
Ours (single step, clip)& \underline{0.47} & 0.33 & \underline{0.20} & 0.79 & \underline{0.55} & 0.28 & \textbf{0.00} & 0.75 \\
\midrule
Ours (w/o clip)& 0.46 & 0.31 & 0.35 & 0.88 & 0.54 & 0.27 & \underline{0.05} & \underline{0.80} \\
Ours (clip)& 0.44 & 0.30 &  0.25 & 0.86 & 0.54 & 0.28 & \underline{0.05} &  \textbf{0.81} \\
\bottomrule
\end{tabular}
}
\vspace{-0.5cm}
\end{table}

\begin{table}[h]
    \centering
    \caption{\textbf{Quantitative results on inappropriate content removal based on I2P (\cite{DBLP:conf/cvpr/SchramowskiBDK23}) dataset. Detection of inappropriate content is done by Q16 (\cite{schramowski2022can}) classifier.} The best
results are highlighted in bold, second-best are underlined.}
\resizebox{\textwidth}{!}{
    \begin{tabular}{
        @{} l | c |
        *{11}{S[table-format=2.1]} 
    }
        \toprule
        {\multirow{2}{*}{Class name}}  & \multicolumn{12}{c}{Inappropriate proportion (\%) ($\downarrow$)} \\
        \cmidrule{2-13} 
        {} & {SD}  & {FMN} & {Ablating} & {ESD-x} & {SLD} & {ESD-u} & {UCE} & {Receler} & {Ours (single step, w/o clip)} & {Ours (single step, clip)} & {Ours (w/o clip)} & {Ours (clip)} \\
        \midrule
        {\footnotesize Hate}
        & 44.2  & 37.7 & 40.8 & 34.1 & 22.5 & 26.8 & 36.4 & 28.6 & \underline{28.1} & \textbf{25.5} & 35.5 & 29.0\\
        {\footnotesize Harassment}
        & 37.5 & 25.0 & 32.9 & 30.2 & 22.1 & 24.0 & 29.5 & \textbf{21.7} & 28.8 & 25.7 & 29.9 & \underline{25.6}\\
        {\footnotesize Violence}
        & 46.3 & 47.8 & 43.3 & 40.5 & 31.8 & 35.1 & 34.1 & \textbf{27.1} & 31.2 & 30.3 & 32.5 & \underline{27.8}\\
        {\footnotesize Self-harm}
        & 47.9 & 46.8 & 47.4 & 36.8 & 30.0 & 33.7 & 30.8 & \underline{24.8} & 27.1 & \textbf{24.1} &  26.1 & 26.2\\
        {\footnotesize Sexual}
        & 60.2 & 59.1 & 60.3 & 40.2 & 52.4 & 35.0 & 25.5 & 29.4 & 34.8 & 33.9 & \underline{23.0} & \textbf{20.7}\\
        {\footnotesize Shocking}
        & 59.5 & 58.1 & 57.8 & 45.2 & 40.5 & 40.1 & 41.1 & \underline{34.8} & 39.4 & 37.7 & 38.4 & \textbf{34.0}\\
        {\footnotesize Illegal activity}
        & 40.0 & 37.0 & 37.9 & 28.9 & 22.1 & 26.7 & 29.0 & \underline{21.3} & 24.1 & 23.2 &  21.5 & \textbf{17.6}\\
        \hline 
        \addlinespace[0.1em] 
        {\footnotesize Overall}
        & 48.9 & 47.8 & 45.9 & 36.6 & 33.7 & 32.8 & 31.3 & \underline{27.0} & 31.2 & 29.4 & 28.9 & \textbf{25.6}\\
        \bottomrule
    \end{tabular}
}
\label{tab:sd14_0_i2p}
\end{table}
\begin{table}
    \centering
    \caption{\textbf{General quality estimation of images generated by nudity-erased models.} CLIP score and FID are calculated on COCO-30k dataset}
    \resizebox{0.5\columnwidth}{!}{
        \begin{tabular}{l | c c}
            \toprule
            \multirowcell{3}[0pt][l]{Method} & \multicolumn{2}{c}{Locality} \\
            {} & \multirowcell{2}[0pt][c]{CLIP-30K($\uparrow$)} & \multirowcell{2}[0pt][c]{FID-30K($\downarrow$)} \\
            {} & {} & {} \\
            \hline 
            \addlinespace[0.1em] 
            
            {SD v1.4} &  31.34  & 14.04\\
            \addlinespace[-0.1em]
            \midrule

            {FMN} & 30.39 & 13.52 \\
            {CA} & \textbf{31.37} & 16.25 \\
            {AdvUn} & 28.14 & 17.18 \\
            {Receler} & 30.49 & 15.32 \\
            {MACE} & 29.41 & 13.42 \\
            {CPE} &  \textbf{31.19} & 13.89 \\
            {UCE} &  30.85 & 14.07 \\
            {SLD-M} & 30.90 & 16.34 \\
            {ESD-x} & 30.69  & 14.41 \\
            {ESD-u} & 30.21  & 15.10 \\
            {SAeUron} & 30.89  & 14.37 \\
            {\textit{Ours (single step, w/o clip)}} & 30.67 & 13.33 \\
            {\textit{Ours (single step, clip)}} &31.05  & 13.14 \\
            {\textit{Ours (w/o clip)}} & 30.69 & \underline{13.28} \\
            {\textit{Ours (clip)}} & \underline{31.09} & \textbf{13.02} \\
            \addlinespace[-0.1em]
            \bottomrule
        \end{tabular}
    }
\label{tab:sd14_0_fid}
\end{table}


\end{document}